\documentclass[12pt,onecolumn]{IEEEtran}
\usepackage{multicol}
\usepackage{amsmath,epsfig,comment}
\usepackage{amsthm}
\usepackage{tcolorbox}
\usepackage{amssymb}
\usepackage{multirow}
\usepackage{mathrsfs}
\usepackage{mathrsfs}
\usepackage{amsmath}
\usepackage{cases} 
\usepackage{amssymb,amsmath,cite}
\usepackage{epsfig}
\usepackage{color}
\usepackage{bm}
\usepackage{graphicx,subfigure}
\usepackage{algorithm}
\usepackage{algorithmic}
\usepackage{multirow}
\usepackage{soul}

\usepackage{extarrows}

\newcommand{\jjm}[1]{{\color{purple}#1}}

\def\s{\mathbf{s}}

\def\br{\boldsymbol{r}}

\def\n{\mathbf{n}}

\def\br{\mathbf{r}}
\def\bm{\boldsymbol{m}}

\def\bq{\boldsymbol{q}}

\def\R{\mathbb{R}}

\def\H{\mathbf{H}}

\def\y{\mathbf{y}}

\def\x{\mathbf{x}}
\def\s{\mathbf{s}}

\def\h{\mathbf{h}}

\def\bH{\mathbf{H}}

\def\bw{\boldsymbol{w}}
\def\bz{\mathbf{z}}
\def\bb{\mathbf{b}}

\def\w{\mathbf{w}}
\def\z{\mathbf{z}}

\newtheorem{theorem}{Theorem}

\newtheorem{lemma}{Lemma}
\newtheorem{remark}{Remark}

\newtheorem{definition}{Definition}
\newtheorem{proposition}{Proposition}

\newtheorem{corollary}{Corollary}
\newtheorem{assumption}{Assumption}

\makeatletter
\newcommand{\appendixlemma}{%
  \renewcommand{\thelemma}{\Alph{section}.\arabic{lemma}}%
  \setcounter{lemma}{0}%
  \@addtoreset{lemma}{section}  
}

\newcommand{\appendixpro}{%
  \renewcommand{\theproposition}{\Alph{section}.\arabic{proposition}}  
  \setcounter{proposition}{0}  
    \@addtoreset{proposition}{section} 
}
\newcommand{\appendixdef}{%
  \renewcommand{\thedefinition}{\Alph{section}.\arabic{definition}}  
  \setcounter{definition}{0}  
    \@addtoreset{definition}{section} 
}
\makeatother

\newtheoremstyle{noparens}%
  {}{}%
  {\itshape}{}%
  {\bfseries}{.}%
  { }%
  {\thmname{#1}\thmnumber{ #2}\mdseries\thmnote{ #3}}

\theoremstyle{noparens}

\title{Asymptotic  Analysis of Nonlinear One-Bit Precoding in Massive MIMO Systems via Approximate Message Passing}
	\author{\IEEEauthorblockN{Zheyu Wu, Junjie Ma, Ya-Feng Liu, and Bruno Clerckx,}

	\thanks{

	Z. Wu and and B. Clerckx are with the Department of Electrical and Electronic Engineering, Imperial College London, London, SW7 2AZ, U.K. (email: \{zheyu.wu, b.clerckx\}@imperial.ac.uk).  Junjie Ma is with the State Key Laboratory of Mathematical Sciences, Institute of Computational Mathematics and Scientific/Engineering Computing, Academy of Mathematics and Systems Science, Chinese Academy of Sciences, Beijing 100190, China (e-mail: majunjie@lsec.cc.ac.cn). 
Y.-F. Liu is with the Ministry of Education Key Laboratory of Mathematics and Information Networks, School of Mathematical Sciences, Beijing University of Posts and Telecommunications, Beijing 102206, China (email: yafengliu@bupt.edu.cn).
	
}
  }
\begin{document}
\maketitle
\begin{abstract}

Massive multiple-input multiple-output (MIMO) systems employing one-bit digital-to-analog converters offer a hardware-efficient solution for wireless communications. However, the one-bit constraint poses significant challenges for precoding design, as it transforms the problem into a discrete and nonconvex optimization task. In this paper, we investigate a widely adopted ``convex-relaxation-then-quantization" approach for nonlinear symbol-level one-bit precoding. Specifically, we first solve a convex relaxation of the discrete minimum mean square error precoding problem, and then quantize the solution to satisfy the one-bit constraint. Focusing on a real-valued system with an independently and identically distributed (i.i.d.)  Gaussian channel, we develop a novel analytical framework based on approximate message passing (AMP) to characterize the high-dimensional asymptotic performance of the  considered scheme. The key technical ingredient is an auxiliary AMP iteration that dedicatedly incorporates the nonlinear quantization function into the state evolution analysis. With the proposed framework, we derive a closed-form expression for the symbol error probability (SEP) at the receiver side in the large-system limit,  which provides a quantitative characterization of how model and system parameters affect the SEP performance. Our empirical results suggest that the $\ell_\infty^2$ regularizer, when paired with an optimally chosen regularization parameter, achieves optimal SEP performance within a broad class of convex regularization functions. As a first step towards a theoretical justification, we prove the optimality of the $\ell_\infty^2$ regularizer within the mixed $\ell_\infty^2$-$\ell_2^2$ regularization functions.
 
\end{abstract}


\section{Introduction}
Massive multiple-input multiple-output (MIMO) is a key enabling technology to meet the growing demands for spectral and energy efficiency in both 5G and future 6G wireless  communication systems \cite{massivemimo1}.  While massive MIMO offers significant performance gains, its practical implementation faces critical challenges due to the prohibitive hardware costs and energy consumption. In particular, digital-to-analog converters (DACs) and power amplifiers (PAs), which are the most power-hungry and costly components at the base station (BS), scale linearly with the number of antennas \cite{book:DAC,PAbook}. A promising and widely adopted solution is to replace high-resolution DACs with one-bit DACs, which not only minimizes cost and energy consumption  associated with DACs,  but also enables PAs to operate at peak efficiency thanks to the constant-envelope nature of one-bit signals. These advantages have spurred considerable research interest in one-bit precoding techniques \cite{MFrate,ZF,MMSE2,WSR,duality_onebit,SEPlinear,SQUID,CIfirst,CImodel,sep2,GEMM,PBB,onebit_CI,liu2024survey}. 

One-bit precoding involves designing transmit signals at the BS with each element constrained to one-bit resolution. Existing one-bit precoding schemes can be broadly categorized into two classes: linear-quantized precoding and nonlinear (symbol-level) precoding. Linear-quantized precoding applies a linear transformation to the data symbols, followed by quantization to meet the one-bit constraint~\cite{MFrate,ZF,MMSE2,WSR,duality_onebit,SEPlinear}. While this approach is computationally efficient, its performance is often limited by the coarse quantization. In contrast, nonlinear precoding designs the one-bit transmit signals on a symbol-by-symbol basis by solving dedicated optimization problems~\cite{SQUID,CIfirst,CImodel,sep2,GEMM,PBB,onebit_CI,liu2024survey}. Although nonlinear precoding generally offers better performance, it incurs higher computational complexity due to the need for symbol-rate processing and the difficulty of solving discrete optimization problems.

Despite extensive algorithmic developments in nonlinear one-bit precoding, analytical performance characterizations remain relatively scarce. This paper aims to bridge this gap by providing a rigorous asymptotic performance analysis of a class of nonlinear one-bit precoding schemes.

\subsection{Related Works}

The performance analysis of conventional linear precoding schemes, such as matched filter, zero-forcing, and regularized zero-forcing, in unquantized systems has been extensively studied~\cite{unquantized1,unquantized3,unquantized2,RMTbook}. For their quantized counterparts, the nonlinearity introduced by one-bit quantization complicates the analysis. The Bussgang decomposition has been a useful tool in this context, allowing the expression of a nonlinear transformation of a Gaussian input as a scaled version of the input plus an uncorrelated distortion term~\cite{Bussgang}. This approach has facilitated progress in analyzing and designing linear-quantized precoding schemes, including achievable-rate analysis under both perfect and estimated channel state information (CSI)~\cite{MFrate}, symbol error probability (SEP) analysis~\cite{ZF}, and precoder design based on minimizing the mean square error (MSE)~\cite{MMSE2}, maximizing the weighted sum rate~\cite{WSR}, and maximizing the minimum rate \cite{duality_onebit}. However, the Bussgang decomposition relies on heuristic assumptions, such as the Gaussianity of the input signal and the independence of the distortion term from the signal, which may not hold in all scenarios. Recent work has developed analytical frameworks based on random matrix theory to provide rigorous justifications for these analyses~\cite{SEPlinear}.

These analytical approaches for linear-quantized precoding leverage the ``linear-then-quantize" structure and the explicit expression of the precoding matrix. In contrast, nonlinear precoding schemes are typically obtained by solving dedicated optimization problems without closed-form solutions, rendering these approaches inapplicable. Some works have analyzed the performance of unquantized nonlinear precoding schemes, such as box-constrained precoding \cite{PAPR}, using tools like the convex Gaussian min-max theorem (CGMT) \cite{CGMT}. However, the one-bit constraint in nonlinear one-bit precoding introduces nonconvexity, which substantially complicates the analysis. 

Convex-relaxation-then-quantization (CRQ) has emerged as a widely adopted approximation approach for handling the discrete one-bit constraint, where the transmit signal is obtained by solving a convex  relaxation of the discrete model, followed by one-bit quantization \cite{SQUID,CIfirst,CImodel}. Although this approach has demonstrated strong empirical performance, its theoretical justification remains limited.  The main difficulty lies in that the quantization step alters the statistical properties of the continuous-valued relaxation solution. It is worth mentioning that the idea of handling binary-valued constraints via convex optimization has a long  history. Early works have shown that binary-valued solutions to linear systems can be recovered by solving convex optimization problems under Gaussian measurements \cite{mangasarian2011probability,bey2015sparsity}. However, these results focus on signal recovery and assume the existence of an exact binary solution, and are fundamentally different from the nonlinear one-bit precoding setting considered in this paper.

\subsection{Our Contributions}

This paper focuses on the asymptotic performance analysis of nonlinear one-bit precoding. We investigate a specific CRQ precoding scheme tailored to the one-bit minimum mean square error (MMSE) formulation under a real-valued system model with independently and identically distributed (i.i.d.) Gaussian channel coefficients. The considered scheme 
 encompasses SQUID~\cite{SQUID}, a widely adopted nonlinear one-bit precoding method, as a special case. Building on our analytical results, we further derive the optimal regularization parameters for the considered model. The key contributions of this paper are summarized as follows:

\begin{itemize}
\item \textit{Approximate Message Passing (AMP)-Based Analytical Framework}: We develop a novel analytical framework based on the AMP framework \cite{donoho2009message,AMP,feng2022unifying,LASSO,elasticnet} and its state evolution theory to rigorously characterize the asymptotic empirical distribution of $\mathbf{H} q(\hat{\mathbf{x}})$, where $\mathbf{H}$ is the channel matrix, $\hat{\mathbf{x}}$ is the solution to the relaxation model, and $q(\cdot)$ represents the one-bit quantization function. The key technical ingredient is an auxiliary AMP iteration that properly incorporates the nonlinear mapping $q(\cdot)$, which enables a state-evolution-based analysis of the desired distribution. The proposed framework applies to a broad class of nonlinear functions $q(\cdot)$, allowing the analysis of various transmitter-side nonlinear effects, such as multi-bit quantization and other hardware impairments. From a broader analytical perspective, the proposed technique opens a new  avenue for analyzing nonlinear post-processing of convex optimization solutions and may be extended to non-Gaussian measurement matrices.
\item \textit{Asymptotic SEP Analysis}: Utilizing the proposed AMP-based analytical framework, we derive the asymptotic SEP of the considered precoding scheme. Our analysis quantifies the impact of system parameters, including additive noise power and user-to-antenna ratio, as well as regularization parameters in the considered model, on SEP performance. As a notable special case, our analysis provides a theoretical performance characterization of the popular SQUID precoder, offering insights into its empirical effectiveness.

\item \textit{Optimal Regularization Parameters}: Our empirical results suggest that the $\ell_\infty^2$ regularizer, when paired with an optimally chosen regularization parameter, achieves optimal SEP performance within a broad class of convex regularization functions. To theoretically substantiate this observation, we {\color{black} derive the optimal regularization parameters for our considered relaxation model, which corresponds to the composite regularizer ${\lambda}\|\mathbf{x}\|_\infty^2 + \rho\|\mathbf{x}\|^2$.} We prove that the optimal choice of $\rho$ is zero, thereby establishing the optimality of the $\ell_\infty^2$ regularizer within this mixed regularization class.
\end{itemize}

During the preparation of this manuscript, we became aware of a concurrent work~\cite{ma2025} that extends the analysis of box-constrained precoding \cite{PAPR} to the one-bit case. While both studies focus on the performance analysis of nonlinear one-bit precoding, our work differs in both the optimization model and the analytical framework. Moreover, our study advances further by analytically deriving the optimal regularization parameters in the considered model. For a detailed comparison of these two works, please refer to Section~\ref{subsec:compare}.

 \subsection{Organization and Notations}
\textit{Organization}: The remaining parts of this paper are organized as follows. In Section \ref{sec:2}, we introduce the system model and the considered precoding scheme. In Section \ref{sec:3}, we give an asymptotic performance analysis of the considered scheme and give a brief introduction of our analytical framework. In Section \ref{subsubsec:parameter}, we derive the optimal regularization parameters in the considered model based on our asymptotic result.  A rigorous proof of the asymptotic result is provided in Section \ref{sec:proof}, and the paper is concluded in Section \ref{sec:conclusion}. 

\textit{Notations:}  Throughout this paper, we use boldface lower-case letters (e.g., $\x$), boldface upper-case letters (e.g., $\mathbf{X}$), and upper-case calligraphic letters (e.g., $\mathcal{X}$) to represent vectors, matrices, and sets, respectively. For a vector $\x\in\R^n$,  $\|\x\|_1$, $\|\x\|$, and $\|\x\|_{\infty}$ represent its $\ell_1$ norm, $\ell_2$ norm, and $\ell_\infty$ norm, respectively. The notation  $\left<\x\right>=\frac{1}{n}\sum_{i=1}^nx_i$  represents the pointwise average of $\x$. 
For vectors $\x,\y\in\R^n$,  we write $\left<\x,\y\right>=\frac{1}{n}\sum_{i=1}^nx_iy_i$ as their normalized inner product. For two functions $f, g: \R\to\R$, $f\circ g$ refers to the composition of $f$ and $g$, i.e., $f\circ g(x)=f(g(x))$; $f*g$ denotes the convolution of $f$ and $g$, i.e., $(f*g)(x)=\int_{-\infty}^{+\infty}f(y)g(x-y)dy$. The operators $\mathbb{E}[\cdot]$ and $\mathbb{P}(\cdot)$ return the expectation and the probability of their corresponding arguments, respectively. For a nonempty closed convex set $\mathcal{X}$, the operator $\mathcal{P}_{\mathcal{X}}(\cdot)$ denotes the projection onto  $\mathcal{X}$, i.e., 
$\mathcal{P}_{\mathcal{X}}(\x)=\arg\min_{\y\in\mathcal{X}}~\|\y-\x\|^2$.
The notation $\mathcal{N}({0},\sigma^2)$ represents the Gaussian distribution with zero mean and variance $\sigma^2$.   We use the notations $\phi(x)$, $\Phi(x)$, and $Q(x)$ to denote the probability  density function (PDF), cumulative distribution function (CDF), and tail distribution function of the standard Gaussian distribution, respectively, i.e.,
\begin{equation}
\phi(x)=\frac{1}{\sqrt{2\pi}}e^{-\frac{1}{2}x^2},  ~\Phi(x)=\int_{-\infty}^{x}\phi(t)dt,~Q(x)=1-\Phi(x).
\end{equation}
 A function $\psi:\R^n\rightarrow \R$ is said to be pseudo-Lipschitz (of order
2) if there exists a constant $L_\psi > 0$ such that for any $\x,\y\in\R^n$, 
\begin{equation}
|\psi(\x)-\psi(\y)|\leq L_{\psi}(1+\|\x\|+\|\y\|)\|\x-\y\|.
\end{equation}
We denote almost sure convergence by $\xrightarrow{a.s.}$ and almost sure equality by $\overset{a.s.}{=}$.

\section{Problem Formulation}\label{sec:2}

\subsection{System Model}
Consider a flat-fading downlink multiuser multiple-input single-output (MISO) system, where a BS equipped with $N$ antennas transmits datas to $K$ single-antenna users simultaneously.   The BS is equipped with one-bit DACs to reduce the hardware cost, while the ADCs at the user sides are assumed to have infinite resolution \cite{SQUID}.  In this paper, we focus on a real-valued baseband system model commonly adopted for theoretical analysis \cite{PAPR,ma2025}, as follows:
\begin{equation}\label{y=Hx+n}
\y=\H\x_T+\n.
\end{equation}
In \eqref{y=Hx+n}, $\y\in\R^K$ is the received signal vector at all $K$ users,  $\H\in\R^{K\times N}$ is the channel matrix between the BS and the users, $\n\sim\mathcal{N}(0,\sigma^2\mathbf{I}_K)$ is the additive white Gaussian noise (AWGN), and $\x_T\in\R^N$ is the transmit signals at the BS. We assume an i.i.d. Gaussian channel, i.e., the entries of $\mathbf{H}$ are i.i.d. Gaussian random variables. A summary of all modeling assumptions is provided in Assumption \ref{ass} stated later. Due to  one-bit DACs, the transmit signal vector must satisfy the following one-bit constraint:
\begin{equation}
\x_{T}\in\{-1,1\}^N,
\end{equation}
where for simplicity we have assumed unit transmit power at each antenna. 

Let $\s\in\mathbb{R}^K$ denote the data symbols intended for the users, and assume that perfect channel state information (CSI) is available at the BS. The transmit signal $\x_T=f(\mathbf{H},\s,\sigma^2)$  is allowed to depend not only on the channel matrix $\mathbf{H}$ and the noise variance $\sigma^2$, but also on  the instantaneous data symbols $\s$, hence the term ``symbol-level precoding''. Ideally, the precoding strategy, characterized by the mapping function $f$, is designed such that the received signal $\mathbf{y}=\mathbf{H}\mathbf{x}_T+\mathbf{n}$ (possibly after suitable scaling) closely approximates the intended data symbols $\s$ according to an appropriate distortion measure.

\subsection{CRQ Precoding}\label{precoder}

In this paper, we focus on a convex-relaxation-then-quantization (CRQ) precoding scheme, defined as follows. The rationale behind the proposed CRQ precoding is deferred to Section~\ref{precoder2}.

\begin{definition}[CRQ  precoding]\label{def:onebit}
The transmit signal is 
\begin{equation}\label{eq:onebit}
{\x}_T=q(\hat{\x}),
\end{equation}
where $q(\cdot)=\text{sgn}(\cdot)$ is the one-bit quantizer, and $\hat{\x}$ is the solution to the following problem:
\begin{equation}\label{opt:relaxation}
\begin{aligned}
\hat{\x}=\arg\min_{\x}~&\bigg\{\frac{1}{N}\|\s-\H\x\|^2+\frac{\rho}{N}\|\x\|^2+\lambda \|\x\|_\infty^2\bigg\}. 
\end{aligned}
\end{equation}
Here, $\rho>0$ and $\lambda>0$ are regularization parameters.
\end{definition}

\vspace{3pt}

\begin{remark}
Note that in \eqref{opt:relaxation}, the first two terms are scaled by $\frac{1}{N}$, while the last term does not. This specific scaling ensures that, for fixed values of $\rho > 0$ and $\lambda > 0$, the optimal value of the objective function converges to a well-defined limit as $N, K \to \infty$, under Assumption \ref{ass} stated later.
\end{remark}


\begin{remark}
The popular SQUID precoder \cite{SQUID} is a special case of the above scheme with $\rho=0$ and $\lambda=\frac{\sigma^2K}{N}$, specifically: $\x_T^{\text{\normalfont SQUID}}=q(\x^{\text{\normalfont SQUID}})$, where 
\begin{equation}
{\x}^{\text{\normalfont SQUID}}=\arg\min_{\x}~\left\{\frac{1}{N}\|\s-\H\x\|^2+\frac{\sigma^2K}{N} \|\x\|_\infty^2\right\}.
\end{equation}
\end{remark}

\vspace{5pt}

The performance of the CRQ precoder depends on the joint distribution of the received signal $\y=\H\x_T+\n$ and the target signal $\s$. We take the symbol error probability (SEP) as the performance metric, given by
\begin{equation}\label{def:SEP}
\text{SEP}=\frac{1}{K}\sum_{k=1}^K\mathbb{P}\left(\hat{s}_k\neq s_k\right),
\end{equation}
where $\hat{s}_k$ is a symbol detector for user $k$ based on  $y_k$, and ${\x}_T$ is defined in Definition \ref{def:onebit}.

\subsection{Discussions}\label{precoder2}
The CRQ precoding in Definition \ref{def:onebit} is inspired by the SQUID precoder originally proposed in~\cite{SQUID}. For clarity, we briefly summarize the key principles underlying the SQUID precoder. Following the standard practice in symbol-level nonlinear precoding literature, we start by considering a simple linear data estimator $\tilde{\s} = \xi \y$, where $\xi$ represents a suitably chosen scaling factor. The associated MSE of this estimator is given by:
\begin{equation}\label{eq:MSE}
\frac{1}{N}\mathbb{E}\left[\|\tilde{\s}-\s\|^2\right]=\frac{1}{N}\|\s-\xi\H\x_T\|^2+\frac{K\sigma^2}{N}\xi^2,
\end{equation}
with the expectation taken over the channel noise $\n \sim \mathcal{N}(\mathbf{0},\sigma^2\mathbf{I}_K)$. A natural criterion for designing the transmitted signal $\x_T$ is to minimize this MSE subject to the one-bit constraint, leading to the following optimization problem:
\begin{equation}\label{Eqn:formulation2}
\begin{aligned}
\min_{\x_T,\xi}~&\frac{1}{N}\|\s-\xi\H\x_T\|^2+\frac{K\sigma^2}{N}\xi^2\\
\text{s.t. }~&\x_T\in\{-1,1\}^N.
\end{aligned}
\end{equation}
To simplify this problem further, we introduce an auxiliary variable $\x = \xi\x_T$, thus eliminating the explicit scaling factor $\xi$:
\begin{equation}\label{opt:discrete}
\begin{aligned}
\min_{\x}~&\frac{1}{N}\|\s-\H\x\|^2+\frac{\sigma^2K}{N}\|\x\|^2_\infty\\
\text{s.t. }~&x_1^2 = x_2^2 = \dots = x_N^2.
\end{aligned}
\end{equation}
Here we utilize the relationship $\|\x\|^2_\infty = \xi^2$. However, globally solving~\eqref{opt:discrete} remains computationally prohibitive due to its nonconvexity and high dimensionality. A practical relaxation introduced in \cite{SQUID} is simply to drop the equality constraint, leading to the relaxed problem:
\begin{equation}\label{opt:discrete2}
\min_{\x}~\frac{1}{N}\|\s-\H\x\|^2+\frac{\sigma^2K}{N}\|\x\|^2_\infty,
\end{equation}
followed by quantizing the solution to satisfy the original one-bit constraint. This approach defines the SQUID precoder. Compared to \eqref{opt:discrete2}, the method introduced in \eqref{opt:relaxation} considers a general $\ell_\infty^2$ regularization parameter and incorporates an additional $\ell_2^2$ regularization term, providing additional flexibility for performance optimization.


\section{Asymptotic Performance Analysis}\label{sec:3}

In this section, we present an asymptotic performance analysis of the considered one-bit precoding scheme.  {We begin by stating the technical assumptions in Section \ref{Sec:assumptions}. In Section \ref{subsec:framework}, we provide an overview of the analytical framework. Then, Section \ref{subsec:mainresult}  introduces our main results. 
A heuristic derivation of the main results is given in  Section \ref{subsec:derivation}.} Finally, we give  a detailed comparison of our work with the state-of-the-art work \cite{ma2025} in Section \ref{subsec:compare}.

\subsection{Assumptions}\label{Sec:assumptions}

Our asymptotic analysis is based on the following assumptions.
\begin{assumption}\label{ass}
The following assumptions hold throughout the paper:
\begin{itemize}
\item[(i)] The number of transmit antennas $N$ and the number of users $K$ tend to infinity at the same rate, i.e., 
\begin{equation}
N, K\to \infty \text{ with } \lim_{N,K\to\infty}\frac{K}{N}=\delta\in(0,\infty).
\end{equation}

\item[(ii)] The entries of the  channel matrix $\bH$ are independently drawn from $\mathcal{N}(0,\frac{1}{K})$. 

\item[(iii)] The entries of the data symbol vector $\s$ are independently drawn from {a finite set $\mathcal{S}$, i.e., $s_i\sim{\normalfont \text{Unif}}(\mathcal{S})$.}

\item[(iv)] The quantization function $q:\R\rightarrow \R$ is differentiable except at a finite number of points $\{x_1,x_2,\dots, x_M\}$, where  $x_1,x_2,\dots, x_M\notin\{-a^*,a^*\}$ and $a^*$ is defined in Lemma \ref{lemma:astar} further ahead. In addition, there exists a constant $C>0$ such that $|q'(x)|\leq C$ for all $x\in[-a^*,a^*]\backslash\{x_1,x_2,\dots,x_M\}$.
\end{itemize}
\end{assumption}

Assumptions (i) and (ii) are standard in the literature of high-dimensional asymptotics. {The quantity $\delta>0$ characterizes the system load. Assumption (iii) is general and encompasses arbitrary signal constellations, e.g., pulse amplitude modulation (PAM).} 
In Assumption (iv), we consider a general quantization function $q(\cdot)$ beyond the one-bit quantizer, to highlight the generality of our technique for asymptotic analysis. 

\subsection{An Overview of the Analytical Framework}\label{subsec:framework}
 Our goal is to characterize the asymptotic performance of the one-bit precoding scheme in Definition \ref{def:onebit}, which ultimately depends on the joint distribution of $(\mathbf{H} q(\hat{\x}), \mathbf{s})$, where $\hat{\x}$ is the solution to the convex optimization problem in \eqref{opt:relaxation}. In this subsection, we give an overview of the proposed analytical framework. 

Our analysis is based on the AMP theory. AMP defines an iterative algorithm involving large random matrices in the following form: 
\begin{equation}\label{AMP_form}
\begin{aligned}
\x_{t+1}&=\eta_t(\x_{t}+\H^T\bz_t),\\
\bz_t&=\s-\H\x_t+\frac{1}{\delta}\left<\eta_{t-1}'(\x_{t-1}+\H^T\bz_{t-1})\right>\bz_{t-1},
\end{aligned}
\end{equation}
where $\eta_t:\R\rightarrow \R$ is known as  the denoising function, applied element-wise to its input vector, and $\bH$ satisfies Assumption \ref{ass} (ii). 
A key property of AMP is that, in the high-dimensional limit, its dynamic behavior (e.g., the empirical distribution of $\{\x_t\}_{t\geq 0}$ and $\{\z_t\}_{t\geq 0}$) can be precisely characterized by a one-dimensional iteration known as the state evolution (SE). Leveraging this property, AMP is widely used in the literature to characterize the asymptotic statistical properties of the solution to convex optimization problems, by tailoring the denoising function $\eta_t$ in \eqref{AMP_form} such that the sequence $\{\x_t\}_{t\geq 0}$ converges to the solution of the problem under consideration. Please see Appendix \ref{appendix:AMP} and \cite{donoho2009message,AMP,feng2022unifying,LASSO,elasticnet} for more detailed discussions on AMP. 

Compared to existing analyzes based on AMP,  our considered problem has two unique challenges. First, the regularizer $\|\x\|_\infty^2$ in \eqref{opt:relaxation} is not separable,  whereas the AMP algorithm in \eqref{AMP_form} uses a separable denoiser $\eta_t(\cdot)$ that acts independently on each entry of its input vector. This makes it difficult to establish a direct connection between the optimal solution to \eqref{opt:relaxation} and the limiting point of an AMP algorithm. Second, the solution $\hat{\x}$ to \eqref{opt:relaxation} is post-processed by a nonlinear function $q(\cdot)$,  and our goal is to analyze the distribution of $\H q(\hat{\x})$.  This requires information about the correlation between $\bH$ and $\hat{\x}$, and thus is  more complicated than typical AMP-based analyzes that focus on only the distribution of $\hat{\x}$. 

To address the above challenges, our analytical framework contains two main steps. First, we construct an auxiliary problem with separable regularizers and design an AMP algorithm tailored to it. We show that the auxiliary problem exhibits the same asymptotic statistical behavior as \eqref{opt:relaxation}. Second, we introduce an additional AMP iteration that incorporates the nonlinear function $q(\cdot)$, which enables us to characterize the distribution of $(\bH q(\hat{\x}),\s)$ via the SE theory. 
 
The proposed analytical framework is detailed below. 

\emph{Step 1: Formulate an auxiliary problem with separable regularizers.} 
First, we rewrite \eqref{opt:relaxation} as the following equivalent optimization problem over $(\x,a)$:
\begin{equation}\label{opt:relaxation2}
\begin{aligned}
\min_{a\geq 0}\min_{\x\in[-a,a]^N}~&\frac{1}{N}\|\s-\H\x\|^2+\frac{\rho}{N}\|\x\|^2+\lambda a^2.\\
\end{aligned}
\end{equation}
Problem \eqref{opt:relaxation2} is equivalent to \eqref{opt:relaxation} because, at optimality,  the constraint $\|\x\|_{\infty} \le a$ must be active i.e., $\|\x\|_{\infty}=a$; otherwise, reducing $a$ would strictly decrease the objective value. Consequently, $\x^*$ solves \eqref{opt:relaxation} if and only if  $(\x^*, \|\x^*\|_{\infty})$ solves \eqref{opt:relaxation2}, and both problems share the same optimal value.

Compared to the original formulation in \eqref{opt:relaxation}, problem \eqref{opt:relaxation2} is more technically tractable as, for a fixed $a$, the inner problem involves separable box constraints and can be analyzed using AMP.
Let 
\begin{equation}\label{def:fN}
f_N(a;\s,\H):=\min_{\x\in[-a,a]^N}~\left\{\frac{1}{N}\|\s-\H\x\|^2+\frac{\rho}{N}\|\x\|^2\right\}
\end{equation}
denote the inner problem over variable $\x$, and let 
\begin{equation}\label{def:xa}
{\x}_a:=\arg\min_{\x\in[-a,a]^N}~\left\{\frac{1}{N}\|\s-\H\x\|^2+\frac{\rho}{N}\|\x\|^2\right\}
\end{equation} be the solution to the inner problem for a given $a\geq 0$.  With the above notations, problem \eqref{opt:relaxation2} can be written as 
\begin{equation}\label{relaxation:a}
\min_{a\geq 0}f_N(a;\s;\H)+\lambda a^2, 
\end{equation}
and the optimal solution to \eqref{opt:relaxation} can be expressed as $\hat{\x}={\x}_{\hat{a}_N}$, where 
\begin{equation}\label{def:hata}
\hat{a}_N=\arg\min_{a\geq 0} \left\{f_N(a;\s,\H)+\lambda a ^2\right\}
\end{equation}
 is the solution to \eqref{relaxation:a}. 
 
For any given $a\geq 0$, the box-constrained problem in \eqref{def:xa} can be analyzed via AMP; see Appendix \ref{appD:heuristic} for a heuristic derivation of the AMP algorithm and \eqref{eq:etat} for the expression of the denoising function $\eta_t(x)$.   Let $\{(\x_{t+1},\z_t)\}_{t\geq 0}$ be the sequence generated by the corresponding AMP algorithm. 
Using basic properties of AMP,  we can prove that 
\begin{equation}\label{xttoxa}
\x_{t+1}\to\x_a~~\text{as }t,N\to\infty, (\text{in some sense; see Proposition \ref{prob:converge})}
\end{equation}
i.e., the  iterations generated by the AMP algorithm converge to the solution of the box-constrained problem in \eqref{def:xa}. In addition,
 \begin{equation}\label{eq:convergefN}
f_N(a;s,\H)+\lambda a^2\xrightarrow{a.s.} f(a),~\text{as }N\to\infty,~~~\text{(see Lemma \ref{convergefN})}
\end{equation} i.e., the objective function of problem \eqref{relaxation:a} converges to a deterministic function $f(a)$ in the high-dimensional limit; see Lemma   \ref{lemma:astar} for the expression of $f(a)$.  Based on this, we can further  establish the convergence of the solution to \eqref{relaxation:a}: 
 \begin{equation}\label{eq:aNtoastar}
\hat{a}_N\xrightarrow{a.s.}a^*:= \arg\min_{a\geq 0} f(a), ~~\text{as }N\to\infty.~~\text{(see Lemma \ref{converge:a})}
\end{equation}
The above result implies that $\hat{\x}=\hat{\x}_{a_N}$ exhibits the same statistical properties as  $\x_{a^*}$ in a limiting sense, which reduces our task of analyzing $\bH q(\hat{\x})$ to analyzing $\bH q(\x_{a^*})$; see Proposition \ref{the1} for a rigorous statement. 
 
\emph{Step 2: Introduce an additional AMP iteration that incorporates $q(\cdot)$ to analyze the distribution of $\bH q(\x_{a^*})$.} Based on Step 1, we now focus on the box-constrained problem in \eqref{def:xa}.   Our goal is to analyze the distribution of $\bH q(\x_{a^*})$. For notational simplicity, we omit the superscript ``$*$'' on $a$ in the subsequent discussions.

  We introduce the following post-processing step for the sequence $\{(\x_{t+1},\bz_t)\}_{t\geq 0}$ generated by the AMP algorithm corresponding to \eqref{def:xa}, where $\eta_t(x)$ denotes its denoising function:
 \begin{subequations}\label{xqzq}
 \begin{align}
 \tilde{\x}_{t+1}&=q(\x_{t+1}),\label{xq}\\
\tilde{\z}_{t+1}&=\s-\bH\tilde{\x}_{t+1}+\frac{1}{\delta}\left<(q\circ\eta_t)'(\x_t+\bH^T\z_t)\right>\z_t.\label{zq}
 \end{align}
 \end{subequations}
 Hereafter, we use the tilde symbol over scalars/vectors to denote constants/random variables related to post-processing steps. 
 It follows from  \eqref{xttoxa} that $\tilde{\x}_{t+1}\to q(\x_{a})$, and thus $\H\tilde{\x}_{t+1}\to\bH q(\x_{a})$. Hence, to analyze the distribution of $\bH q(\x_{a})$, we can study the limiting distribution of $\bH\tilde{\x}_{t+1}$. According to  \eqref{zq},  
 \begin{equation}\label{Hxq}
  \bH\tilde{\x}_{t+1}=\s-\tilde{\z}_{t+1}+\frac{1}{\delta}\left<(q\circ\eta_t)'(\x_t+\bH^T\z_t)\right>\z_t.
  \end{equation}
The motivation for defining the post-processing step as in \eqref{xqzq} is that  the pair $(\tilde{\x}_{t+1},\tilde{\z}_{t+1})$ can be viewed as the result of performing an additional AMP iteration after obtaining $(\x_t,\z_t)$, with the denoising  function set as the post-processing function $q\circ\eta_t(x).$   In other words, $(\x_{\leq t},\tilde{\x}_{t+1})$ and $(\bz_{\leq t},\tilde{\bz}_{t+1})$ are the resulting sequences   of  the AMP algorithm with \begin{equation}\label{etat}
\tilde{\eta}_i(x)=\left\{
\begin{aligned}
&\eta_i(x),~~~~~~~\text{if }~i< t;\\%
&q\circ\eta_t(x),~~~\text{if }~i=t.
\end{aligned}\right.
\end{equation}
 Therefore, for each finite $t$, we can calculate various quantities related to $\tilde{\x}_{t+1}$ based on the SE theory, and our prediction for $\bH q(\x_{a})$ can be obtained by taking the limit $t\to\infty$; see Section \ref{subsec:derivation} for a heuristic derivation based on this argument and Section \ref{sec:proof} for a rigorous proof.
  
\begin{remark}\label{remark:generaility}
 The above technique provides a general analytical framework for characterizing the large-system behavior of {nonlinearly post-processed} convex optimization solutions. In many applications, one first computes a solution $\x^*$ obtained from a convex optimization problem, whose asymptotic behavior can be characterized via AMP, while the actual  signal is obtained through a nonlinear mapping $q(\x^*)$. For example, in hardware-limited communication systems, $q(\cdot)$ may represent a general quantization or clipping operation that capture hardware constraints beyond one-bit DACs \cite{WFQ,CE,SEPlinear}; in signal detection and  recovery problems, $q(\cdot)$ may serve as a hard decision or thresholding operation to recover discrete-valued signals \cite{mangasarian2011probability,bey2015sparsity,thrampoulidis2018symbol}. In these cases, system performance depends on quantities involving $\bH q(\x^*)$, rather than on $\x^*$ itself.

Such nonlinear post-processing can be incorporated into the SE framework through an auxiliary AMP iteration that embeds $q(\cdot)$, as in \eqref{xqzq}, which enables asymptotic performance analysis for metrics involving $\bH q(\x^*)$. It is also worth mentioning that  this analytical approach is not restricted to i.i.d. Gaussian measurement matrices and may be extended to broader classes of random matrices by combining it with advanced AMP variants \cite{ma2017orthogonal,rangan2019vector,fan2022approximate,liu2024unifying,zhao2022asymptotic,dudeja2023universality,wang2024universality}.

  \end{remark}

\subsection{Main Results}\label{subsec:mainresult}
In this subsection, we present the main results of the paper. 
 Using the proposed analytical framework  in Section \ref{subsec:framework}, we first characterize the asymptotic statistical properties of the convex relaxation solution $\hat{\x}$ in \eqref{opt:relaxation}. The corresponding result is summarized in Theorem \ref{mainresult:ed} in Section \ref{subsubsec:hatx}.  Then, we give the asymptotic empirical distribution of $(\mathbf{H} q(\hat{\x}), \mathbf{s})$, which is formalized in Theorem \ref{mainresult} in Section \ref{subsubsec:SEP} and is the main technical contribution of the paper. Using Theorem \ref{mainresult}, we further derive a closed-form expression for the asymptotic SEP of the proposed scheme under binary phase shift keying (BPSK)  constellation, given as Theorem \ref{SEP} in Section \ref{subsubsec:SEP}.

 We begin with the following lemmas, which are important for presenting our main results. The function and quantities defined therein are associated with the AMP algorithm tailored to the convex optimization problem in \eqref{opt:relaxation}. 
 
\begin{lemma}\label{lemma:unique}
Given $\rho>0$,  $\delta>0$, and $a\geq 0$, define $\eta_a:\mathbb{R}\times\mathbb{R}_{>0}\to\mathbb{R}$ as
\begin{equation}\label{def:etaa}
\eta_a(x;\gamma)=\mathcal{P}_{[-a,\,a]}\left(\frac{x}{\gamma+1}\right).
\end{equation}
Then, there exists a unique solution $(\tau_a^2,\gamma_a)$ to the following equations:
\begin{subequations}\label{condition0}
\begin{align}
\tau^2&=1+\frac{1}{\delta}\mathbb{E}\left[\eta_a^2(\tau Z ;\gamma)\right],\label{conditiona}\\
\rho&=\gamma\left(1-\frac{1}{\delta(\gamma+1)}\mathbb{P}\left(\frac{\tau Z}{\gamma+1}\in[-a,a]\right)\right),\label{condition0b}
\end{align}
\end{subequations} 
where $Z\sim\mathcal{N}(0,1)$.  
\end{lemma}
\begin{proof}
We prove the existence of $(\tau_a^2,\gamma_a)$ in Appendix \ref{app:exist}. The uniqueness of $(\tau_a^2,\gamma_a)$ is a corollary of Proposition \ref{prob:converge} in Section \ref{sec:proof}; see a detailed proof in Appendix \ref{app:unique}.\end{proof}
The AMP algorithm corresponding to problem \eqref{def:xa} is obtained by setting the denoising function as 
\begin{equation}\label{eq:etat}
\eta_t(x)=\eta_a(x;\gamma_a).
\end{equation} 
In particular, \eqref{conditiona} is the SE and \eqref{condition0b} guarantees that the limit point of the AMP algorithm corresponds to the solution of  \eqref{def:xa}. See Appendix \ref{appD:heuristic} for detailed discussions. 

\begin{lemma}\label{lemma:astar}
Given $\rho>0$, $\lambda>0$, and $\delta>0$, 
define the function $f: [0,\infty)\to\R$ as
\begin{equation}\label{phia}
{f}(a)=\delta\rho(\tau_a^2-1)+\frac{\delta\rho^2\tau_a^2}{\gamma_a^2}+\lambda a^2,
\end{equation} 
where $(\tau_a^2,\gamma_a)$ is the solution to \eqref{condition0}.
Then $f(a)$ is strongly convex and continuously differentiable on $[0,\infty)$. Furthermore, the unique minimizer of $f(a)$ over $[0,\infty)$,   defined by
 \begin{equation}\label{minfa}
a^*=\arg\min_{a\geq 0}\,f(a),
\end{equation} satisfies $a^*>0$.
\end{lemma}
\begin{proof}
See Appendix \ref{proof:astar}.
\end{proof}
For illustration, we plot $f(a)$ in Fig. \ref{fig:fa}.  The function $f(a)$  corresponds to the large-system limit of the objective function in problem \eqref{relaxation:a}; see \eqref{eq:convergefN}. 

\begin{figure}
\includegraphics[width=0.45\textwidth]{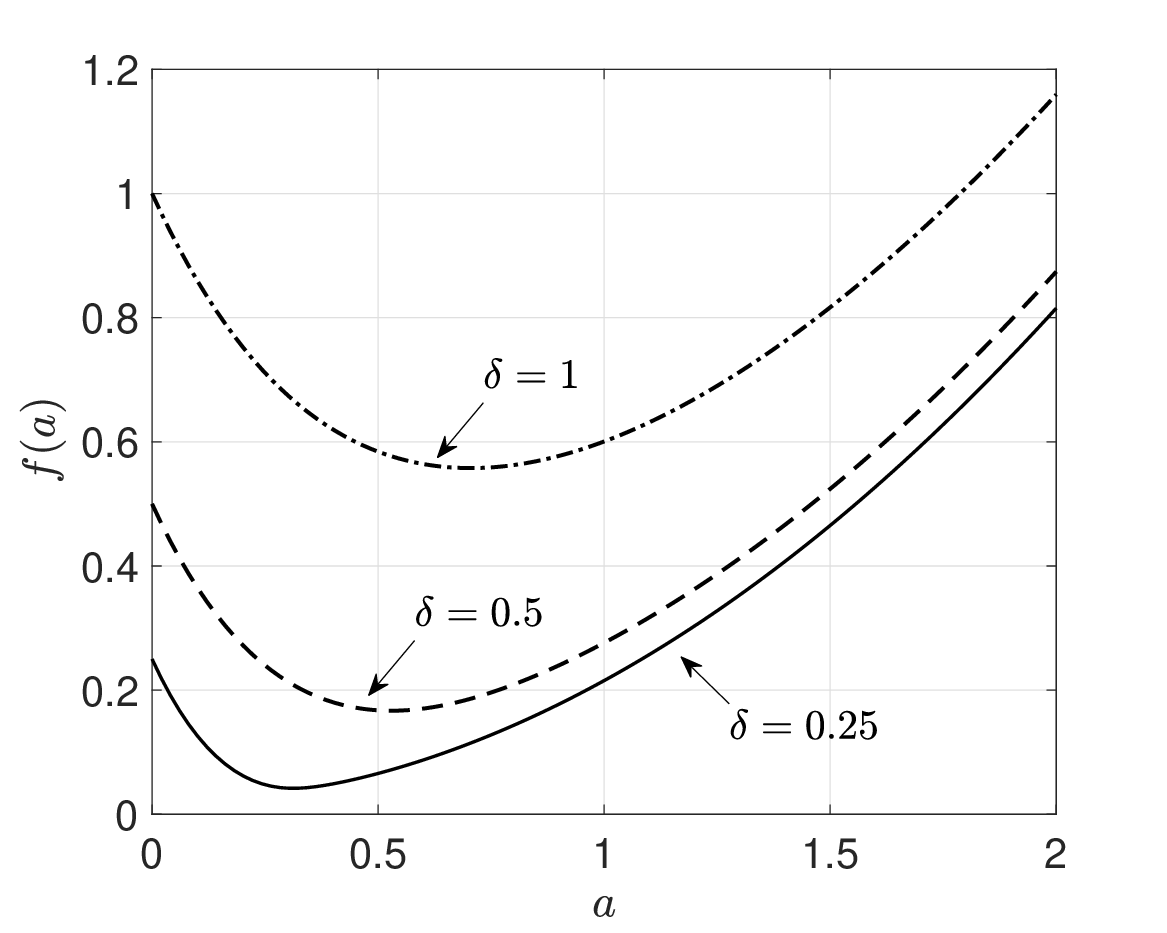}
\centering
\caption{An illustration of $f(a)$, where $\lambda=\rho=0.2$.}
\label{fig:fa}
\end{figure}

We are now ready to present our main results. 
\subsubsection{Asymptotic Statistical Properties of $\hat{\x}$}\label{subsubsec:hatx} 
We first present the asymptotic statistical properties of the optimal solution $\hat{\x}$ to the convex relaxation model in \eqref{opt:relaxation}. The following results are intermediate results obtained from Step 1 of the proposed analytical framework.

\begin{proposition}[Convergence of $\|\hat{\x}\|_{\infty}$]\label{x_infty}
Under Assumption \ref{ass} (i) -- (iii), the following result holds:
\begin{equation}
\lim_{N\to\infty}\|\hat{\x}\|_{\infty}\overset{a.s.}{=} a^*,
\end{equation}
where $\hat{\x}$ and $a^*$ are defined in \eqref{opt:relaxation} and \eqref{minfa}, respectively.
\end{proposition}
\begin{proof}
Proposition \ref{x_infty} follows directly from \eqref{eq:aNtoastar} and the definition of $\hat{a}_N$. See Lemma \ref{converge:a} in Section \ref{sec:proof} for a rigorous statement of \eqref{eq:aNtoastar} and Appendix \ref{proof:Lemma4} for the proof.
\end{proof}
The following theorem characterizes the limiting empirical distribution of $\hat{\x}$.

\begin{theorem}[Limiting Empirical Distribution of $\hat{\x}$]\label{mainresult:ed}
Denote   $(\tau_*^2,\gamma_*):=(\tau_{a^*}^2,\gamma_{a^*})$ as the solution to \eqref{condition0} with $a=a^*$. 
For any pseudo-Lipschitz function $\varphi:\R\rightarrow\R$, the following result holds under Assumption \ref{ass} (i) -- (iii):
\begin{equation}
\lim_{N\to\infty}\frac{1}{N}\sum_{i=1}^N \varphi(\hat{x}_i)\overset{a.s.}{=}\mathbb{E}\left[\varphi(\hat{X})\right],
\end{equation}
where 
\begin{equation}\label{def:hatX}
\hat{X}=\eta_{a^*}\left({\tau_*Z};{\gamma_*}\right),~~Z\sim\mathcal{N}(0,1).\end{equation}\end{theorem}
\begin{proof}
Theorem \ref{mainresult:ed} is obtained by analyzing the distribution of the solution  $\x_a$  to the box-constrained problem  in \eqref{def:xa} via AMP and invoking Proposition \ref{x_infty}, which implies that $\hat{\x}$ exhibits the same statistical properties as $\x_{a^*}$; see also the discussions below \eqref{eq:aNtoastar}. 
The detailed proof is given in Appendix \ref{proof:theorem1}.
\end{proof}
Theorem \ref{mainresult:ed} illustrates that in the high-dimensional limit,  the elements of $\hat{\x}$ follow a truncated Gaussian distribution. An interesting implication of Proposition \ref{x_infty} and Theorem \ref{mainresult:ed} is the following remark, which provides insights into why the convex model in \eqref{opt:relaxation} is a good  relaxation of \eqref{opt:discrete}. 
\begin{remark}\label{remark1}
Given any $\epsilon> 0$, by Proposition \ref{x_infty} and Theorem \ref{mainresult:ed},  the following holds for sufficiently large $N$: 
\begin{equation}\label{eq:proportion}
\begin{aligned}
\frac{\#\left\{|\hat{x}_i|\geq \|\hat{\x}\|_{\infty}(1-\epsilon)\right\}}{N}\approx&\frac{\#\left\{|\hat{x}_i|\geq a^*(1-\epsilon)\right\}}{N}
\xrightarrow{a.s.}\mathbb{P}\left(\frac{\tau_*|Z|}{\gamma_*+1}\geq a^*(1-\epsilon)\right),
\end{aligned}
\end{equation}
where $\#$  denotes the cardinality of the corresponding set, and the convergence is obtained by specifying $\varphi=\mathbf{1}_{\{|x|\geq a^*-\epsilon\}}$ in Theorem \ref{mainresult:ed}\footnote{A subtle point here is that Theorem \ref{mainresult:ed} is not directly applicable since the indicator function is not pseudo-Lipschitz. Nevertheless, the set of discontinuous points $\{\pm(a^*-\epsilon)\}$ has zero measure with respect to the truncated Gaussian distribution. Thus, using a standard argument to approximate the indicator function by a Lipschitz function we can show the desired  convergence; see Appendix \ref{app:nonsmooth}.
}.
This means that the proportion of elements in $\hat{\x}$ that cluster around $\pm\|\x\|_{\infty}$ can be characterized by the probability given on the right-hand side of \eqref{eq:proportion}.  In Fig. \ref{fig:proportion}, we plot the values of the left- and right-hand sides of
\eqref{eq:proportion} as a function of $\epsilon$. As can be observed, a large amount of entries of $\hat{\x}$ cluster around $\pm\|\hat{\x}\|_{\infty}$, especially for small $\delta$. This suggests that a large fraction of the entries of $\hat{\x}$ tend to have similar magnitudes,  indicating that the nonconvex constraint in \eqref{opt:discrete} remains largely preserved in the relaxation model \eqref{opt:relaxation}. Consequently, \eqref{opt:relaxation} serves as a good relaxation of \eqref{opt:discrete}.
\begin{figure}
\includegraphics[width=0.45\textwidth]{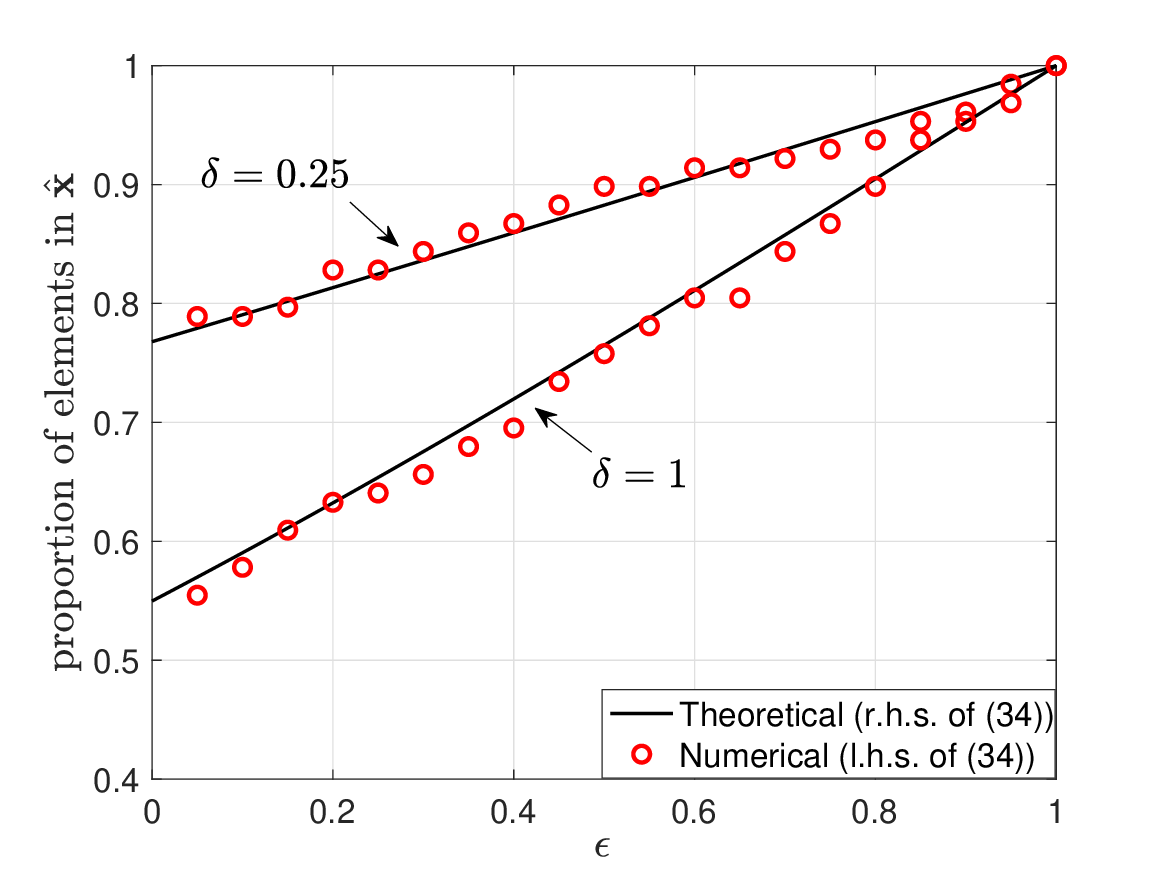}
\centering
\caption{An illustration of the values in Eq. \eqref{eq:proportion} as a function of $\epsilon$. The red circles represent the proportion of elements in  $\hat{\x}$ that cluster around $\pm\|\hat{\x}\|_{\infty}$, i.e., the value on the l.h.s. of \eqref{eq:proportion}. The number of transmit antennas and users are $N=128$ and $K=\delta N$, respectively. The black lines represent the theoretical prediction, i.e., the probability on the right-hand side of \eqref{eq:proportion}. The regularization parameters are set to their optimal values specified in Theorem \ref{mainresult:lambda_rho}.}
\label{fig:proportion}
\end{figure}
\end{remark}

\subsubsection{Limiting Empirical Distribution of $(\bH q(\hat{\x}),\s)$ and Convergence of SEP}\label{subsubsec:SEP}

The following theorem {is the main technical result of this paper, which characterizes  the limiting empirical  distribution of $(\bH q(\hat{\x}),\s)$ for a general quantization function $q(\cdot)$ that satisfies Assumption \ref{ass} (iv). }
\begin{theorem}[Limiting Empirical Distribution of $(\bH q(\hat{\x}),\s)$]\label{mainresult}
 Under Assumption \ref{ass}, the following result holds for any pseudo-Lipschitz  function $\psi:\R^2\rightarrow \R$: 
  \begin{equation}\label{h_kq}
  \lim_{K\to\infty}\frac{1}{K}\sum_{k=1}^K\psi(\h_k^T q(\hat{\x}),s_k)\overset{a.s.}{=}\mathbb{E}\big[\psi(\bar{\alpha}_*S+\bar{\beta}_*^{\frac{1}{2}}Z,S)\big],
  \end{equation}
 where 
\begin{equation}
 \begin{aligned}
 \bar{\alpha}_*=&\frac{1}{\delta\tau_*}\mathbb{E}\left[Z\, q\left(\eta_{a^*}(\tau_*Z;\gamma_*)\right)\right],\\
  \bar{\beta}_*=&\frac{1}{\delta}\,\mathbb{E}\bigg[\big(\bar{\alpha}_*\hat{X}-q(\hat{X})\big)^2\bigg],
 \end{aligned}
 \end{equation}
  $S\sim\text{\normalfont{Unif}}({\mathcal{S}})$, ${Z}\sim\mathcal{N}(0,1)$ is independent of $S$, and  $\hat{X}$ is defined in \eqref{def:hatX}. 
\end{theorem}
\begin{proof}
{The main idea for proving Theorem \ref{mainresult} is presented in Section \ref{subsec:framework}. See Section \ref{subsec:derivation} for a heuristic derivation and Section \ref{sec:proof} for the rigorous proof}.
\end{proof}

Note that $q(x)=\text{sgn}(x)$ satisfies Assumption \ref{ass} (iv) as  $a^*>0$. By specifying $q(x)=\text{sgn}(x)$ in Theorem \ref{mainresult}, we can {characterize the limiting empirical distribution of  $(\bH {\x}_T,\s)$, where ${\x}_T=q(\hat{\x})$  is the one-bit transmit signal defined in Definition \ref{def:onebit}. }
\begin{corollary}[Limiting Empirical Distribution: $q(\cdot)=\text{sgn}(\cdot)$]\label{Corollary1}
 Under Assumptions  \ref{ass} (i) -- (iii), the following result holds for any pseudo-Lipschitz  function $\psi:\R^2\rightarrow \R$: 
  \begin{equation}\label{mainresult:sgnfunction}
  \lim_{K\to\infty}\frac{1}{K}\sum_{k=1}^K\psi(\h_k^T {\x}_T,s_k)\overset{a.s.}{=}\mathbb{E}\big[\psi(\bar{\alpha}_0S+\bar{\beta}_0^{\frac{1}{2}}Z,S)\big],
  \end{equation}
  where $\x_T$ is defined  in Definition \ref{def:onebit} and
  \begin{subequations}\label{Eqn:alpha_beta}
\begin{align}
\bar{\alpha}_0&=\sqrt{\frac{2}{\pi}}\frac{1}{\delta\tau_*},\\
\bar{\beta}_0&=\frac{1}{\delta}\mathbb{E}\left[(\bar{\alpha}_0|\hat{X}|-1)^2\right],
\end{align}
\end{subequations}
  $S\sim\text{\normalfont{Unif}}({\mathcal{S}})$, ${Z}\sim\mathcal{N}(0,1)$ is independent of $S$, and $\hat{X}$ is defined in \eqref{def:hatX}. 
\end{corollary}
Corollary \ref{Corollary1} implies that the statistical properties of the considered system model, i.e., $\y=\bH{\x}_T+\n$, with the precoding scheme in Definition \ref{def:onebit}, can be asymptotically characterized by the following scalar model:
\begin{equation}\label{asymptotic}
\bar{y}=\bar{\alpha}_0S+\sqrt{\bar{\beta}_0+\sigma^2}Z,
\end{equation}
where $\sigma^2$ is the variance of the channel noise, and $(\bar{\alpha}_0,\bar{\beta}_0)$ are deterministic constants given by \eqref{Eqn:alpha_beta}. 

Based on Corollary \ref{Corollary1}, we are ready to characterize the asymptotic SEP of the proposed one-bit precoder. {To enable a closed-form expression for the SEP, we focus on the BPSK constellation in the following discussions.}

\begin{theorem}[Convergence of SEP]\label{SEP}
{Assume a BPSK constellation, i.e., $\mathcal{S}=\{-1,1\}$, and consider} the symbol detector $\hat{s}_k=\text{sgn}\left(\h_k^T{\x}_T+n_k\right),$~$k=1,2,\dots, K$. The following result holds under Assumption \ref{ass}  (i) -- (iii):
\begin{equation}
\lim_{K\to\infty}\frac{1}{K}\sum_{k=1}^K\mathbb{P}\left(\hat{s}_k\neq s_k\right)\overset{a.s.}{=}Q\left(\sqrt{\overline{\text{\normalfont{SNR}}}}\right),
\end{equation}
where $s_k$ and $\hat{s}_k$ are the intended data symbol and the detected symbol at user $k$, respectively, $Q(x)$ is the  tail distribution function of the standard Gaussian distribution, and $\overline{\text{\normalfont SNR}}$ is the signal to noise ratio (SNR) of the asymptotic model in \eqref{asymptotic}, given by
\begin{equation}\label{barSNR}
\overline{\text{\normalfont SNR}}=\frac{\bar{\alpha}_0^2}{{\bar{\beta}_0+\sigma^2}},
\end{equation}
where $\bar{\alpha}_0$ and $\bar{\beta}_0$ are defined in \eqref{Eqn:alpha_beta}.
\end{theorem}
\begin{proof}
{The SEP of the symbol estimator $\hat{s}_k=\text{sgn}\left(\h_k^T{\x}_T+n_k\right)$ is 
\begin{equation}\label{sep:1}
\begin{aligned}
\mathbb{P}(\hat{s}_k\neq s_k)&=\mathbb{E}\left[\mathbb{P}\left(\hat{s}_k\neq s_k\mid\h_k^T{\x}_T,s_k\right)\right]\\
&=\mathbb{E}\left[\mathbb{P}\left(s_k(\h_k^T{\x}_T+n_k)\leq 0\mid\h_k^T{\x}_T,s_k\right)\right]\\
&=\mathbb{E}\left[\mathbb{P}\left( s_kn_k\leq-s_k\h_k^T{\x}_T \mid\h_k^T{\x}_T,s_k\right)\right]\\
&=\mathbb{E}\left[Q\left(\frac{s_k\h_k^T{\x}_T}{\sigma}\right)\right],
\end{aligned}
\end{equation}
where the first equality applies the law of total probability and the expectation is taken over $\h_k^T{\x}_T$ and $s_k$,}  and the last equality holds since $s_kn_k\sim\mathcal{N}(0,\sigma^2)$ as 
$s_k\sim\text{Unif}(\{-1,1\})$ and $n_k\sim\mathcal{N}(0,\sigma^2)$ are independent. It is simple to check that the function $Q(xs/\delta)$ is pseudo-Lipschitz in $(x,s)$. Hence, by specifying $\psi(x,s)=Q(xs/\delta)$ in Corollary \ref{Corollary1}, we get
\begin{equation}
\begin{aligned}
\lim_{K\to\infty}\frac{1}{K}\sum_{k=1}^KQ\bigg(\frac{s_k\h_k^T{\x}_T}{\sigma}\bigg)&\overset{a.s.}{=}\mathbb{E}\bigg[Q\bigg(\frac{\bar{\alpha}_0+\bar{\beta}_0^{\frac{1}{2}}ZS}{\sigma}\bigg)\bigg]=\mathbb{E}\bigg[Q\bigg(\frac{\bar{\alpha}_0+\bar{\beta}_0^{\frac{1}{2}}Z}{\sigma}\bigg)\bigg],
\end{aligned}
\end{equation}
where the last equality holds since $SZ\sim\mathcal{N}(0,1)$, i.e., $SZ$ follows the same distribution as 
$Z$, as $S\sim\text{Unif}(\{-1,1\})$ and $Z\sim\mathcal{N}(0,1)$ are independent. {Note that the random variable 
$\frac{1}{K}\sum_{k=1}^K Q\left(\frac{s_k\h_k^T{\x}_T}{\sigma}\right)$ is uniformly bounded since $Q(\cdot)$ is bounded. Hence, according to the dominated convergence theorem, we get
\begin{equation}
\hspace{-0.3cm}\lim_{K\to\infty}\mathbb{E}\bigg[\frac{1}{K}\sum_{k=1}^K Q\bigg(\frac{s_k\h_k^T{\x}_T}{\sigma}\bigg)\bigg]=\mathbb{E}\bigg[Q\bigg(\frac{\bar{\alpha}_0+\bar{\beta}_0^{\frac{1}{2}}Z}{\sigma}\bigg)\bigg].
 \end{equation}
 Combining this with \eqref{sep:1} gives 
\begin{equation}\lim_{K\to\infty}\frac{1}{K}\sum_{k=1}^K\mathbb{P}(\hat{s}_k\neq s_k)=\mathbb{E}\left[Q\left(\frac{\bar{\alpha}_0+\bar{\beta}_0^{\frac{1}{2}}Z}{\sigma}\right)\right].
 \end{equation}}
The remaining task is to calculate the expectation in the above equation. By the definition of $Q$-function, we have
\begin{equation}
Q\left(\frac{\bar{\alpha}_0+\bar{\beta}_0^{\frac{1}{2}}Z}{\sigma}\right)=\mathbb{P}\left(W>\frac{\bar{\alpha}_0+\bar{\beta}_0^{\frac{1}{2}}Z}{\sigma}\right),
\end{equation}
where $W\sim\mathcal{N}(0,1)$ is independent of $Z$.  It follows that 
\begin{equation}
\begin{aligned}\mathbb{E}\bigg[Q\bigg(\frac{\bar{\alpha}_0+\bar{\beta}_0^{\frac{1}{2}}Z}{\sigma}\bigg)\bigg]&=\mathbb{E}\bigg[\mathbb{P}\bigg(W>\frac{\bar{\alpha}_0+\bar{\beta}_0^{\frac{1}{2}}z}{\sigma}\bigg)\mid Z=z\bigg]\\&=\mathbb{P}\left(W>\frac{\bar{\alpha}_0+\bar{\beta}_0^{\frac{1}{2}}Z}{\sigma}\right)\\
&=\mathbb{P}\left(\sigma W-\bar{\beta}_0^{\frac{1}{2}}Z>{\bar{\alpha}_0}\right)\\
&=Q\left(\frac{\bar{\alpha}_0}{\sqrt{\bar{\beta}_0+\sigma^2}}\right),
\end{aligned}
\end{equation}
where the last equality is due to $\sigma W-\bar{\beta}_0^{\frac{1}{2}}Z\sim\mathcal{N}(0,\bar{\beta}_0+\sigma^2)$. This completes the proof.

\end{proof}
In Fig. \ref{fig_SEP}, we validate Theorem \ref{SEP} under different system and regularization parameters.  As can be observed, the actual SEP closely  matches the theoretical result given in  Theorem \ref{SEP}.
\begin{figure}
\subfigure[$\rho=0.2,\lambda=0.2$.]{\includegraphics[width=0.45\textwidth]{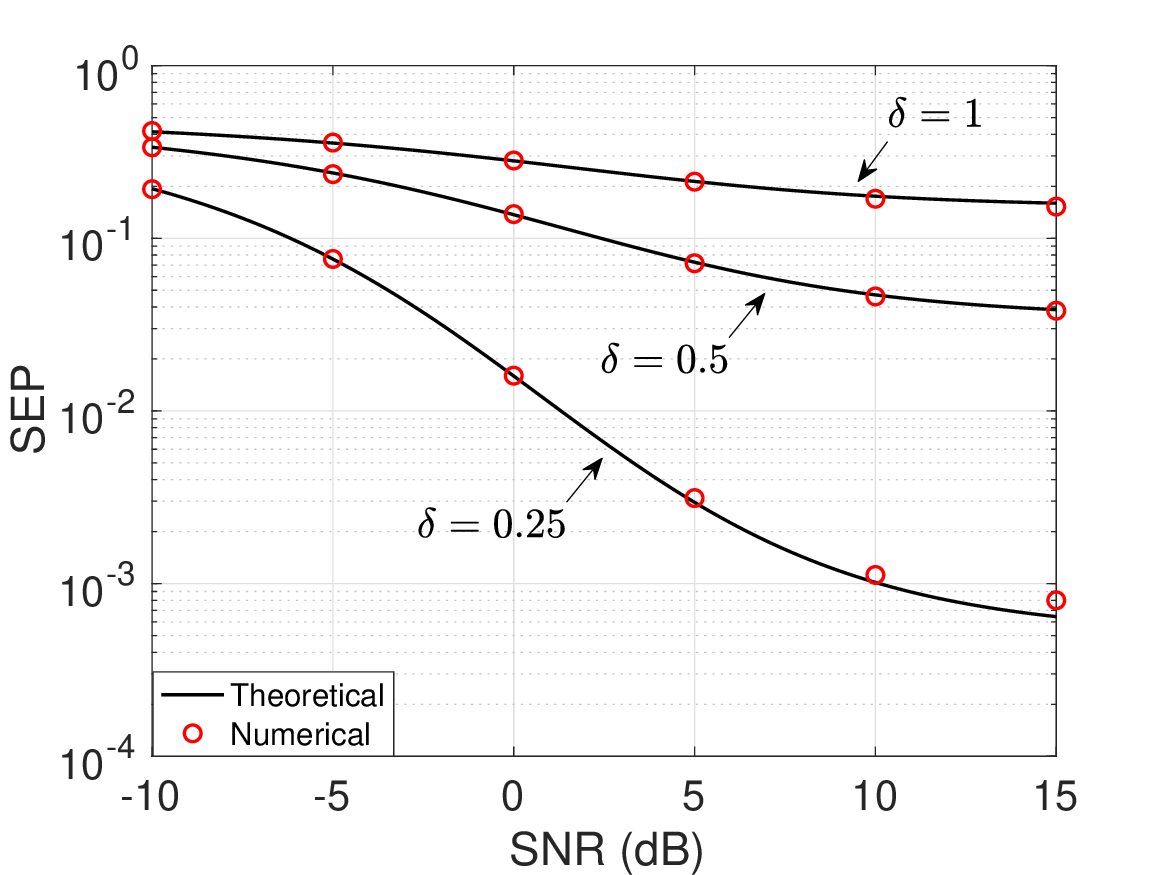}}
\subfigure[$\rho=0,\lambda=0.3$.]{\includegraphics[width=0.45\textwidth]{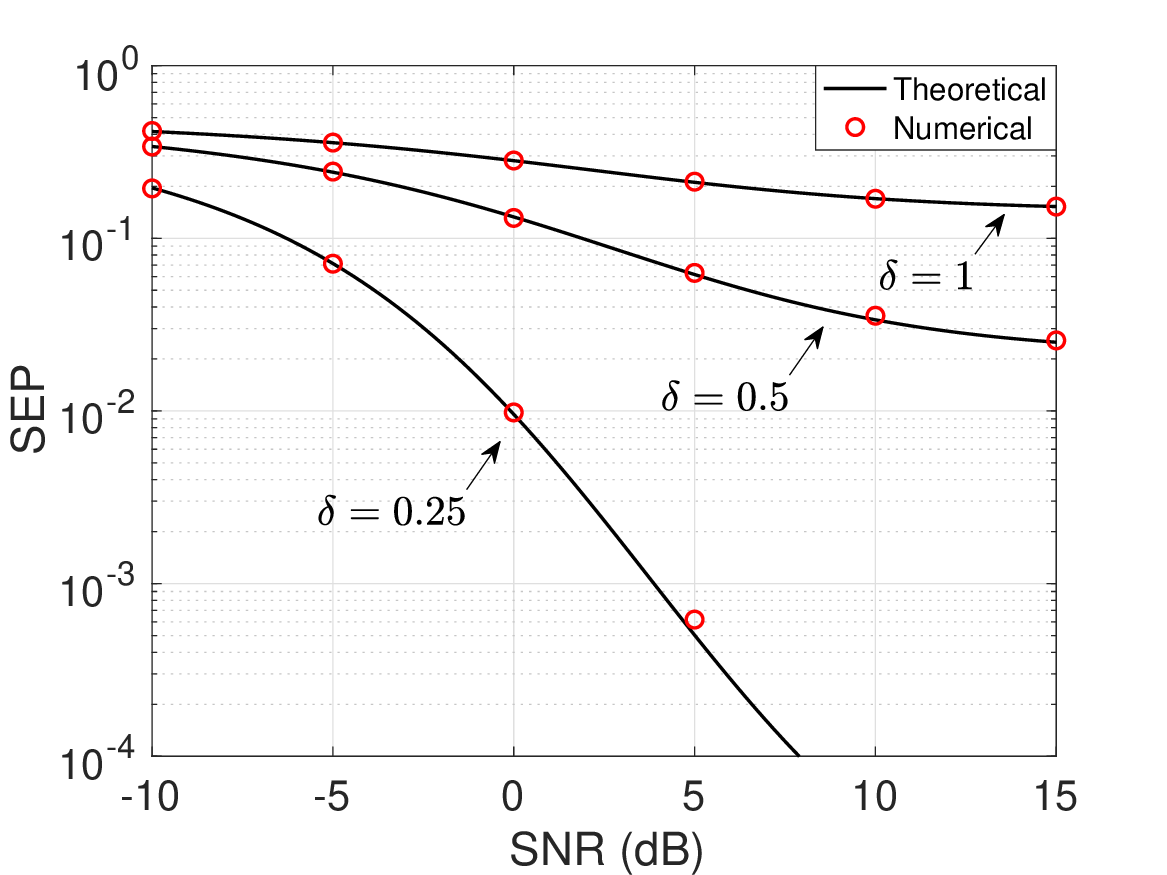}\label{SEPb}}
\centering
\caption{Theoretical and numerical SEP, where the number of transmit antennas is fixed as $N=128$. The number of users is set to  $K=\delta N$. The numerical results are averaged over $10^4$ channel realizations.}
\label{fig_SEP}
\end{figure}

A well-known phenomenon in one-bit precoding is the presence of an SEP error floor at high SNRs. This behavior is rigorously characterized by Theorem \ref{SEP}. In particular, as the SNR tends to infinity (i.e., $\sigma^2\to 0$), the resulting asymptotic SEP given by Theorem \ref{SEP} remains strictly positive for any positive system load $\delta>0$. This nonzero error floor arises because of the Gaussian distortion term $\bar{\beta}_0Z$ in \eqref{asymptotic}, which is induced by one-bit quantization and is nonzero for any $\delta>0$. Nevertheless, the SEP can be extremely small for light system loads, as shown in Fig. \ref{fig:SEP_sigma0}.  Our results provide insights into how large the system load $\delta$ can be to achieve a prescribed SEP requirement.
  For example, according to Fig. \ref{fig:SEP_sigma0}, achieving an SEP below $10^{-4}$ requires $\delta\leq 0.3$, i.e., the number of served users $K$ does not exceed $0.3N$.
\begin{figure}
\includegraphics[width=0.4\textwidth]{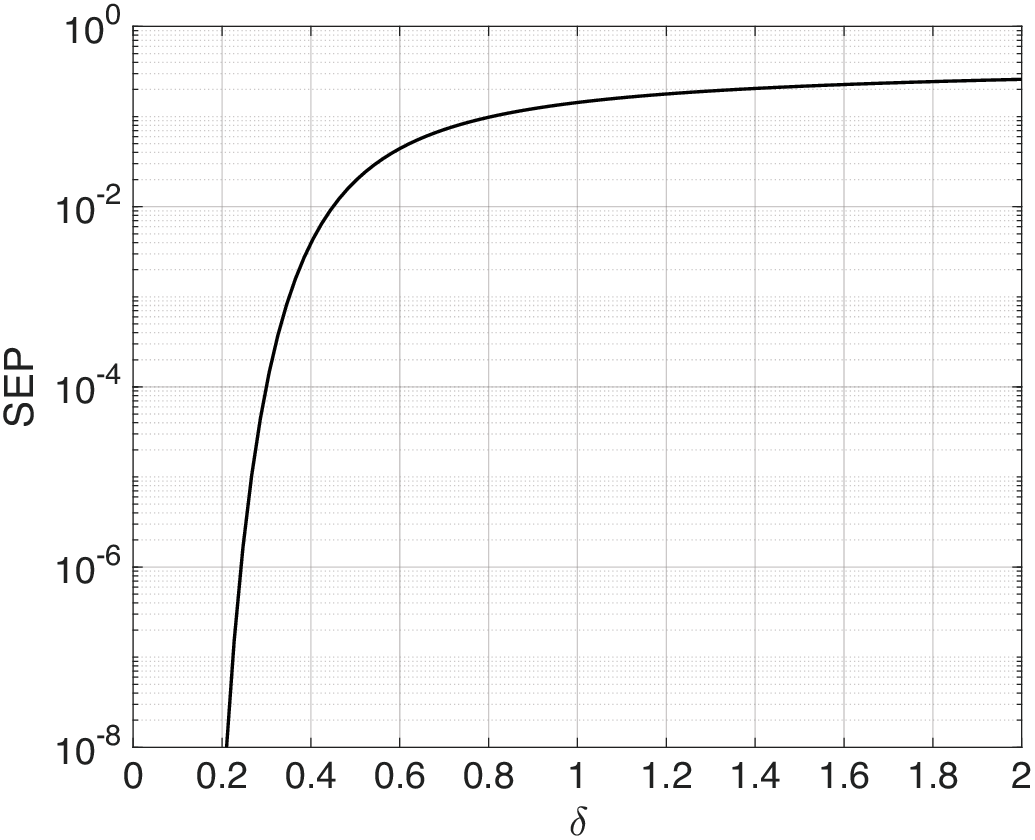}
\centering
\caption{Asymptotic SEP given by Theorem \ref{SEP} versus $\delta$ in the high-SNR limit, i.e., $\sigma=0$. The parameters $(\rho,\lambda)$ are set to their optimal values specified in Theorem \ref{mainresult:lambda_rho}.}
\label{fig:SEP_sigma0}
\end{figure}

\subsection{A Heuristic Derivation of Theorem \ref{mainresult}}\label{subsec:derivation}
In this subsection, we provide a heuristic derivation of our main result in Theorem \ref{mainresult} to illustrate the key insights. A rigorous proof is provided in Section \ref{sec:proof}.

Based on the analytical framework developed in Section \ref{subsec:framework}, analyzing $\H q(\hat{\x})$ transforms to analyzing the asymptotic empirical distribution of the following vector  with $a=a^*$:
\begin{equation}\label{Hxq2}
  \bH\tilde{\x}_{t+1}=\s-\tilde{\z}_{t+1}+\frac{1}{\delta}\left<(q\circ\eta_{a})'(\x_t+\bH^T\z_t;\gamma_{a})\right>\z_t,
  \end{equation}
  where $\eta_{a}(x,\gamma_{a})$ is the denoising function of the AMP algorithm corresponding to problem \eqref{def:xa}; see Lemmas \ref{lemma:unique} and \ref{lemma:astar} for the definitions of $a^*, \gamma_a,$  and $ \eta_a(\cdot;\gamma)$. As discussed in Section \ref{subsec:framework}, $(\x_{\leq t},\tilde{\x}_{t+1})$ and $(\z_{\leq t},\tilde{\z}_{t+1})$ are the resulting sequence of the AMP algorithm with 
  \begin{equation}\label{etat2}
\tilde{\eta}_i(x)=\left\{
\begin{aligned}
&\eta_{a}(x;\gamma_{a}),~~~~~~~\text{if }~i< t;\\%
&q\circ\eta_{a}(x;\gamma_{a}),~~~\text{if }~i=t.
\end{aligned}\right.
\end{equation}
In the following, we characterize the asymptotic empirical distribution of $\bH\tilde{\x}_{t+1}$ using the SE theory of AMP. 

Following the convention of AMP, let ${\bb}_t=\s-{\bz}_t$ and $\tilde{\bb}_{t+1}=\s-\tilde{\bz}_{t+1}$ (see \eqref{AMP_relation}). According to the AMP theory, ${\bb}_t$ and $\tilde{\bb}_{t+1}$ are approximately Gaussian distributed and are independent of $\s$, and 
\begin{equation}\label{productst}
\begin{aligned}
\lim_{N \to \infty} \langle {\bb}_t, \tilde{\bb}_{t+1} \rangle
   &\overset{a.s.}{=}\frac{1}{\delta} \lim_{N \rightarrow \infty} \langle {\x}_t, \tilde{\x}_{t+1}   \rangle,\\
   \lim_{N \to \infty} \langle {\bb}_t, {\bb}_{t} \rangle
   &\overset{a.s.}{=}\frac{1}{\delta} \lim_{N \rightarrow \infty} \langle {\x}_t, {\x}_{t}   \rangle,\\
   \lim_{N \to \infty} \langle \tilde{\bb}_{t+1}, \tilde{\bb}_{t+1} \rangle
   &\overset{a.s.}{=}\frac{1}{\delta} \lim_{N \rightarrow \infty} \langle \tilde{\x}_{t+1}, \tilde{\x}_{t+1}   \rangle;
   \end{aligned}
   \end{equation}
   see Proposition \ref{AMP_property} (ii) and \eqref{AMP_relation}. 
  With the above notations,  \eqref{Hxq2} can be rewritten as 
  \begin{equation}\label{Htildex}
\begin{aligned}
\H \tilde{\x}_{t+1}&=\frac{1}{\delta}\left<(q\circ\eta_a)'({\x}_{t}+\H^T{\bz}_{t};\gamma_a)\right>\s+\underbrace{\tilde{\bb}_{t+1}-\frac{1}{\delta}\left<(q\circ\eta_{a})'({\x}_{t}+\H^T{\bz}_{t};\gamma_a)\right>{\bb}_t}_{\text{Gaussian noise}}.
\end{aligned}
\end{equation}
Applying \eqref{productst}, we have
\begin{equation}
\begin{aligned}
  \frac{1}{K} \left\|\tilde{\bb}_{t + 1}-\frac{1}{\delta}\left<(q\circ\eta_a)'({\x}_{t}+\H^T{\bz}_{t};\gamma_a)\right>\bb_t \right\|^2 \approx  \frac{1}{\delta N}\left\|\tilde{\x}_{t+1}-\frac{1}{\delta}\left< (q\circ\eta_a)' \left({\x}_t + \bH^T{\bz}_t;\gamma_a
  \right) \right>  {\x}_t
  \right\|^2.
  \end{aligned}
\end{equation}
Now let ${\br}_{t+1}= -({\x}_{t} +\bH^T{\bz}_{t})$ (see \eqref{AMP_relation}). Then 
\begin{equation}
\begin{aligned}
{\x}_t&=\eta_a(-{\br}_{t};\gamma_a),\\
\tilde{\x}_{t+1}&=q\circ\eta_a(-{\br}_{t+1};\gamma_a).
\end{aligned}
\end{equation}
It can be shown that  the AMP algorithm with $\eta_t(x)=\eta_a(x;\gamma_a)$ converges, and thus 
$\br_{t+1}\approx\br_t$ for sufficiently large $t$; this follows from Lemma \ref{difference}. Moreover, for large $t$,   $\br_t$ is approximately distributed as $\mathcal{N}(0,\tau_a^2\mathbf{I})$, which 
 can be rigorously established by combining Proposition \ref{pro:4}  with Lemma \ref{R1:1.1}. 
   Based on the above discussions, the scaling factor in front of $\s$ in model \eqref{Htildex} satisfies
    \begin{equation}
   \begin{aligned}
  \frac{1}{\delta}\left<(q\circ\eta_a)'({\x}_{t}+\H^T{\bz}_{t};\gamma_a)\right>&=\frac{1}{\delta}\left<(q\circ\eta_a)'(-{\br}_{t+1};\gamma_a)\right>\\
   &\rightarrow\frac{1}{\delta}\mathbb{E}\left[(q\circ\eta_a)'(\tau_a Z;\gamma_a)\right]\\
   &\overset{(a)}{=}\frac{1}{\delta\tau_a}\mathbb{E}\left[Z\, q\left(\eta_{a}(\tau_aZ;\gamma_a)\right)\right]=\bar{\alpha}_*,
   \end{aligned}
   \end{equation}
   where $Z\sim\mathcal{N}(0,1)$ and (a) follows from the Stein's lemma. The variance of the Gaussian noise in \eqref{Htildex} can be calculated as  
  \begin{equation}
   \begin{aligned}
  & \frac{1}{\delta N}\left\|\tilde{\x}_{t+1}-\frac{1}{\delta}\left< (q\circ\eta_a)' \left({\x}_t + \bH^T{\bz}_t;\gamma_a
  \right) \right>  {\x}_t
  \right\|^2\\
  &= \frac{1}{\delta N}\big\|q\circ\eta_a(-\br_{t+1};\gamma_a)-\frac{1}{\delta}\left< (q\circ\eta_a)' \left(-\br_{t+1};\gamma_a\right) \right>  \eta_a(-\br_t;\gamma_a) \big\|^2\\
&\rightarrow \frac{1}{\delta}\mathbb{E}\left[\left(q\circ\eta_a(\tau_a Z;\gamma_a)-\bar{\alpha}_*\eta_a(\tau_aZ;\gamma_a)\right)^2\right]=\bar{\beta}_*.
\end{aligned}
  \end{equation}
  Combining the above, for sufficiently large $t$, we obtain 
  \begin{equation}
  \bH q(\x_a)\approx\bH \tilde{\x}_{t+1}\approx\bar{\alpha}_*\, \s+\bar{\beta}_*^{\frac{1}{2}}\, \mathbf{z},
  \end{equation}
  where $\mathbf{z}\sim\mathcal{N}(\mathbf{0},\mathbf{I})$ is independent of $\s$. This, together with the fact that $\bH q(\hat{\x})\approx\bH q(\x_a)$, gives the result in Theorem \ref{mainresult}.

     \subsection{Comparison with \cite{ma2025}}\label{subsec:compare}
A recent work \cite{ma2025} also provides an asymptotic performance analysis of nonlinear one-bit precoding, which is closely related to this paper. However, our work differs from \cite{ma2025} in the following three key aspects. 
\begin{itemize}
\item[(i)] \textit{Different optimization models:} In \cite{ma2025}, the authors focus on one-bit box-constrained precoding scheme, where the one-bit transmit signal is obtained by quantizing the solution to the following box-constrained problem: 
\begin{equation}\label{boxprecoder}
\min_{\x\in[-a,a]^N}\,\frac{1}{N}\|\s-\H\x\|^2+\frac{\rho}{N}\|\x\|^2.
\end{equation}
In this formulation, the threshold $a$ is treated as a fixed parameter. As shown in \cite{ma2025}, the performance of the resulting one-bit precoder is highly sensitive to the choice of $a$, which must be specified as a priori and carefully tuned. In practical systems, however, prior knowledge of the optimal threshold $a$ is typically unavailable.
 In contrast, our approach adopts a penalized formulation in \eqref{opt:relaxation} that incorporates the box constraint implicitly through an $\ell_{\infty}^2$-regularization term. These two relaxation models are closely related: for any $\lambda>0$, there exists a corresponding $a(\lambda)$ such that the two problems yield identical solutions, and vice versa. 

\item[(ii)] \textit{Different analytical frameworks:} The analytical tool in \cite{ma2025} is the convex Gaussian min-max theorem (CGMT), where the authors innovatively connect $\bH\text{sgn}(\x_a)$, i.e., the variable under investigation, to the solution of a min-max optimization problem and dedicatedly tailor a variant of CGMT to analyze its distribution. However, the Gaussian process proposed in \cite{ma2025} is tailored explicitly to the  sign function, and its  analytical framework is not readily  generalizable to a broader class of nonlinear functions $q(\cdot)$, as each choice of $q(\cdot)$ typically requires redesigning the associated Gaussian process.   In contrast, our analysis is based on the AMP theory, which provides a more versatile and systematic approach for handling general nonlinear functions.  As discussed in Remark \ref{remark:generaility},  our appraoch provides a general framework for analyzing the effects of nonlinear post-processing applied to convex optimization solutions and thus support broader potential applications, such as analyzing other nonlinear effects at the transmitter sides (e.g.,   multi-bit quantization and hardware impairments).

\item[(iii)] \textit{Optimization of the regularization parameters:} {The performance of the proposed precoder and that of the precoder in \cite{ma2025} is governed by the parameters $(\rho, \lambda)$ in \eqref{opt:relaxation} and $(\rho, a)$ in \eqref{boxprecoder}, respectively. In particular, when $(\rho,\lambda)$ and $(\rho,a)$ are both set to their optimal values, the two convex relaxations achieve identical performance. While \cite{ma2025} explores the impact of $(\rho, a)$ on system performance through numerical simulations, it does not provide an analytical characterization of their optimal values. In contrast, our approach yields closed-form expressions for the optimal $(\rho, \lambda)$ that minimize the SEP, derived from the asymptotic analysis; see Section \ref{sec:4a} for details. Notably, our analysis reveals that the optimal $\rho$ is $\hat{\rho} = 0$, highlighting the optimality of the $\ell_\infty^2$ regularization. We further provide numerical evidence to show that the regularizer $\lambda\|\mathbf{x}\|_\infty^2$ is potentially optimal among a wide class of convex regularization functions; see Section  \ref{Conjecture}.}

\end{itemize}
\section{Optimal Design based on Asymptotic Analysis}\label{subsubsec:parameter}

In this section, we first investigate the optimal selection of regularization parameters $(\rho,\lambda)$ for the relaxation model in \eqref{opt:relaxation} by leveraging the asymptotic analysis in Section \ref{sec:4a}. Then, in Section \ref{Conjecture}, we go beyond the current model and discuss the optimal design under a more general convex relaxation framework. We specifically focus on the case of one-bit quantization, namely, $q(\cdot)=\text{sgn}(\cdot)$, {with BPSK constellation }in this section. 

\subsection{Optimization of the Regularization Parameters}\label{sec:4a}
As shown in Theorem \ref{SEP}, the asymptotic SEP depends only on the system parameters (i.e., the user-antenna ratio $\delta$ and the variance of channel noise $\sigma^2$) and the regularization parameters $(\rho,\lambda)$ in model \eqref{opt:relaxation}. Therefore, for a given system, $(\rho,\lambda)$ can be optimized to minimize the  SEP by solving the following optimization problem:

\begin{equation}\label{opt:parameter}
\begin{aligned}
\min_{\rho\geq 0, \,\lambda\geq 0}\,&Q\left(\sqrt{\overline{\text{SNR}}}\right),
\end{aligned}
\end{equation}
where $\overline{\text{SNR}}$ is defined in \eqref{barSNR}.
\begin{remark}
Rigorously speaking, our analysis builds on the assumption that both $\rho$ and $\lambda$ are  positive. In problem \eqref{opt:parameter}, however, we consider the closed feasible set $\{(\rho,\lambda)\mid\rho\geq 0,\,\lambda\geq 0\}$ to ensure that the optimal solution is attainable. We remark that our asymptotic analysis can be directly extended to  the case $\lambda=0$, which corresponds to $a^*=\infty$. 
However,  for technical reasons, we cannot remove the positive assumption on $\rho$. Nevertheless, numerical results suggest that our analytical result still holds true for $\rho=0$; see Fig. \ref{SEPb}. 
\end{remark}
The following theorem gives the optimal solution to \eqref{opt:parameter}.
{
\begin{theorem}[Optimal Regularization Parameters for {One-Bit Precoding}]\label{mainresult:lambda_rho}
Let 
$\overline{\text{\normalfont SNR}}$ be defined as in \eqref{barSNR}. The optimal solution to \eqref{opt:parameter} is given by
\begin{equation}
\hat{\rho}=0
\end{equation}
and 
\begin{equation}
\hat{\lambda}=\frac{2\big(1-\frac{2\Phi(\hat{z})-1}{\delta}\big)\bigg(\phi(\hat{z})-\hat{a}\hat{\tau}(1-\frac{2\Phi(\hat{z})-1}{\delta})(1-\Phi(\hat{z}))\bigg)}{\hat{z}\left(1+\frac{2\hat{a}^2}{\delta}(1-\Phi(\hat{z}))\right)},
\end{equation}
where $\hat{a}=\left({\hat{z}^{-2}}-h(\hat{z})\right)^{-\frac{1}{2}}$ and $\hat{\tau}=(1-\hat{z}^{2}h(\hat{z}))^{-\frac{1}{2}}$. Here, $\hat{z}$ is the unique solution to $\zeta(z)=0$ in $(0,z_0)$, where 

\begin{equation}\label{def:zetaz}
\begin{aligned}
\zeta(z)=&\sqrt{\frac{\pi}{2}}(2\Phi(z)-1)+\sqrt{2\pi}z^2\bigg(\Phi(z)+\frac{\phi(z)}{z}-1\bigg)+\left(\frac{\pi}{2}\delta(1+\sigma^2\delta)-1\right)z-\sqrt{\frac{\pi}{2}}\delta;
\end{aligned}
\end{equation}
$z_0$ is defined as the unique positive solution\footnote{The uniqueness of $z_0$ and $\hat{z}$ is proved in Appendix \ref{proof:lambda_rho}.} to $z^2h(z)-1=0$ when $\delta\in(0,1)$, 
where 
\begin{equation}\label{def:hz}
h(z)=\frac{2\Phi(z)-1}{\delta z^2}-\frac{2}{\delta}\left(\frac{\phi(z)}{z}+\Phi(z)-1\right), 
\end{equation}
and  $z_0=\infty$ when $\delta\in[1,\infty)$. 
\end{theorem}}
\begin{proof}
See Appendix \ref{proof:lambda_rho}. 
\end{proof}
\begin{figure}
\subfigure[Theoretical and numerical SEP versus $\lambda$, where $\rho=0$.]{\includegraphics[width=0.43\textwidth]{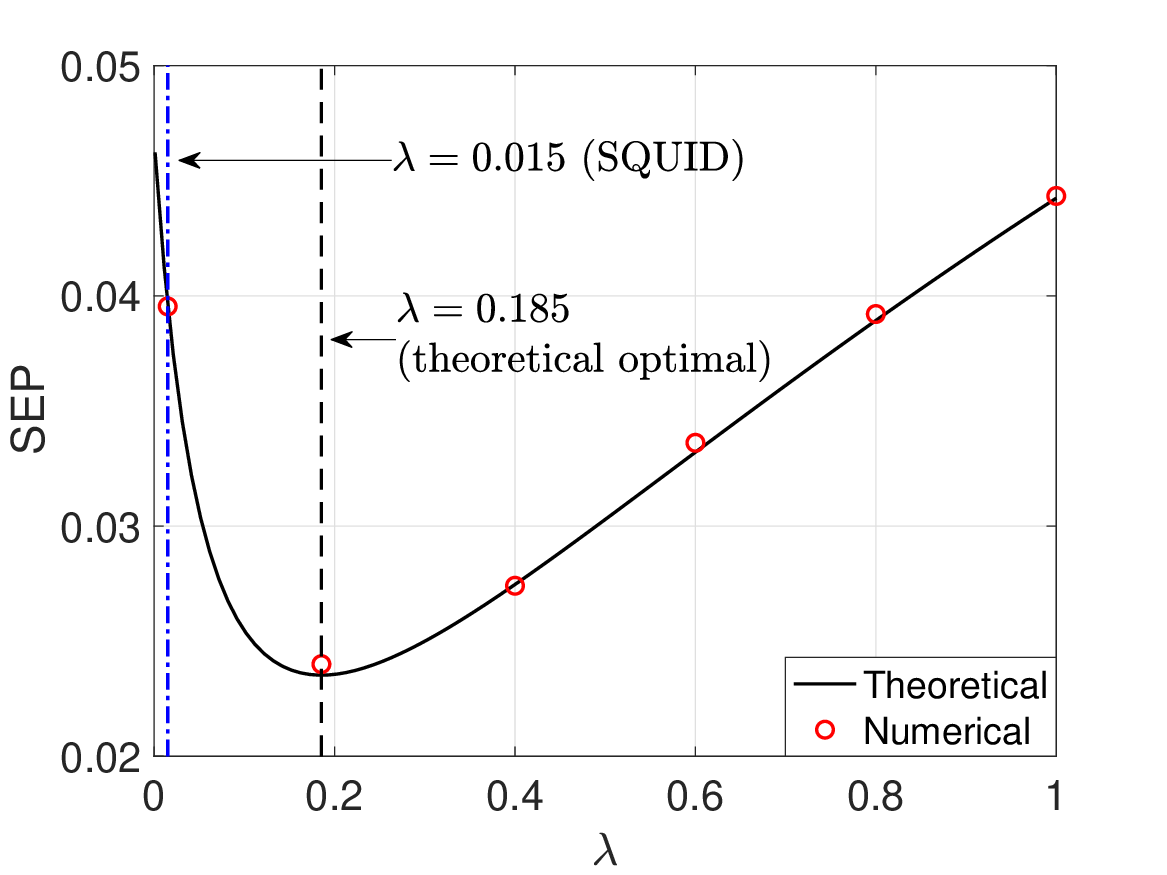}\label{lambda}}
\subfigure[Theoretical and numerical SEP versus $\rho$, where $\lambda=0.185$.]{\includegraphics[width=0.43\textwidth]{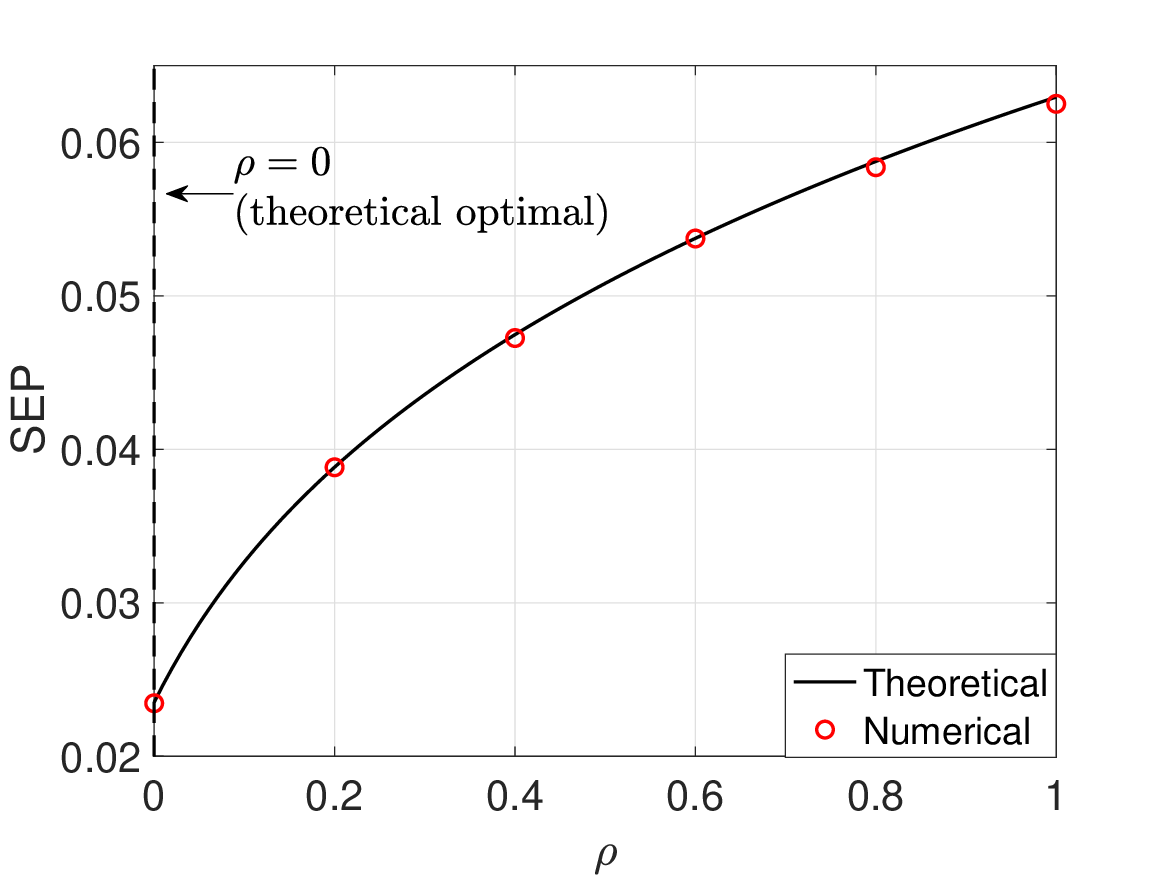}\label{rho}}
\centering
\caption{Theoretical and numerical SEP versus the regularization parameters $(\rho,\lambda)$, where  $N=128$, $K=64$ (i.e., $\delta=0.5)$, and SNR$=15$\,dB.The numerical results are averaged over $10^4$ channel realizations.}
\label{fig_parameter}
\end{figure}

In Fig. \ref{fig_parameter}, we {demonstrate} the optimal parameters given in Theorem \ref{mainresult:lambda_rho}. Specifically, in Fig. \ref{lambda}, we fix $\rho$ to its optimal value and depict the SEP versus $\lambda$; in Fig. \ref{rho}, we fix $\lambda$ to its optimal value and vary $\rho$. We also mark the regularization parameters of the SQUID precoder in Fig. \ref{lambda} for comparison.  As expected, the theoretical SEP gives an accurate prediction of the actual SEP as the regularization parameters vary. In addition, the optimal parameters specified in Theorem \ref{mainresult:lambda_rho} yield the minimum SEP and demonstrate a performance gain over the classical SQUID precoder.

\begin{figure}
\includegraphics[width=0.48\textwidth]{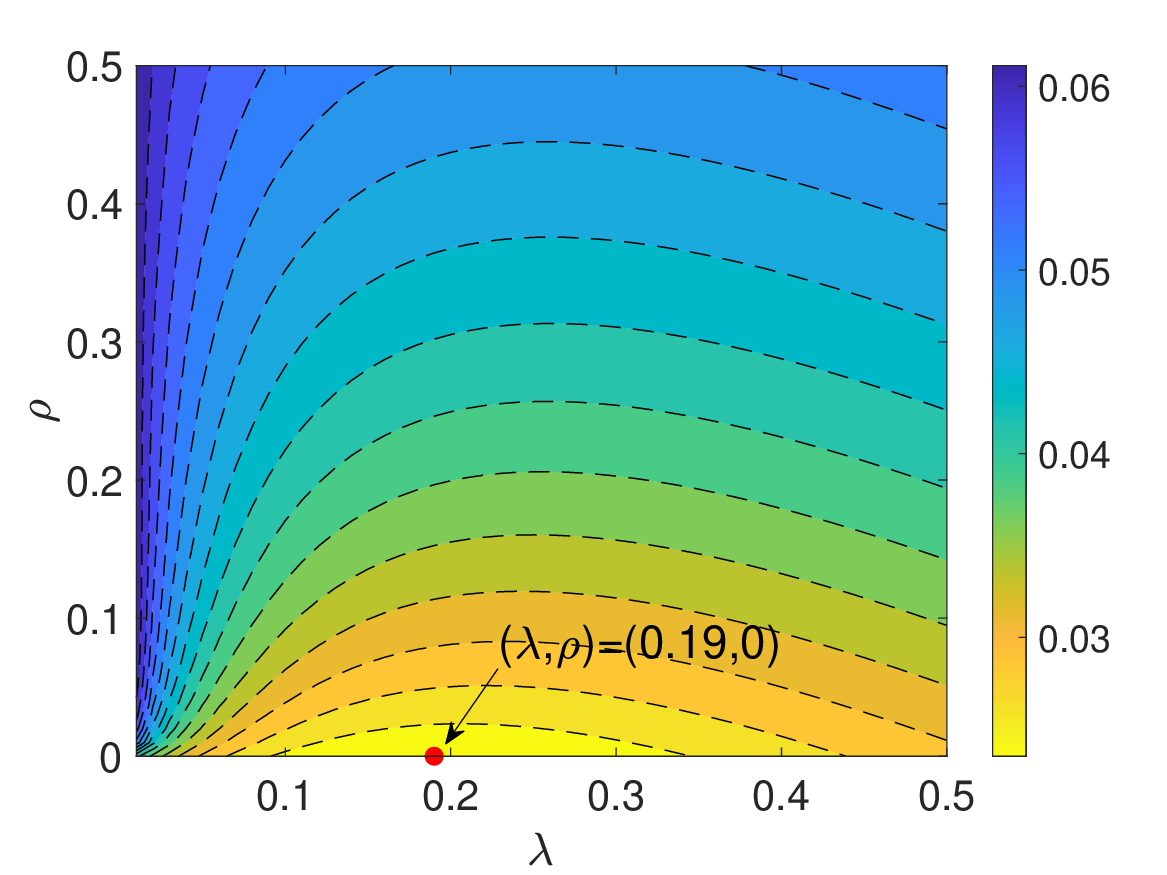}
\centering
\caption{Contour plot of the theoretical SEP   over the $(\rho,\lambda)$ space, where $\delta=0.5$ and SNR$=15$\,dB. The plot is generated by grid search with a step size of $0.01$.}
\label{J_parameter_2dim}
\end{figure}
\begin{figure}
\includegraphics[width=0.42\textwidth]{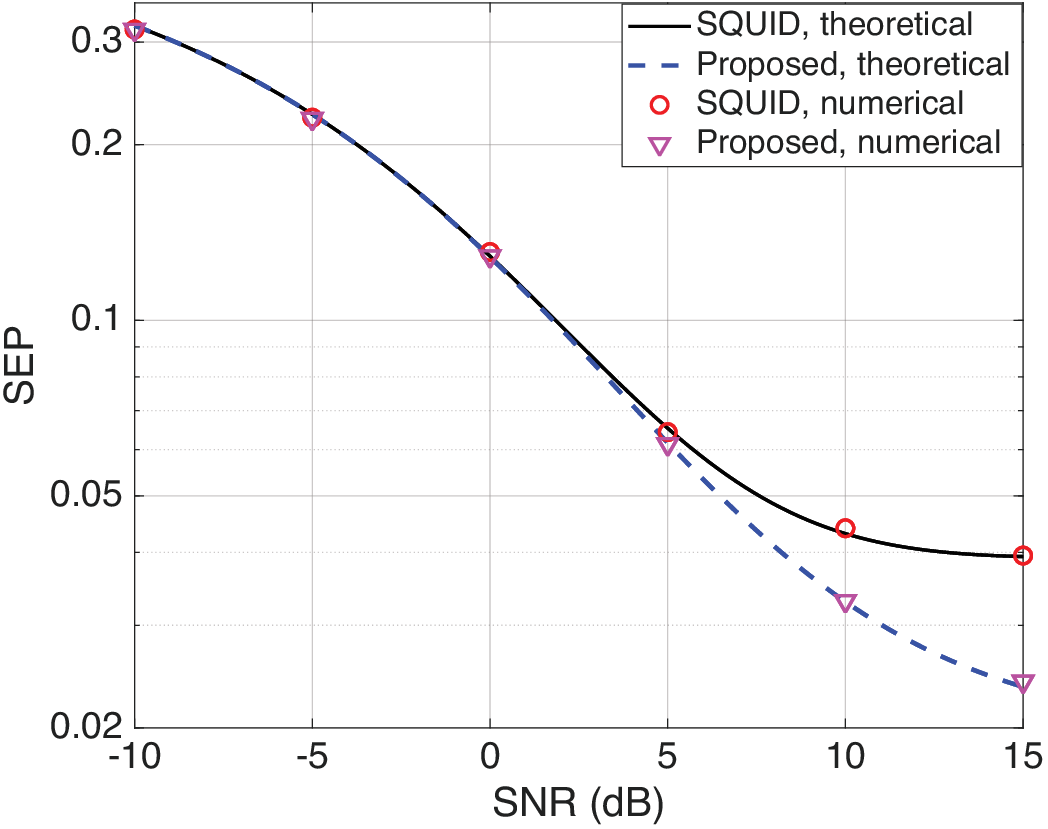}
\centering
\caption{{Theoretical and numerical SEP of SQUID and the proposed precoder (with parameters specified by  Theorem \ref{mainresult:lambda_rho}) versus the SNR, where $N=128$ and $K=64$ (i.e., $\delta=0.5$).  The numerical results are averaged over $10^4$ channel realizations.}}
\label{SQUID_SNR}
\end{figure}

 {For a more complete view, Fig. \ref{J_parameter_2dim} depicts the contour plot of the theoretical SEP over the two-dimensional $(\rho,\lambda)$ parameter space. As shown in the figure, the theoretical SEP is minimized exactly at the optimal parameters given in Theorem \ref{mainresult:lambda_rho}, further validating our analysis.}

In Fig. \ref{SQUID_SNR}, we further compare  SQUID with the proposed optimal precoder, where $\rho$ and $\lambda$ are set to their optimal values according to Theorem \ref{mainresult:lambda_rho}, over a range of SNRs. It can be observed that SQUID achieves near-optimal performance in the low-SNR regime. As the SNR increases, the proposed precoder yields increasingly larger performance gains.

\subsection{Conjecture on the Optimal Regularizer}\label{Conjecture}
Theorem  \ref{mainresult:lambda_rho} gives the optimal regularization parameters for the  convex relaxation model in  \eqref{opt:relaxation} {in the context of one-bit precoding}.  This naturally leads to a further question: can we identify the optimal form of the regularization function?
 Answering this  question is  important  for both theoretical understanding and algorithm design of one-bit precoding. We provide some preliminary discussions in the following. 

{Consider the following one-bit precoding scheme:  the transmit signal is 
\begin{equation}
\mathbf{x}_T = \operatorname{sgn}(\hat{\mathbf{x}}),
\end{equation}
where $\hat{\mathbf{x}}$ is the solution to a convex relaxation of the discrete optimization problem in \eqref{opt:discrete}, specifically:  
\begin{equation}
\hat{\x}=\arg\min_{\mathbf{x}}~\frac{1}{N}\|\mathbf{s} - \mathbf{H}\mathbf{x}\|^2 + R(\mathbf{x}).
\end{equation}
Here, $R(\mathbf{x})$ is a convex regularization function that plays a central role in shaping the performance of the precoding scheme. Identifying an optimal choice of $R(\mathbf{x})$ is thus of fundamental importance in the design of one-bit precoders. Our main result, Theorem~\ref{mainresult:lambda_rho}, establishes that the regularizer $R(\mathbf{x}) = \hat{\lambda}\|\mathbf{x}\|_\infty^2$ is optimal within the class of functions of the form $R(\mathbf{x}) = \rho \frac{\|\mathbf{x}\|^2}{N} + \lambda\|\mathbf{x}\|_\infty^2$. To explore the generality of this result, we further evaluate in Fig. \ref{fig_l1l4} the SEP performance of two extended regularizers of the form $R(\mathbf{x}) = \rho\, r(\mathbf{x}) + \lambda\|\mathbf{x}\|_\infty^2$, where $r(\mathbf{x})=\frac{1}{N}\|\mathbf{x}\|_4^4$ and $r(\mathbf{x})=\frac{1}{N}\sum_{i=1}^N h_\mu(x_i)$, with $h_\mu(x)$ denoting the Huber function with parameter $\mu$. 

Interestingly, for both cases, the SEP is minimized at $\rho=0$. Moreover, this observation remains consistent across all tested configurations of $\sigma>0$ and $\delta>0$. These findings support the conjecture that $R(\mathbf{x}) = \hat{\lambda}\|\mathbf{x}\|_\infty^2$ remains optimal over a considerably broader class of convex regularizers satisfying symmetry and (asymptotic) separability conditions. A rigorous theoretical justification of this conjecture is left for future work.}

 \begin{figure}
\subfigure[$r(\x)=\frac{1}{N}\sum_{i=1}^Nh_\mu(x_i)$,~$\mu=0.1$.]{\includegraphics[width=0.45\textwidth]{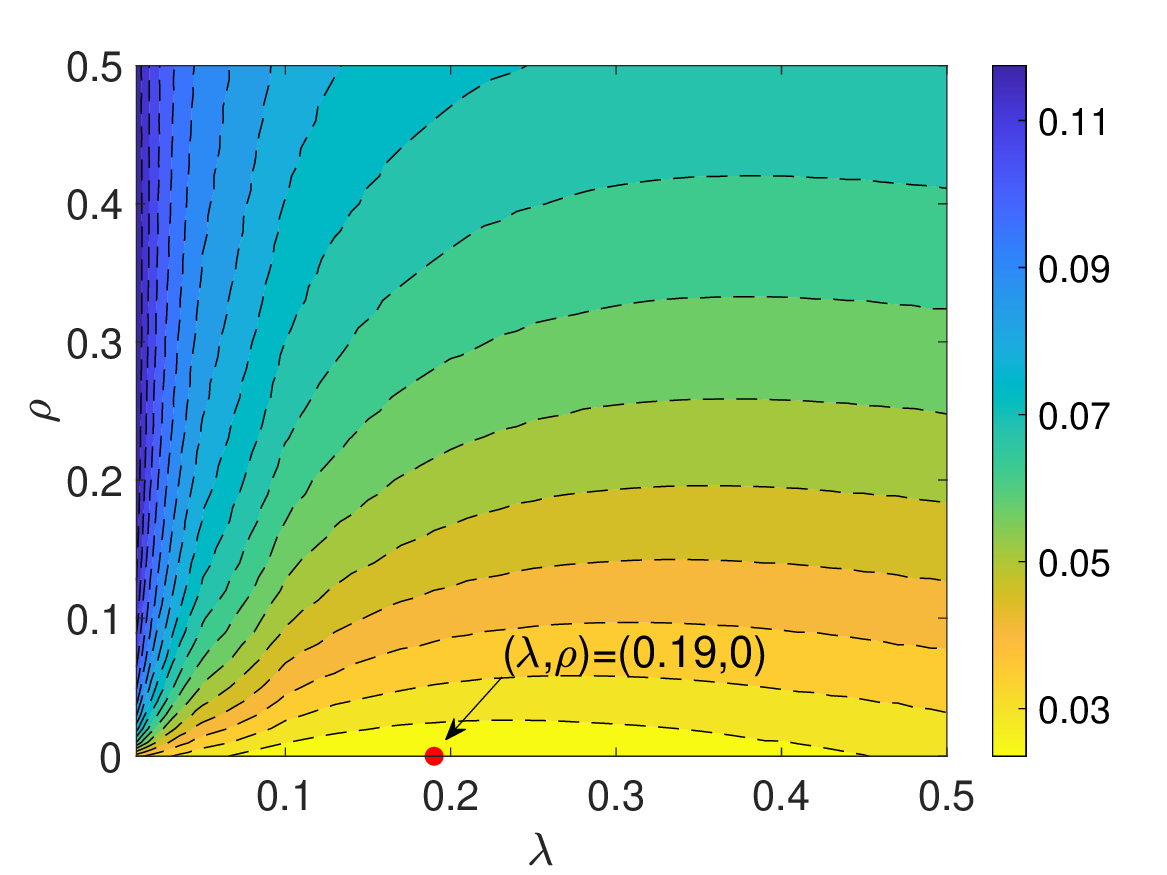}\label{J_parameter_2dim_huber}
}
\subfigure[$r(\x)=\frac{1}{N}\|\x\|_4^4$.]{\includegraphics[width=0.45\textwidth]{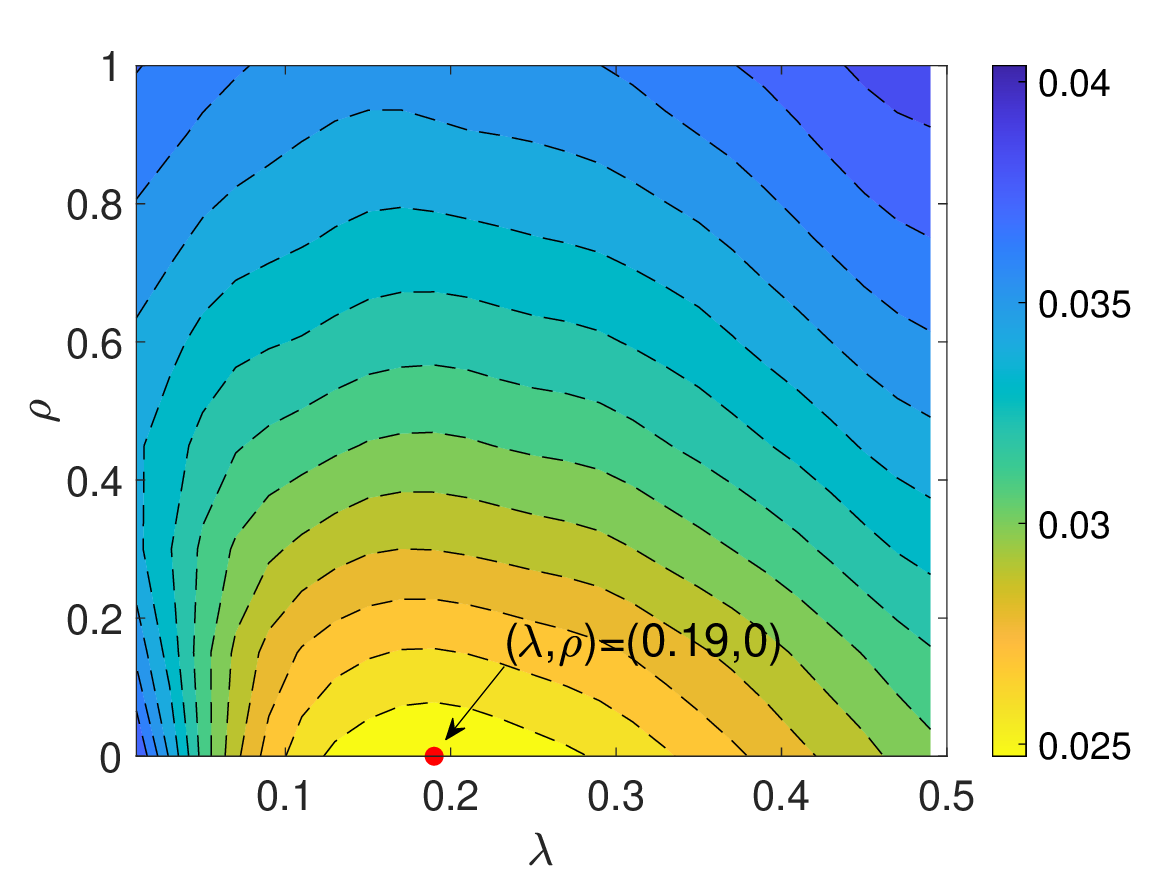}\label{J_parameter_2dim_l4}}
\centering
\caption{Contour plot of the SEP over the $(\rho,\lambda)$ space for regularizers of the form $R(\x)={\rho}\,r(\x)+\lambda\|\x\|_\infty^2$, where $\delta=0.5$ and SNR$=15$\,dB.  Fig. \ref{J_parameter_2dim_huber} shows the theoretical asymptotic SEP for the Huber-regularized model, which is obtained via the state evolution theory of AMP (which closely matches empirical performances). The plot is generated by grid search with a step size of $0.01$ for both $\lambda$- and $\rho$-axes. Fig. \ref{J_parameter_2dim_l4} shows the empirical SEP for the $\ell_4$-regularized model,  averaged over $10^4$ channel realizations with $N=128$; here, the state evolution is not included due to the lack of a closed-form proximal operator (although a numerical evaluation is possible). The grid step is $0.01$ for $\lambda$-axis and $0.2$ for $\rho$-axis.}
\label{fig_l1l4}
\end{figure}

   \section{Proof of Theorem \ref{mainresult}}\label{sec:proof}
In this section, we give the proof of our main result, namely, Theorem \ref{mainresult}. As discussed in Section \ref{subsec:framework}, our proof contains two main steps. In Section \ref{sec:auxiliary}, we introduce an auxiliary problem to facilitate the analysis. Then we analyze the auxiliary problem in Section \ref{sec:AMP}.
\subsection{Auxiliary Problem}\label{sec:auxiliary}
Building on the analytical framework discussed in Section \ref{subsec:framework}, we first analyze the AMP algorithm corresponding to the box-constrained problem in \eqref{def:xa}, {defined as follows:}
\begin{equation}\label{AMP_copy}
\begin{aligned}
\x_{t+1}&=\eta_a(\x_{t}+\H^T\bz_t;\gamma_a),\\
\bz_t&=\s-\H\x_t+\frac{1}{\delta}\left<\eta_a'(\x_{t-1}+\H^T\bz_{t-1};\gamma_a)\right>\bz_{t-1},
\end{aligned}
\end{equation}
where $\eta_a(x;\gamma)$ is given in Lemma \ref{lemma:unique}, $\gamma_a$ is a solution\footnote{In Appendix \ref{app:exist}, we have shown the existence of $\gamma_a$,  and thus the AMP algorithm in \eqref{AMP_copy} is well-defined. We remark that till now, the uniqueness of $\gamma_a$ has not yet be proven.} to \eqref{condition0},  {and the initial condition is set as $\x_0=\mathbf{0}$ (any variable with negative index is also assumed to be zero).} 
Following a  procedure  similar to that  in \cite{LASSO, elasticnet}, we establish the following result.
\begin{proposition}\label{prob:converge}
The following results hold under Assumption \ref{ass}  (i) -- (iii):
\begin{enumerate}
\item[(i)] The sequence $\{\x_t\}_{t\geq 0}$ generated by the AMP algorithm in \eqref{AMP_copy} satisfies 
\begin{equation}
\lim_{t\to\infty}\lim_{N\to\infty}\frac{\|\x_t-\x_a\|^2}{N}\overset{a.s.}{=}0,
\end{equation}
 where $\x_a$ {denotes the solution to the box-constrained optimization problem in} \eqref{def:xa}.
\item[(ii)] For any pseudo-Lipschitz function $\varphi:\R\rightarrow \R$, 
\begin{equation}
\lim_{N\to\infty}\frac{1}{N}\sum_{i=1}^N\varphi(x_{a,i})\overset{a.s.}{=}\mathbb{E}\left[\varphi(X_a)\right],
\end{equation}
where $x_{a,i}$ denotes the $i$-th element of $\x_a$, 
\begin{equation}\label{def:Xa}
X_a=\eta_a(\tau_aZ;\gamma_a), ~Z\sim\mathcal{N}(0,1),
\end{equation}
and $(\tau_a^2,\gamma_a)$ is a solution to \eqref{condition0}.
\end{enumerate}
\end{proposition}
\begin{remark}
{The parameter $\gamma_a$ is calibrated with $\rho$ via the fixed-point equation in \eqref{condition0}, which is derived by linking the limiting point of AMP and the first-order optimality condition of \eqref{def:xa}; see discussions in Appendix \ref{appD:heuristic}. Hence, the impact of $\rho$ on the optimization problem is entirely mediated through $\gamma_a$. We have also leveraged this property in our parameter optimization analysis in Section \ref{subsubsec:parameter}.}
\end{remark}
{The proof of Proposition~\ref{prob:converge}, provided in Appendix~\ref{appendix:proofofconverge}, builds on techniques developed in \cite{LASSO, elasticnet}, but differs in that the optimization problem considered here is constrained, whereas those in \cite{LASSO, elasticnet} are unconstrained. The analysis relies heavily on the AMP framework. For completeness, a brief overview of the AMP algorithm is included in Appendix~\ref{appendix:AMP}.}

As a corollary of Proposition \ref{prob:converge}, we can now prove the uniqueness of the solution $(\tau_a^2,\gamma_a)$ to \eqref{condition0}; see Appendix \ref{app:unique}. In addition, we can establish the convergence of  $f_N(a;\s,\bH)$,  i.e., the optimal value of the  box-constrained problem defined in \eqref{def:fN}.
\begin{lemma}\label{convergefN}
Under Assumption \ref{ass}  (i) -- (iii), the function $f_N(a;\s,\H)$ defined in \eqref{def:fN} satisfies 
\begin{equation}
\lim_{N\to\infty} f_{N}(a;\s,\H)\overset{a.s.}{=}\bar{f}({a}):=\delta\rho(\tau_a^2-1)+\frac{\delta\rho^2\tau_a^2}{\gamma_a^2},
\end{equation}
where  $(\tau_a^2,\gamma_a)$ is the unique solution to \eqref{condition0}.
\end{lemma}
\begin{proof}
See Appendix \ref{app:convergefN}.
\end{proof}
Lemma \ref{convergefN} suggests the convergence of the objective function in problem \eqref{relaxation:a}, i.e., $f_N(a;\s,\H)+\lambda a^2\xrightarrow{a.s.} f(a)$, where $f(a)$ is given in Theorem \ref{mainresult}. This further   implies the convergence of the corresponding optimal solution, i.e., $\hat{a}_N\xrightarrow{a.s.} a^*$.   
Since the solution to \eqref{opt:relaxation} can be represented by $\hat{\x}=\x_{\hat{a}_N}$, where $\x_{a}$ is defined in \eqref{def:xa},
it follows that $\hat{\x}$ asymptotically approaches $\x_{a^*}$.  These results are summarized in the following lemma.
\begin{lemma}\label{converge:a}
 Under Assumption \ref{ass}  (i) -- (iii),  the following result holds for $\hat{a}_N$ defined in \eqref{def:hata} and $a^*$ given in Lemma \ref{mainresult}: 
 \begin{equation}
\displaystyle\lim_{N\to\infty}\hat{a}_N\overset{a.s.}{=}a^*.
\end{equation}
 Furthermore, 
\begin{equation}
\displaystyle\lim_{N\to\infty}\frac{1}{\sqrt{N}}{\|\hat{\x}-\x_{a^*}\|}\overset{a.s.}{=}0,
\end{equation}
 where $\hat{\x}$ is the solution to \eqref{opt:relaxation} and $\x_{a^*}$ is the solution to the box-constrained problem in \eqref{def:xa} with $a=a^*$.
\end{lemma}
\begin{proof}
See Appendix \ref{proof:Lemma4}.
\end{proof}

Combining the above, we now present the main result for Step 1.
\begin{proposition}\label{the1}
Under Assumption \ref{ass}, the following result holds for any pseudo-Lipschitz function $\psi:\R^2\rightarrow \R$:
\begin{equation}
\lim_{K\to\infty}\bigg|\frac{1}{K}\sum_{k=1}^K\psi(\h_k^Tq(\hat{\x}),s_k)-\frac{1}{K}\sum_{k=1}^K\psi(\h_k^Tq({\x}_{a^*}),s_k)\bigg|\hspace{-0.03cm}\overset{a.s.}{=}0,
\end{equation}
where $\hat{\x}$ is the solution to \eqref{opt:relaxation} and $\x_{a^*}$ is the solution to the box-constrained problem in \eqref{def:xa} with $a=a^*$.

\begin{proof}
See Appendix \ref{proof:1B}.
\end{proof}
\end{proposition}
\subsection{Analysis of the Auxiliary Problem}\label{sec:AMP}
Based on Proposition \ref{the1},  our remaining task is to analyze  
\begin{equation}\label{desired}
\lim_{K\to\infty}\frac{1}{K}\sum_{k=1}^K \psi(\h_k^Tq(\x_a), s_k),
\end{equation}
where $a=a^*$. 
Throughout the rest of the analysis, we  omit the superscript ``$*$'' on $a$ for notational simplicity. 
As discussed in Section \ref{subsec:framework}, we can construct  a post-processing  step as in \eqref{xqzq} to  characterize the above quantity via the state evolution theory. However, rigorously speaking, the state evolution analysis requires the denoising function to be Lipschitz continuous,  
which does not hold for the composite function $q\circ\eta_a(\cdot;\gamma_a)$ if $q(\cdot)$ is discontinuous. To address this issue, we apply a standard smoothing technique to  approximate $q(\cdot)$  with a smooth function $q_\epsilon(\cdot)$, which is obtained by convolving $q(\cdot)$ with a scaled mollifier. Please see Appendix \ref{app:smoothing} for an introduction of the smoothing technique; in particular, the definition of $q_\epsilon(\cdot)$ is given in Definition \ref{def:qeps}.

Given $\epsilon>0$, we construct the following post-processing steps:
\begin{equation}\label{AMP2}
\hspace{-0.1cm}\begin{aligned}
\tilde{\x}_{t+1}&=q_\epsilon(\x_{t+1}),\\
\tilde{\bz}_{t+1}&=\s-\H\tilde{\x}_{t+1}\hspace{-0.05cm}+\hspace{-0.05cm}\frac{1}{\delta}\left<(q_\epsilon\circ\eta_a)'({\x}_{t}+\H^T{\bz}_{t};\gamma_a)\right>{\bz}_{t}.
\end{aligned}
\end{equation}
Using the update rule of $\x_{t+1}$ in \eqref{AMP_copy}, we can express $\tilde{\x}_{t+1}$ as 
\begin{equation}\tilde{\x}_{t+1}=q_\epsilon\circ \eta_a(\x_t+\bH^T\bz_t;\gamma_a).\end{equation}
Hence, $(\tilde{\x}_{t+1},\tilde{\bz}_{t+1})$  can be viewed as the result of performing one additional AMP iteration on  $(\x_t,\bz_t)$, using the denoising function $\eta_t(x)=q_\epsilon\circ\eta_a(x,\gamma_a)$. We note that $q_{\epsilon}\circ\eta_a(\cdot;\gamma_a)$ is Lipschitz continuous; see Lemma \ref{lemma:lip} for a rigorous proof. 

Following the same procedure as in Section \ref{subsec:derivation}, we can analyze the distribution of $\bH \tilde{\x}_{t+1}$. 
 Specifically, analogous to  \eqref{Htildex}, we have 
\begin{equation}
\H\tilde{\x}_{t+1}=\tilde{\alpha}_t\s-\tilde{\alpha}_t{\bb}_{t}+\tilde{\bb}_{t+1},
\end{equation}
where ${\bb}_t=\s-{\z}_t$, $\tilde{\bb}_{t+1}=\s-\tilde{\z}_{t+1}$, and $\tilde{\alpha}_t=\frac{1}{\delta}\left<(q_\epsilon\circ\eta_a)'({\x}_{t}+\H^T{\bz}_{t};\gamma_a)\right>$.
 Hence, 
 \begin{equation}
 \frac{1}{K}\sum_{k=1}^K\psi(\h_k^T\tilde{\x}_{t+1},s_k)
 =\frac{1}{K}\sum_{k=1}^K\psi(\tilde{\alpha}_ts_k-\tilde{\alpha}_t{b}_{t,k}+\tilde{b}_{t+1,k},s_k).
 \end{equation}
In the following proposition, we characterize the convergence of the above quantity.
\begin{proposition}\label{pro:3}
Under Assumption \ref{ass}, the following result holds {for any pseudo-Lipschitz  function $\psi:\R^2\rightarrow \R$:}
 \begin{equation}
 \lim_{t\to\infty}\lim_{K\to\infty}\frac{1}{K}\sum_{k=1}^K\psi\left(\h_k^T\tilde{\x}_{t+1},s_k\right)\overset{a.s.}{=}\mathbb{E}[\psi(\tilde{\alpha}_* S+\tilde{\beta}_*^{\frac{1}{2}} Z, S)],
  \end{equation}
 where 
 \begin{equation}
\begin{aligned}
\tilde{\alpha}_*&=\frac{1}{\delta}\mathbb{E}\left[q_{\epsilon}'(\eta_a(\tau_aZ;\gamma_a))\eta_a'(\tau_aZ;\gamma_a)\right],\\
\tilde{\beta}_*&=\frac{1}{\delta}\mathbb{E}\left[(\tilde{\alpha}_*X_a-q(X_a))^2\right],
\end{aligned}
 \end{equation}
  ${S}\sim\text{\normalfont{Unif}}({\mathcal{S}})$ and  $Z\sim\mathcal{N}(0,1)$  are independent, and  $X_a$ is defined in \eqref{def:Xa}.
  \end{proposition}
\begin{proof}
See Appendix \ref{app:pro3}.
\end{proof}
By letting $\epsilon\to 0$ in the above result and noting that ${\x}_{t+1}$ converges to $\x_a$ as shown in Proposition \ref{prob:converge} (i), we get the following result.
\begin{proposition}\label{pro:convergexa}
 Under Assumption \ref{ass}, the following result holds for {any pseudo-Lipschitz  function $\psi:\R^2\rightarrow \R$}:
 \begin{equation}
\lim_{K\to\infty}\frac{1}{K}\sum_{k=1}^K \psi(\h_k^Tq(\x_a), s_k)=\mathbb{E}[\psi(\bar{\alpha}_*S+\bar{\beta}_*^{\frac{1}{2}}Z,S)],
 \end{equation}
where  $(\bar{\alpha}_*,\bar{\beta}_*)$ are defined in Theorem \ref{mainresult}, ${S}\sim\text{\normalfont{Unif}}({\mathcal{S}})$ and  $Z\sim\mathcal{N}(0,1)$  are independent.

 \end{proposition}

\begin{proof}
See Appendix \ref{app:convergexa}.
\end{proof}
Finally, Theorem \ref{mainresult} is obtained by combining  Proposition \ref{the1} and Proposition \ref{pro:convergexa}.

 \section{Conclusions and Future Directions}\label{sec:conclusion}

In this paper, we investigated the performance of nonlinear one-bit precoding schemes in massive MIMO systems. We developed a rigorous AMP-based analytical framework to precisely characterize the asymptotic performance of the CRQ precoder defined in Definition \ref{def:onebit}, under the assumptions of a real-valued system and an i.i.d. Gaussian channel. This analytical framework enables us to identify optimal regularization parameters that minimize the asymptotic SEP. While our study primarily focuses on one-bit quantization, the developed framework is general and can readily accommodate other nonlinear transmitter-side effects.

We conclude by highlighting several promising directions for future research:
\begin{itemize}
  \item[(1)] \textbf{Optimality of $\ell_\infty^2$ regularization:} For one-bit precoding, numerical results in Section \ref{Conjecture} strongly suggest that $R(\x)=\hat{\lambda}\|\x\|_\infty^2$ is the optimal regularizer for the relaxed MMSE model among a broad range of convex regularization functions. In Theorem \ref{mainresult:lambda_rho}, we rigorously established the optimality of the $\ell_\infty^2$ regularizer within the mixed $\ell_\infty^2$–$\ell_2^2$ norm context. Extending this optimality result to an even broader class of regularization functions represents an interesting future direction.

  \item[(2)] \textbf{General channel matrices:} Our analysis assumes i.i.d. Gaussian entries for the channel matrix. Relaxing this assumption to consider more general, practically relevant channel models is of significant interest. Advanced AMP variants developed for general design matrices \cite{ma2017orthogonal,rangan2019vector,fan2022approximate,liu2024unifying,zhao2022asymptotic,dudeja2023universality,wang2024universality} could provide useful tools in this direction.

\item[(3)] \textbf{Complex-valued systems}: 
The current analysis focuses on a real-valued system model. 
Extending the proposed analytical framework to complex-valued systems, which are more relevant in practical wireless communications, is an important and interesting direction for future work. 

 \item[(4)] \textbf{Other precoding schemes:} 
Beyond the considered vector-valued relaxation models, the semidefinite relaxation (SDR) technique can also be applied to \eqref{Eqn:formulation2}, which is another important CRQ-based approach. However, the SDR formulation involves matrix-valued optimization variables and constraints, and is therefore not amenable to the proposed AMP-based analytical framework. In addition to the CRQ schemes, more advanced nonconvex optimization-based methods have been proposed \cite{GEMM,PBB,onebit_CI}. Developing analytical tools to characterize the asymptotic performance of these sophisticated approaches is a promising direction for future research.
 
\end{itemize}

\appendices

\section{Proof of Lemma \ref{lemma:unique} and Lemma \ref{lemma:astar}}\label{app:lemma1} 
In this appendix, we prove Lemmas \ref{lemma:unique} and \ref{lemma:astar}. To begin, we rewrite \eqref{condition0} into the following form by explicitly calculating the expectation and probability involved in \eqref{condition0}:\begin{subequations}\label{condition0:2}
\begin{align}
F_1(\tau^2,\gamma;a)&:=\tau^2-1-\frac{1}{\delta}\mathbb{E}[\eta_a^2(\tau Z;\gamma)]\notag\\
&=\tau^2\hspace{-0.05cm}-\hspace{-0.05cm}1\hspace{-0.05cm}-\hspace{-0.1cm}\frac{\tau^2}{\delta(\gamma+1)^2}u(\tau^2,\gamma,a)+\frac{a^2}{\delta}u(\tau^2,\gamma,a)+\frac{2a\sqrt{\tau^2}}{\delta(\gamma+1)}v(\tau^2,\gamma,a)-\frac{a^2}{\delta}=0,\label{condition0:2a}\\
F_2(\tau^2,\gamma; a)&:=\rho-\gamma\left(1-\frac{1}{\delta(\gamma+1)}\mathbb{P}\left(\frac{\tau Z}{\gamma+1}\in[-a,a]\right)\right)=\rho-\gamma+\frac{\gamma}{\delta(\gamma+1)}u(\tau^2,\gamma,a)=0,\label{condition0:2b}
\end{align}
\end{subequations}
where $u(\tau^2,\gamma,a)$ and $u(\tau^2,\gamma,a)$ are defined by 
\begin{equation}\label{def:uv}
\begin{aligned}
u(\tau^2,\gamma,a)&=2\Phi\left(\frac{a(\gamma+1)}{\sqrt{\tau^2}}\right)-1,~v(\tau^2,\gamma,a)&=\phi\left(\frac{a(\gamma+1)}{\sqrt{\tau^2}}\right),
\end{aligned}
\end{equation}
with $\phi(x)$ and $\Phi(x)$ being the PDF and CDF of the standard Gaussian distribution, respectively. In the following, we shall write $u$ and $v$, instead of $u(\tau^2,\gamma,a)$ and $v(\tau^2,\gamma,a)$ to keep notations light.

\subsection{Proof of Lemma \ref{lemma:unique}: Existence of $(\tau_a^2,\gamma_a)$}\label{app:exist}
In this part, we prove that there exists a solution $(\tau_a^2,\gamma_a)$ to \eqref{condition0}, or equivalently \eqref{condition0:2}.   Our proof follows the same procedure as  \cite[Section 1.2.1]{LASSO} and \cite[Section 4.2]{elasticnet}.  First, we show that for any $\gamma>-1$, there exists a solution to \eqref{condition0:2a}. 
 \appendixlemma
\begin{lemma}\label{R1:1}
Given $\delta>0$ and $a\geq0$.  For any $\gamma>-1$, the  equation $F_1(\tau^2,\gamma;a)=0$ admits a unique solution $\tau_\gamma^2$.
\end{lemma}
\begin{proof}
Through simple calculation,  we have

\begin{equation}
\frac{\partial^2 F_1}{\partial(\tau^2)^2}(\tau^2,\gamma;a)=\frac{ a^3(\gamma+1)}{(\sqrt{\tau^2})^5}v,\end{equation} {where $v$ is defined in \eqref{def:uv}.} Hence, for $\gamma>-1$,
$F_1(\tau^2,\gamma;a)$ is convex in $\tau^2$. It is easy to check that $F_1(0,\gamma;a)=-1<0$ and 
\begin{equation}
\lim_{\tau^2\to\infty} F_1(\tau^2,\gamma;a)=\infty>0.
\end{equation}  By the intermediate value theorem and the convexity of $F_1$, we can conclude that  there exists a unique $\tau_\gamma^2$  such that $F_1(\tau_\gamma^2,\gamma;a)=0$.  
\end{proof}

{
\begin{lemma}
Given $ \delta > 0 $, $ \rho > 0 $, and $ a \geq 0 $, there exists $ \gamma > 0 $ such that the pair $ (\tau_\gamma^2, \gamma) $ satisfies \eqref{condition0:2b}, where $ \tau_\gamma^2 $ denotes the unique solution to 
$ F_1(\tau^2, \gamma; a) = 0 $ 
in $ \tau^2 $ for a given $ \gamma > 0 $.
\end{lemma}
}

\begin{proof}
We first claim that the mapping $\gamma\to\tau_\gamma^2$ is continuously differentiable on $(-1,\infty)$.
Recalling from Lemma \ref{R1:1} that $F_1(\tau^2,\gamma;a)$ is convex in $\tau^2$,  we have 
\begin{equation}
F_1(0,\gamma;a)\geq F_1(\tau_\gamma^2,\gamma;a)-\frac{\partial F_1}{\partial \tau^2}(\tau_\gamma^2,\gamma;a)\tau_\gamma^2.
\end{equation} This, together with $F_1(0,\gamma;a)=-1$ and $F_1(\tau_\gamma^2,\gamma;a)=0$, gives 
\begin{equation}
\frac{\partial F_1}{\partial \tau^2}(\tau_\gamma^2,\gamma;a)\geq  \frac{1}{\tau_\gamma^2}>0.
\end{equation}
 The claim then follows immediately from the implicit function theorem. Therefore, the function $F_2(\tau_\gamma^2,\gamma;a)$
is continuously differentiable on $\gamma\in(-1,\infty)$, with $F_2(\tau_\gamma^2,\gamma;a)\mid_{\gamma=0}=\rho$ and $\lim_{\gamma\to\infty}F_2(\tau_\gamma^2,\gamma;a)=-\infty$. By the  intermediate value theorem, we  get the desired result.
\end{proof}
Combining the above two lemmas, we can conclude that there exists a solution $(\tau_a^2,\gamma_a)$ that satisfies  \eqref{condition0:2}, where $\tau_a^2:=\tau_{\gamma_a}^2$. The uniqueness of the solution can be proved by properties of AMP; see  Appendix \ref{app:unique}.  
\subsection{Proof of Lemma \ref{lemma:astar}}\label{proof:astar}
\subsubsection{Strong Convexity of $f(a)$}
The strong convexity of $f(a)$ follows from the fact that it is the pointwise limit of a sequence of convex functions (see Lemma \ref{convergefN}) plus a quadratic term. 
\subsubsection{Continuous Differentiability of $f(a)$}\label{app:taua}
In this part, we show that both $\tau_a^2$ and $\gamma_a$ are continuously differentiable functions of $a$ on $[0,\infty)$, which further implies that $f(a)$ is continuously differentiable on $[0,\infty)$.  It suffices to prove that 
\begin{equation}\label{def:t0a}
t_0(a):=\det\left[\begin{matrix}\frac{\partial F_1}{\partial \tau^2}(\tau_a^2, \gamma_a;a)&\frac{\partial F_1}{\partial \gamma}(\tau_a^2, \gamma_a;a)\\\frac{\partial F_2}{\partial \tau^2}(\tau_a^2, \gamma_a;a)&\frac{\partial F_2}{\partial \gamma}(\tau_a^2, \gamma_a;a)\end{matrix}\right]\neq 0.
\end{equation}
The continuous differentiability of $\tau_a^2$ and $\gamma_a$ then follows immediately from the implicit function theorem. 
We now prove the above claim.  Through simple calculation, 
\begin{equation}\label{deriv}
\begin{aligned}
\frac{\partial F_1}{\partial \tau^2}(\tau^2,\gamma;a)&=1-\frac{u}{\delta(\gamma+1)^2}+\frac{2a}{\delta\sqrt{\tau^2}(\gamma+1)}v,\\
\frac{\partial F_1}{\partial \gamma}(\tau^2,\gamma;a)&=\frac{2\tau^2}{\delta(\gamma+1)^3}u-\frac{4 a\sqrt{\tau^2}}{\delta(\gamma+1)^2}v,\\
\frac{\partial F_2}{\partial \tau^2}(\tau^2,\gamma;a)&=-\frac{a\gamma}{\delta(\sqrt{\tau^2})^3}v,\\
\frac{\partial F_2}{\partial \gamma}(\tau^2,\gamma;a)&=-1+\frac{1}{\delta(\gamma+1)^2}u+\frac{2a\gamma}{\delta\sqrt{\tau^2}(\gamma+1)}v,
\end{aligned}
\end{equation}where $u$ and $v$ are given in \eqref{def:uv}, and thus 
\begin{equation}
\begin{aligned}
&\frac{\partial F_1}{\partial \tau^2}\frac{\partial F_2}{\partial \gamma}-\frac{\partial F_1}{\partial \gamma}\frac{\partial F_2}{\partial \tau^2}=\bigg(1-\frac{u}{\delta(\gamma+1)^2}+\frac{2a}{\delta\sqrt{\tau^2}(\gamma+1)}v\bigg)\bigg(\hspace{-0.1cm}-1+\frac{1}{\delta(\gamma+1)^2}u\bigg)+\frac{2a\gamma}{\delta\sqrt{\tau^2}(\gamma+1)}v.
\end{aligned}
\end{equation}
Denote by $u_a$ and $v_a$ the values of $u$ and $v$ at $(\tau_a^2, \gamma_a)$, respectively.  Utilizing the fact that $(\tau_a^2, \gamma_a)$ is the solution to \eqref{condition0:2}, we can express $t_0(a)$ in \eqref{def:t0a} as 
\begin{equation}
\begin{aligned}
t_0(a)=
\frac{1}{\tau_a^2}\left(1+\frac{a^2}{\delta}(1-u_a)\right)\frac{\gamma_a^2-\rho/\gamma_a}{\gamma_a+1}-\frac{\rho+\gamma_a^2}{\gamma_a+1}.\end{aligned}
\end{equation}
Note that
\begin{equation}
\begin{aligned}
1&<1+\frac{a^2}{\delta}(1-u_a)\\
&\overset{(a)}{=}\tau_a^2-\frac{\tau_a^2}{\delta(\gamma_a+1)^2}u_a+\frac{2a\tau_a}{\delta(\gamma_a+1)}v_a\\
&\overset{(b)}{=}\tau_a^2-\frac{2\tau_a^2}{\delta(\gamma_a+1)^2}\bigg(\Phi\bigg(\frac{a(\gamma_a+1)}{\sqrt{\tau_a^2}}\bigg)-\frac{a(\gamma_a+1)}{\sqrt{\tau_a^2}}\phi\bigg(\frac{a(\gamma_a+1)}{\sqrt{\tau_a^2}}\bigg)-\frac{1}{2}\bigg)\\
&\overset{(c)}{<}\tau_a^2,
\end{aligned}
\end{equation}
where (a) is due to \eqref{condition0:2a}, (b) uses the definitions of $u_a$ and $v_a$, and (c) holds since $\Phi(x)-x\phi(x)-\frac{1}{2}\geq0$ for $x\geq 0$.
Therefore, if $\gamma_a^2-\rho/\gamma_a\geq 0$, then 
\begin{equation}\label{t0bound1}
t_0(a)
<\frac{\gamma_a^2-\rho/\gamma_a}{\gamma_a+1}-\frac{\rho+\gamma_a^2}{\gamma_a+1}=\frac{-\rho/\gamma_a-\rho}{\gamma_a+1}<0;
\end{equation}
otherwise 
\begin{equation}\label{t0bound2}
t_0(a)
<-\frac{\rho+\gamma_a^2}{\gamma_a+1}<0,
\end{equation}
which proves the claim.
\subsubsection{$a^*>0$}\label{app:astar} Since $f(a)$ is continuously differentiable and strongly convex on $[0,\infty)$,  it suffices to prove that $f'(a)=0$ has a positive root.  Then, the root is exactly the minimizer of $f(a)$ on $[0,\infty)$, i.e., $a^*$. For this purpose, we next show that $f'(0)<0$ and $\lim_{a\to\infty}f'(a)>0$. The desired result then follows immediately from the  intermediate value theorem. 

{We first give} the explicit expression of  $f'(a)$. By the definition of $f(a)$ in \eqref{phia}, 
\begin{equation}
{f}'(a)=\delta\rho\left(1+\frac{\rho}{\gamma_a^2}\right)\frac{\partial \tau_a^2}{\partial a}-\frac{2\delta\rho^2\tau_a^2}{\gamma_a^3}\frac{\partial \gamma_a}{\partial a}+2\lambda a.
\end{equation}
According to the implicit function theorem, 
\begin{equation}\label{implicitfunction}
\begin{aligned}
&\bigg[\begin{matrix}\frac{\partial \tau^2_a}{\partial a}\\\frac{\partial \gamma_a}{\partial a}\end{matrix}\bigg]\hspace{-0.05cm}=\hspace{-0.05cm}-\hspace{-0.05cm}\bigg[\begin{matrix}\frac{\partial F_1}{\partial \tau^2}(\tau_a^2, \gamma_a,a)&\hspace{-0.1cm}\frac{\partial F_1}{\partial \gamma}(\tau_a^2, \gamma_a,a)\\\frac{\partial F_2}{\partial \tau^2}(\tau_a^2, \gamma_a,a)&\hspace{-0.1cm}\frac{\partial F_2}{\partial \gamma}(\tau_a^2, \gamma_a,a)\end{matrix}\bigg]^{\hspace{-0.05cm}-1}
\hspace{-0.05cm}\bigg[\begin{matrix}\frac{\partial F_1}{\partial a}(\tau_a^2, \gamma_a,a)\\\frac{\partial F_2}{\partial a}(\tau_a^2, \gamma_a,a)\end{matrix}\bigg].
\end{aligned}
\end{equation}
Hence, by direct calculation,  we get
\begin{equation}\label{derivativefa}
f'(a)=-\frac{t_1(a)}{t_0(a)}+\frac{t_2(a)}{t_0(a)}+2\lambda a,
\end{equation}
where
\begin{equation}
\begin{aligned}
t_1(a)&=\frac{2\rho a\gamma_a(1-u_a)}{\gamma_a+1}(1+\frac{\rho}{\gamma_a^2})^2+\frac{4\rho\gamma_a v_a}{\sqrt{\tau_a^2}(\gamma_a+1)}(1+\frac{\rho}{\gamma_a^2})(1-\tau_a^2),\\
t_2(a)&=\frac{4\rho^2v_a}{\sqrt{\tau_a^2}\gamma_a^2}.
\end{aligned}
\end{equation}
When $a=0$, we have $\tau_a^2=1$ and $\gamma_a=\rho$. Hence,  $t_1(0)=0$, $t_2(0)= 2\sqrt{\frac{2}{\pi}},$ and $t_0(0)=-1$, which gives $f'(0)=-2\sqrt{\frac{2}{\pi}}<0$.
  When $a=\infty$, the condition in \eqref{condition0} reduces to 
\begin{subequations}\label{condition:ainfty}
\begin{align}
&\tau^2\left(1-\frac{1}{\delta(\gamma+1)^2}\right)=1,\label{ainfty:1}\\
&\rho=\gamma\left(1-\frac{1}{\delta(\gamma+1)}\right)\label{ainfty:2}.
\end{align}
\end{subequations}
It is clear from \eqref{ainfty:1} that $\tau_a^2\geq 1$ and from  \eqref{ainfty:2} that $\rho\leq \gamma_a\leq \rho+\frac{1}{\delta}$. 
Combining this with the discussions in Appendix \ref{app:taua}, we have 
\begin{equation}
\begin{aligned}
t_0(a)&<-  \frac{1}{\gamma_a+1}\min\left\{\frac{\rho}{\gamma_a}+\rho,\rho+\gamma_a^2\right\}\\
&<-\frac{1}{\gamma_a+1}\rho\\
&<-\frac{\rho}{\rho+1+1/\delta}, 
\end{aligned}
\end{equation}
 where the first inequality follows from \eqref{t0bound1} and \eqref{t0bound2}, and the second inequality is due to $\gamma_a\leq \rho+1/\delta$.  In addition, we can show that $\lim_{a\to\infty}a(1-u_a)=\lim_{a\to\infty}v_a=0$, and thus  
 \begin{equation}\label{f'(a)infty}
 \lim_{a\to\infty}t_1(a)=\lim_{a\to\infty}t_2(a)=0.
 \end{equation} Therefore, $\lim_{a\to\infty}f'(a)=\infty$. This completes our proof.

 \section{Proof of Theorem \ref{mainresult:lambda_rho}}\label{proof:lambda_rho}
In this appendix, we prove Theorem \ref{mainresult:lambda_rho}. The main steps of the proof are outlined in Appendix \ref{proof:lambda_rho_1}, while all routine algebraic manipulations are deferred to Appendix \ref{proof:lambda_rho_2} to keep the logical flow transparent. 
\subsection{Proof Outline}\label{proof:lambda_rho_1}
 {Since the $Q$-function is decreasing,  minimizing the asymptotic SEP as in \eqref{opt:parameter} is equivalent to minimizing $\overline{\text{SNR}}^{-1}$. By simple calculation, $\overline{\text{SNR}}^{-1}$ can be expressed as a function of $(\tau_*^2,\gamma_*,a^*)$ given as follows:
   \begin{equation}
\overline{\text{SNR}}^{-1}=\frac{\bar{\beta}_0^2+\sigma^2}{\bar{\alpha}_0^2}=F(\tau_*^2,\gamma_*,a^*),
\end{equation}
 where   
 \begin{equation}\label{F:SNR}
\begin{aligned}
&F(\tau^2,\gamma,a)=\bigg(\frac{\pi}{2}\delta(1+\sigma^2\delta)+1\bigg)\tau^2-1-\frac{2\tau^2}{\gamma+1}+4\sqrt{\frac{\pi}{2}}\frac{\tau^2}{\gamma+1}v(\tau^2,\gamma,a)-2a\sqrt{\frac{\pi}{2}\tau^2}(1-u(\tau^2,\gamma,a)),
\end{aligned}
\end{equation}
and $u(\tau^2,\gamma,a)$ and $v(\tau^2,\gamma,a)$ are defined as follows (cf.~\eqref{def:uv}):
\begin{equation}\label{def:uv_recall}
\begin{aligned}
u(\tau^2,\gamma,a):&=2\Phi\left(\frac{a(\gamma+1)}{\sqrt{\tau^2}}\right)-1,\\~v(\tau^2,\gamma,a):&=\phi\left(\frac{a(\gamma+1)}{\sqrt{\tau^2}}\right).
\end{aligned}
\end{equation}
In what follows, we shall write $u$ and $v$, instead of $u(\tau^2,\gamma,a)$ and $v(\tau^2,\gamma,a)$ to keep notations light.

The tuple $(\tau_*^2,\gamma_*,a^*)$ is further a function of $(\lambda,\rho)$. Specifically,  for given $(\lambda,\rho)$, $(\tau_*^2,\gamma_*)$ is the solution to \eqref{condition0} with $a=a^*$, where $a^*$ is defined in \eqref{minfa} and can be computed as the unique solution to $f'(a)=0$. Following the calculation in Appendix \ref{app:lemma1}, $(\tau_*^2,\gamma_*,a^*)$ is the unique solution to
 \begin{subequations}\label{34}
\begin{align}
&\tau^2\hspace{-0.05cm}-\hspace{-0.05cm}1-\hspace{-0.05cm}\frac{\tau^2}{\delta(\gamma+1)^2}u+\frac{a^2}{\delta}u+\frac{2a\sqrt{\tau^2}}{\delta(\gamma+1)}v-\frac{a^2}{\delta}=0,\label{34a}\\
&\rho-\gamma+\frac{\gamma}{\delta(\gamma+1)}u=0,\label{34b}\\
&\rho a\gamma(1\hspace{-0.05cm}-\hspace{-0.05cm}u)(1\hspace{-0.05cm}+\hspace{-0.05cm}\frac{\rho}{\gamma^2})^2\hspace{-0.1cm}+\hspace{-0.05cm}\frac{2\rho\gamma v}{\sqrt{\tau^2}}(1\hspace{-0.05cm}+\hspace{-0.05cm}\frac{\rho}{\gamma^2})(1\hspace{-0.05cm}-\hspace{-0.05cm}\tau^2)\hspace{-0.05cm}-\hspace{-0.1cm}\frac{2\rho^2v}{\sqrt{\tau^2}\gamma^2}(1\hspace{-0.05cm}+\hspace{-0.05cm}\gamma)\notag\\
&-\frac{\lambda a}{\tau^2}(1+\frac{a^2}{\delta}(1-u))(\gamma^2-\frac{\rho}{\gamma})+\lambda a (\rho+\gamma^2)=0,\label{34c}
\end{align}
\end{subequations}
where \eqref{34a} and \eqref{34b} are equivalent to \eqref{condition0} (see \eqref{condition0:2}), and \eqref{34c} is equivalent to $f'(a)=0$ (see \eqref{derivativefa}). 

Based on the above discussions,  we can rewrite problem \eqref{opt:parameter} as follows: 
 \begin{equation}\label{opt:parameter2}
\begin{aligned}
\min_{\lambda,\rho,\tau^2,\gamma,a}~&F(\tau^2,\gamma,a)\\
\text{s.t. }~~& \eqref{34a}-\eqref{34c},\\
&\gamma\geq 0,~a> 0,~\lambda\geq 0,~ \rho\geq 0,
\end{aligned}\end{equation}
where $F(\tau^2,\gamma,a)$ is given in \eqref{F:SNR}, 
 $a>0$ is due to Lemma \ref{lemma:astar}, and $\gamma=0$ is included in the feasible region as it is the solution to \eqref{condition0} for $\rho=0$. } 

\medskip
\noindent\textbf{Step\,1. Formulate a relaxation problem \eqref{eq:R}.}
The objective function 
$F(\tau^{2},\gamma,a)$ in~\eqref{F:SNR} 
depends \emph{only} on $(\tau^{2},\gamma,a)$,
whereas the variables $(\lambda,\rho)$ appear only
in the affine constraints \eqref{34b} and \eqref{34c}.
Based on this observation, we introduce the following relaxation problem on $(\tau^2,\gamma,a)$ by ignoring constraints \eqref{34b} and \eqref{34c}:
\begin{equation}
\min_{\tau^{2},\gamma,a}F(\tau^{2},\gamma,a)
\quad\text{s.t.}\quad
\eqref{34a},\;
\gamma\ge0,\;a>0.
\tag{$\mathcal R$}\label{eq:R}
\end{equation}
Let $(\hat\tau^{2},\hat\gamma,\hat a)$ be an optimal solution to \eqref{eq:R} and
$\hat{F}$ be the optimal value. Then $\hat{F}$ is a \emph{lower bound}
on the optimal value of the  original problem~\eqref{opt:parameter2}.

\medskip
\noindent\textbf{Step\,2. Simplify \eqref{eq:R}.} {The goal of this step is to transform the 3-dimensional optimization problem \eqref{eq:R} into an equivalent 2-dimensional one, by leveraging the equality constraint \eqref{34a}. However, the constraint \eqref{34a} in its current form  is difficult to deal with, as variables $(\tau^2,\gamma,a)$ are coupled in a complicated way, particularly in the nonlinear functions $\phi(\cdot)$ and $\Phi(\cdot)$; see \eqref{def:uv_recall}.   To simplify it, we introduce an auxiliary variable 
\begin{equation}\label{def:z}
z:=a(\gamma+1)/\sqrt{\tau^{2}}>0.
\end{equation} Next, we  show that the optimization problem \eqref{eq:R} on $(\tau^2,\gamma,a)$ can be transformed into an equivalent problem on $(\tau^2,z)$.

With $z$ defined in \eqref{def:z}, the constraint \eqref{34a} can be rewritten as 
\begin{equation}\label{34a_equi}
{\tau^{2}-1}=a^2h(z),
\end{equation} where $h(z)$ is given in \eqref{def:hz}, which satisfies 
\begin{equation}\label{def:hz_recall}
h(z)>0,~\forall~z>0;
\end{equation} 
see Appendix \ref{proofof53} for a proof of the positiveness of $h(z)$.  Since $a>0$, we further get 
\begin{equation}\label{a:tauz}
a=\sqrt{(\tau^2-1)/h(z)}.
\end{equation} 
Using \eqref{a:tauz}, we can eliminate variable $a$ in \eqref{eq:R}. Specifically, by \eqref{a:tauz} and the definition of $z$ in \eqref{def:z}, the objective function in \eqref{eq:R} transforms into 
\begin{equation}
f(\tau^2,z)=c\tau^{2}-\sqrt{\tau^{2}}\,
\sqrt{\frac{\tau^{2}-1}{h(z)}}\,g(z),
\end{equation}
where $c=\frac{\pi}{2}\delta(1+\sigma^{2}\delta)+1$ and \(
g(z)=2z^{-1}-4\sqrt{\pi/2}\bigl(\Phi(z)+\phi(z)/z-1\bigr)
\). Since the above function is independent of $\gamma$,  we can further eliminate variable $\gamma$, by rewriting the constraint 
$\gamma\geq 0$ into a constraint over $(\tau^2,z)$ as follows:
\begin{equation}\label{gamma>=0}
\begin{aligned}
&\gamma=\frac{z\sqrt{\tau^2}}{a}-1=\frac{z\sqrt{\tau^2}}{\sqrt{(\tau^2-1)/h(z)}}-1\geq 0~\Longleftrightarrow~~\tau^{2}z^{2}h(z)\ge\tau^{2}-1,
\end{aligned}
\end{equation}
where the first equality uses the definition of $z$ in \eqref{def:z} and the second equality applies \eqref{a:tauz}.  In particular, if the above constraint is satisfied with equality, then $\gamma=0$.

To conclude, \eqref{eq:R} transforms into the following problem over $(\tau^2,z)$:
\begin{equation}
\begin{aligned}
\min_{\tau^{2}>1,\;z>0}\;&f(\tau^2,z)=
c\tau^{2}-\sqrt{\tau^{2}}\,
\sqrt{\frac{\tau^{2}-1}{h(z)}}\,g(z)  \\
\text{s.t.}\;~~~&
\tau^{2}z^{2}h(z)\ge\tau^{2}-1.
\end{aligned}
\tag{$\mathcal R'$}\label{eq:R'}
\end{equation}

\medskip
\noindent\textbf{Step\,3. Prove that the inequality constraint in \eqref{eq:R'} is active at the optimum and $\hat{\gamma}=0$.} Let $(\hat{\tau}^2,\hat{z})$ be any optimal solution to \eqref{eq:R'}. Our goal is to show that the inequality constraint in \eqref{eq:R'} is satisfied with equality, i.e.,  \begin{equation}\label{active}
\hat{\tau}^{2}\hat{z}^{2}h(\hat{z})=\hat{\tau}^{2}-1.
\end{equation}As implied by \eqref{gamma>=0}, this further implies 
\begin{equation}\label{gamma=0}
\hat{\gamma}=0.\end{equation}
To show \eqref{active}, we prove the following two results (see Appendices \ref{proofof53} and \ref{proofof54} for a rigorous proof): 
\begin{itemize}
\item Given $\tau^2>1,$ $f(\tau^2,z)$ is strictly decreasing in $z>0$.
\item $z^2h(z)$ is strictly increasing in $z>0$.
\end{itemize}
Now we prove \eqref{active}. Suppose for contradiction that $\hat{\tau}^{2}\hat{z}^{2}h(\hat{z})>\hat{\tau}^{2}-1.$   Since $z^2h(z)$ is strictly increasing in $z>0$,  we can select a smaller value $0<\tilde{z} < \hat{z}$ such that the constraint in \eqref{eq:R'} still holds, i.e., $\hat{\tau}^2 \tilde{z}^2 h(\tilde{z}) \geq\hat{\tau}^2-1$. However, as $f(\tau^2, z)$ is strictly decreasing in $z > 0$, this would lead to a smaller objective value, contradicting the optimality of $(\hat{\tau}^2, \hat{z})$. Hence, the inequality constraint must be active at the optimum.

\medskip
\noindent\textbf{Step\,4. Transform \eqref{eq:R'} into a one–dimensional problem.}
Since the constraint in \eqref{eq:R'} is active as shown in Step 3, we can further eliminate variable $\tau^2$ using the relation
\begin{equation}\label{tauandz}
\tau^{2}=(1-z^{2}h(z))^{-1}.
\end{equation}
We claim that the constraint $\tau^2>1$ in \eqref{eq:R'} is equivalent to $z\in(0,z_0)$, where $z_0$ is defined in Theorem \ref{mainresult:lambda_rho}. 
Hence, \eqref{eq:R'} transforms into the following one-dimensional problem
\begin{equation}
\min_{z\in(0,z_0)}~\xi(z):=\frac{c-zg(z)}{1-z^{2}h(z)}.
\tag{$\mathcal R''$}
\label{eq:R''}
\end{equation}
By simple calculation, $\xi'(z)\propto\zeta(z)$, i.e., $\xi'(z)=k \zeta(z)$ with $k>0$,  where $\zeta(z)$ is given in \eqref{def:zetaz}. In addition,  $\zeta(z)$ is concave on $(0,z_{0})$ with  $\zeta(0)<0<\zeta(z_{0})$. Hence,  there is exactly one stationary point~$\hat z$,
which is therefore the solution to \eqref{eq:R''}. We give a detailed proof of this step in Appendix \ref{proof_Step4}.

By \eqref{a:tauz},  \eqref{gamma=0}, and \eqref{tauandz}, the solution to \eqref{eq:R} is given by 
\begin{equation}
\hat\tau^2=(1-z^{2}h(z))^{-1},\;~
\hat\gamma=0,\;~
\hat a=(\hat z^{-2}-h(\hat z))^{-\frac{1}{2}}.
\end{equation}
}
\medskip
\noindent\textbf{Step\,5. Recover $(\hat\rho,\hat\lambda)$ and verify its
feasibility.}
With $\hat\gamma=0$, constraint~\eqref{34b} yields $\hat\rho=0$. By \eqref{34c}, 
\begin{equation}
\hat{\lambda}=\frac{2\big(1-\frac{2\Phi(\hat{z})-1}{\delta}\big)\left(\phi(\hat{z})-\hat{a}\hat{\tau}\big(1-\frac{2\Phi(\hat{z})-1}{\delta}\big)(1-\Phi(\hat{z}))\right)}{\hat{z}\left(1+\frac{2\hat{a}^2}{\delta}(1-\Phi(\hat{z}))\right)}.
\end{equation}
Note that $(\hat{\tau}^2,\hat{\gamma},\hat{a})$ satisfies \eqref{34a} and \eqref{34b} with $\rho=0$. According to the notations in Lemmas \ref{lemma:unique} and  \ref{lemma:astar}, 
 $\hat{\tau}^2=\tau_{\hat{a}}^2$ and $\hat{\gamma}=\gamma_{\hat{a}}$ with $\rho=0$, and thus constraint \eqref{34c} reduces to
$f'(\hat a)=0$, where $f(a)=\bar{f}(a)+\lambda a^2$  and $\bar{f}(a)$ is defined in Lemma \ref{convergefN}. Therefore,  $\hat{\lambda}$ can be expressed as 
\begin{equation}
\hat{\lambda}=-\frac{\bar{f}'(\hat{a})}{2 \hat{a}}.
\end{equation}
 According to the proof of Lemma \ref{lemma:astar} (see Appendix \ref{proof:astar}), $\bar{f}(a)$ is convex, and thus $\bar{f}'(a)$ is increasing. In addition, $\lim_{a\to\infty}\bar{f}'(a)=0$ (see \eqref{f'(a)infty}). This implies that $\hat{\lambda}\geq 0$.
 
Combining all the above, $(\hat\lambda,\hat\rho,\hat\tau^{2},\hat\gamma,\hat a)$
satisfies \eqref{34a}–\eqref{34c} and achieves $\hat{F}$, and is hence globally optimal.

  \subsection{Algebraic Details}\label{proof:lambda_rho_2}
  This part collects all algebraic details involved in the main steps. 
    \subsubsection{Positiveness and Strict Increasing of $z^2 h(z)$}\label{proofof53}
Note that 
\begin{equation}\label{z2hz}
\frac{\partial z^2h(z)}{\partial z}=\frac{4z}{\delta}(1-\Phi(z))> 0,~\forall~z>0,
\end{equation}
 i.e., $z^2h(z)$ is increasing when $z>0$. Since $z^2h(z)=0$ for $z=0$, it follows that 
$z^2h(z)>0$ for all $z>0.$ This further implies $h(z)>0$ for all $z>0$, which gives  \eqref{def:hz_recall}.
  \subsubsection{Strict Decreasing of $f(\tau^2, z)$}\label{proofof54}
By simple calculation,
\begin{equation}
\frac{\partial f(\tau^2,z)}{\partial z}\hspace{-0.05cm}=\hspace{-0.05cm}-\frac{\sqrt{\tau^2(\tau^2-1)}}{2}(2g'(z)h(z)-h'(z)g(z)),
\end{equation}
and 
\begin{equation}
\begin{aligned}
&2g'(z)h(z)-h'(z)g(z)=\frac{4}{\delta z^3}(1-\Phi(z))(2z-\sqrt{2\pi}(2\Phi(z)-1)).
\end{aligned}
\end{equation}
Let $t(z)=2z-\sqrt{2\pi}(2\Phi(z)-1)$. We have
\begin{equation}
t'(z)=2(1-\sqrt{2\pi}\phi(z))>0,~~\forall~z>0
\end{equation}
  and $t(0)=0$.  
Hence, $t(z)>0$ for $z>0$. Combining this with $\Phi(z)<1$ and $\tau^2>1$, we get 
\begin{equation}
\frac{\partial f(\tau^2,z)}{\partial z}<0,~\forall~z>0,
\end{equation}
i.e., $f(\tau^2,z)$ is strictly decreasing on $z>0$.
\subsubsection{Proof of Step 4}\label{proof_Step4}
We now give a detailed proof of Step 4, i.e., transforming \eqref{eq:R'} into \eqref{eq:R''} and solving \eqref{eq:R''}. As shown by Step 3,  \(
\tau^{2}=(1-z^{2}h(z))^{-1}.
\)
Substituting this relation into the objective function of \eqref{eq:R'} directly gives the objective $\xi(z)$ in \eqref{eq:R''}. In addition, the constraint  $\tau^2>1$ transforms into 
\begin{equation}
0<1-z^2h(z)<1.
\end{equation}
Since $h(z)>0$ for all $z>0$, the right-hand side of the above inequality naturally holds for $z>0$. Next, we show that 
$1-z^2h(z)>0$ is equivalent to $z<z_0$. As shown in \eqref{z2hz},  $z^2h(z)$ is increasing in $z\in(0,\infty)$. In addition, $z^2h(z)=0$ if $z=0$ and 
\begin{equation}
\begin{aligned}
\lim_{z\to\infty} z^2h(z)&\hspace{-0.05cm}=\lim_{z\to\infty}\frac{2\Phi(z)-1}{\delta}-\frac{2}{\delta}\left(z\phi(z)+z^2(\Phi(z)-1)\right)=\frac{1}{\delta}.
\end{aligned}
\end{equation}
It follows that when $\delta\in(0,1)$, there exists a unique solution to $1-z^2h(z)=0$, which we denote by $z_0$. When $\delta\geq1$, $1-z^2h(z)>0$  always holds true and we set $z_0=\infty$. Therefore, we can write $1-z^2h(z)>0$ as  $z<z_0$ and transform \eqref{eq:R'}  into \eqref{eq:R''}, which we copy here for clarity:
\[
\min_{z\in(0,z_0)}~\xi(z):=\frac{c-zg(z)}{1-z^{2}h(z)}.
\tag{$\mathcal R''$}
\label{eq:R''}
\] In the following, we solve \eqref{eq:R''}. 
With simple calculation, 
\begin{equation}\label{xi'}
\xi'(z)=\frac{4(1-\Phi(z))}{\delta(1-z^2h(z))^2}\zeta(z),
\end{equation}where $\zeta(z)$ is given in \eqref{def:zetaz}.
Further, we have 
\begin{equation}
\zeta'(z)=2\sqrt{2\pi}z\bigg(\Phi(z)+\frac{\phi(z)}{z}-1\bigg)+\frac{\pi}{2}\delta(1+\sigma^2\delta)-1
\end{equation}
and 
\begin{equation}
\zeta''(z)=2\sqrt{2\pi}(\Phi(z)-1)\leq 0.
\end{equation}
Hence, $\zeta(z)$ is concave. 
Note that $\zeta(0)=-\sqrt{\frac{\pi}{2}}\delta<0$.  We next show that $\zeta(z_0)>0$. When $\delta\geq1$, it is easy to check that $\lim_{z\to\infty}\zeta(z)=\infty$. Now consider the case $\delta\in(0,1)$.
Using the fact that $z_0$ is the  solution to $1-z^2h(z)=0$, we can express $\delta$ as a function of $z_0$ as follows:
\begin{equation}\label{delta_z}
\delta={2\Phi(z_0)-1}-2z_0^2\left(\frac{\phi(z_0)}{z_0}+\Phi(z_0)-1\right).
\end{equation}
With this, we have
\begin{equation}
\begin{aligned}
\zeta(z_0)\hspace{-0.1cm}\overset{(a)}{=}&2\sqrt{2\pi}z_0^2\bigg(\hspace{-0.05cm}\Phi(z_0)+\frac{\phi(z_0)}{z_0}\hspace{-0.05cm}-\hspace{-0.05cm}1\hspace{-0.05cm}\bigg)\hspace{-0.1cm}+\hspace{-0.1cm}\big(\frac{\pi}{2}\delta(1+\sigma^2\delta)-1\big)z_0\\
\geq\,&2\sqrt{2\pi}z_0^2\left(\Phi(z_0)+\frac{\phi(z_0)}{z_0}-1\right)+\left(\frac{\pi}{2}\delta-1\right)z_0\\
\overset{(b)}{=}\,&z_0\bigg(\frac{\pi}{2}\left(2\Phi(z_0)-1\right)-1+(2\sqrt{2\pi}-\pi z_0)\left(z_0\Phi(z_0)+\phi(z_0)-z_0\right)\bigg)\\
:=\,&z_0 G(z_0),
\end{aligned}
\end{equation}
where (a) and (b) are obtained by replacing $\delta$ by the right-hand side of \eqref{delta_z}. Note that 
\begin{equation}
G'(z)=2\pi\left(z-\sqrt{\frac{2}{\pi}}\right)(1-\Phi(z)),
\end{equation}
and thus $G(z)$ decreases in $z\in(0,\sqrt{{2}/{\pi}})$ and increases in $z\in(\sqrt{{2}/{\pi}},+\infty)$.
 Since $G(\sqrt{{2}/{\pi}})>0$, we have 
 $G(z_0)>0$, which further gives $\zeta(z_0)>0$. Due to the concavity of $\zeta$ and $\zeta(0)<0$, we can now conclude that there exists a unique root to $\zeta(z)=0,~z\in(0,z_0)$, which we denote by $\hat{z}$. In addition, $\zeta(z)<0$ for $z\in(0,\hat{z})$ and  $\zeta(z)>0$ for $z\in(\hat{z},z_0)$. From \eqref{xi'}, $\xi(z)$ decreases in $z\in(0,\hat{z})$ and increases in $z\in(\hat{z},z_0)$, i.e., $\hat{z}$ is the solution to \eqref{eq:R''}. 
 \section{An introduction of AMP}\label{appendix:AMP}
In this appendix, we give a brief introduction of the AMP algorithm and its state evolution theory. The AMP algorithm in its most general form is given as follows.  Let $\{f_t\}_{t\geq 0}$ and $\{g_t\}_{t\geq 0}$ be two sequences of functions, where $f_t: \R^2\rightarrow \R$ and $g_t:\R^2\rightarrow \R$ are Lipschitz continuous. Given $\w\in\R^K$ and $\x_0^*\in\R^N$, whose elements are i.i.d. drawn from two probability measures $p_{W}$ and $p_{X_0^*}$, respectively, and given the initial vector $\bq_0\in\R^{N}$. Let $\bH\in\R^{K\times N}$ be a random matrix satisfying Assumption \ref{ass} (ii). The general AMP algorithm generates sequences of vectors $\{\br_t\}_{t\geq 1}$, $\{\bm_t\}_{t\geq 0}$, $\{\bb_t\}_{t\geq 0}$, $\{\bq_{t}\}_{t\geq 1}$ through the following recursion  \cite[Eq. (3.3)]{AMP}:
\begin{equation}\label{general_AMP}
\begin{aligned}
\mathbf{r}_{t+1}&=\bH^T\bm_t-\xi_t\bq_t,~~\bm_t=g_t(\bb_t,\bw),\\
\bb_t&=\bH\bq_t-\lambda_t\bm_{t-1},~~\bq_t=f_t(\br_t,\x_0^*),
\end{aligned}
\end{equation}
where $\xi_t=\left<g'_t(\bb_t,\bw)\right>$ and $\lambda_t=\frac{1}{\delta}\left<f_t'(\br_t,\x_0^*)\right>,$ and $f_t(\cdot,\cdot)$ and $g_t(\cdot,\cdot)$ act element-wise on their input vectors.  An important property of the AMP algorithm is that the statistical properties of the generated sequences can be characterized by a scalar recursion known as state evolution. The state evolution defines  two  sequences of quantities $\{\tau_t^2\}_{t \geq 0}$ and $\{\sigma_t^2\}_{t\geq 1}$ through 
\begin{equation}\label{general:SE}
\begin{aligned}
\tau_t^2&=\mathbb{E}\left[g_{t}(\sigma_tZ,W)^2\right], \\
\sigma_t^2&=\frac{1}{\delta}\mathbb{E}\left[f_t(\tau_{t-1}Z,X_0^*)^2\right],
\end{aligned}
\end{equation}
where $Z\sim\mathcal{N}(0,1)$ is independent of $W\sim p_W$ and $X_0^*\sim p_{X_0^*}$, and $\sigma_0^2=\lim_{N\to\infty}\frac{1}{\delta N}\|\bq_0\|^2$.
\appendixpro
\begin{proposition}[\hspace{-0.001cm}{\cite[Lemma 1]{AMP}}]\label{AMP_property}
Assume that Assumption \ref{ass} (i) -- (ii) and the following condition hold:
\begin{equation}
\mathbb{E}[(X_0^*)^2]<\infty,~~\mathbb{E}[W^2]<\infty,
\end{equation}
where $W\sim p_W$ and $X_0^*\sim p_{X_0^*}$. {Assume that $\mathbf{q}_0$ is independent of $\mathbf{H}$ and that the limit $\sigma_0^2=\lim_{N\to\infty}\frac{1}{\delta N}\|\bq_0\|^2$ exists almost surely.}

\begin{enumerate}
\item[(i)] For all pseudo-Lipschitz functions $\varphi_r,\varphi_b:~\R^{t+2}\to\R$ and $t\geq 0$,
\begin{equation}
\begin{aligned}
&\lim_{N\to\infty}\frac{1}{N}\sum_{i=1}^N\varphi_r(r_{1,i},r_{2,i},\dots, r_{t+1,i}, x^*_{0,i})\overset{a.s.}{=}\mathbb{E}\left[\varphi_r(\tau_0Z_0,\tau_1Z_1,\dots,\tau_tZ_t,X_0^*)\right],\\
&\lim_{K\to\infty}\frac{1}{K}\sum_{i=1}^K\varphi_b(b_{0,i},b_{1,i},\dots, b_{t,i}, w_{i})\overset{a.s.}{=}\mathbb{E}\left[\varphi_b(\sigma_0\hat{Z}_0,\sigma_1\hat{Z}_1,\dots,\sigma_t\hat{Z}_t,W)\right],
\end{aligned}
\end{equation}
where $(Z_0,Z_1,\dots,Z_t)$ and $(\hat{Z}_0,\hat{Z}_1,\dots,\hat{Z}_t)$ are two zero-mean Gaussian vectors independent of $X_0^*,$ $W$, with $Z_i,\hat{Z}_i\sim \mathcal{N}(0,1)$.
\item[(ii)] For all $0\leq s,n \leq t$, the following equations hold and all limits exist, are bounded and have degenerate distributions (i.e., they are constant random variables): 
\begin{equation}
\begin{aligned}\lim_{N\to\infty}\left<\br_{s+1},\br_{n+1}\right>\overset{a.s.}{=}\lim_{K\to\infty}\left<\bm_{s},\bm_{n}\right>,\\
\lim_{K\to\infty}\left<\bb_{s},\bb_{n}\right>\overset{a.s.}{=}\frac{1}{\delta}\lim_{N\to\infty}\left<\bq_{s},\bq_{n}\right>.
\end{aligned}
\end{equation}
\end{enumerate}
\end{proposition}
 The AMP algorithm in \eqref{AMP_copy} is a special case of \eqref{general_AMP} with $\x_0^*=\mathbf{0},$ $\bw=\s$, $f_t(s,x_0^*)=\eta_a(x_0^*-s;\gamma_a)-x_0^*$, and   $g_t(s,w)=s-w$ \cite{AMP},  where  
  $(\x_t,\bz_t)$ and $(\br_t,\bq_t,\bb_t,\bm_t)$ are related as follows: 
\begin{equation}\label{AMP_relation}
\begin{aligned}
\br_{t+1}&=-(\bH^T\bz_t+\x_t),\\
\bq_t&=\x_t,\\
\bb_t&=\s-\bz_t,\\
\bm_t&=-\bz_t.\\
\end{aligned}
\end{equation}
{The initial condition is $\bq_0=\mathbf{0}$}. In this case, the state evolution reduces to 
\begin{equation}\label{state:AMP}
\tau_{t+1}^2=1+\frac{1}{\delta}\mathbb{E}\left[\eta_a^2(\tau_tZ;\gamma_a)\right],\end{equation}
where $Z\sim\mathcal{N}(0,1)$ and $\tau_0^2=1$. As a direct corollary of Proposition~\ref{AMP_property}, we state the following result, which will be used frequently in our analysis.
\begin{proposition}[\hspace{-0.001cm}{\cite[Theorem 1]{AMP}}]\label{pro:4}
Consider the AMP algorithm defined in \eqref{AMP_copy}.  For all  pseudo-Lipschitz function $\varphi:\R\rightarrow \R$ and all $t\geq 0$, the following results hold under Assumption \ref{ass} (i) -- (iii):
\begin{itemize}
{\item[(i)]For $\br_{t+1}$ and $\bb_{t}$ defined in \eqref{AMP_relation}, 
\begin{equation}\begin{aligned}
\lim_{N\to\infty}\frac{1}{N}\sum_{i=1}^N\varphi(r_{t+1,i})&\overset{a.s.}{=}\mathbb{E}\left[\varphi(\tau_{t}Z)\right],\\
\lim_{K\to\infty}\frac{1}{K}\sum_{i=1}^K\varphi(b_{t,i})&\overset{a.s.}{=}\mathbb{E}\left[\varphi(\sigma_{t}Z)\right],
\end{aligned}
\end{equation}
where $\tau_t^2$ is obtained by the state evolution in \eqref{state:AMP}, 
$\sigma_t^2=\frac{1}{\delta}\mathbb{E}[\eta_a^2(\tau_{t-1}Z;\gamma_a)], $ and $Z\sim\mathcal{N}(0,1)$.}
\item[(ii)] For $\x_{t+1}$ and $\z_t$, 
\begin{equation}
\lim_{N\to\infty}\frac{1}{N}\sum_{i=1}^N\varphi(x_{t+1,i})\overset{a.s.}{=}\mathbb{E}\left[\varphi\circ\eta_a(\tau_{t}Z;\gamma_a)\right],
\end{equation}
\begin{equation}
\lim_{K\to\infty}\frac{1}{K}\sum_{i=1}^K\varphi(z_{t,i})\overset{a.s.}{=}\mathbb{E}\left[\varphi(S+\sigma_{t}Z)\right],
\end{equation}
where $S\sim\text{\normalfont Unif}(\mathcal{S})$ is independent of $Z$.

\end{itemize}
\end{proposition}

\section{Proof of Proposition \ref{prob:converge}}\label{appendix:proofofconverge}
In this appendix, we give the proof of Proposition \ref{prob:converge}. The proof contains three main steps. First, we show that  $\tau_t^2$ defined by the state evolution in \eqref{state:AMP} converges to $\tau_a^2$ as $t\to\infty$.
\begin{lemma}\label{R1:1.1}
Let $\{\tau_t^2\}_{t\geq 0}$ be defined through the state evolution in \eqref{state:AMP}. Then $\lim_{t\to\infty}\tau_t^2=\tau_a^2$, where $\tau_a^2$ is the unique solution to \eqref{conditiona} with $\gamma=\gamma_a$.
\end{lemma}
\begin{proof}
See Appendix \ref{app:convergeA}. 
\end{proof}
 \hspace{-0.4cm}The second step is to show that the difference between successive iterates vanishes as $t,N\to\infty$.
\begin{lemma}\label{difference}
The following result holds for the AMP iterates in \eqref{AMP_copy}: 
\begin{equation}
\begin{aligned}
&\lim_{t\to\infty}\lim_{N\to\infty}\frac{1}{N}\|\x_{t+1}-\x_t\|^2\overset{a.s.}{=}0,\\
&\lim_{t\to\infty}\lim_{N\to\infty}\frac{1}{N}\|\z_{t+1}-\z_{t}\|^2\overset{a.s.}{=}0.
\end{aligned}
\end{equation}
\end{lemma}
\begin{proof}
See Appendix \ref{app:convergeB}. 
\end{proof}
Proposition \ref{prob:converge} can then be proved by connecting the first-order optimality condition of problem \eqref{def:xa} with the AMP iterations in \eqref{AMP_copy} and utilizing Lemmas \ref{R1:1.1} and \ref{difference}; see Appendix \ref{app:probconverge}.

{Before delving into the detailed proof, we first provide a heuristic derivation in the following subsection to discuss the relation between the AMP iterations in \eqref{AMP_copy} and the solution to \eqref{def:xa}.

\subsection{Heuristic Derivation}\label{appD:heuristic}
Assume that $(\x,\z)$ is a fixed point of the AMP iterations in \eqref{AMP_copy}, and let $\alpha=\delta^{-1}\langle \eta_a'(\x+\bH^T\z;\gamma_a)\rangle$. 
According to the formula of the AMP iterations in \eqref{AMP_copy}, $(\x,\z)$ satisfies 
\begin{equation}
\begin{aligned}
\x&=\eta_a(\x+\bH^T\z;\gamma_a),\\
\z&=\s-\H\x+\alpha\z,
\end{aligned}
\end{equation}
which further gives
\begin{equation}\label{AMP_fixpoint}
\begin{aligned}
\x&=\eta_a\left(\x-\frac{\bH^T(\bH\x-\s)}{1-\alpha};\gamma_a\right)=\mathcal{P}_{[-a,a]^N}\left(\frac{\x}{\gamma_a+1}-\frac{\bH^T(\bH\x-\s)}{(\gamma_a+1)(1-\alpha)}\right),
\end{aligned}
\end{equation}
where the second equality uses the definition of $\eta_a(\cdot;\gamma_a)$ in Lemma \ref{lemma:unique}. On the other hand, the following first-order optimality condition holds for the solution to     \eqref{def:xa} with any $\omega>0$:
\begin{equation}\label{OC}
\begin{aligned}
\x&=\mathcal{P}_{[-a,a]^N}\left({\x}-\omega\left(\bH^T(\bH\x-\s)+\rho\x\right)\right)=\mathcal{P}_{[-a,a]^N}\left((1-\omega\rho){\x}-\omega\bH^T(\bH\x-\s)\right).
\end{aligned}
\end{equation}
Comparing \eqref{AMP_fixpoint} and \eqref{OC}, the fixed point of the AMP algorithm in \eqref{AMP_copy} is the solution to \eqref{def:xa} if the following condition is satisfied: 
\begin{equation}\label{fixpoint_OC}
1-\omega\rho=(\gamma_a+1)^{-1}~\text{ and }~~\omega=\left((\gamma_a+1)(1-\alpha)\right)^{-1}.
\end{equation}
This condition is guaranteed by \eqref{condition0b}. Specifically, by Proposition \ref{pro:4} and Lemma \ref{R1:1.1},
$\alpha\approx\delta^{-1}\langle \eta_a'(\tau_aZ;\gamma_a)\rangle$.  It then follows from  \eqref{condition0b} and the definition of $\eta_a(\cdot;\gamma_a)$ that $1-\alpha\approx\rho\gamma^{-1}_a$,  yielding \eqref{fixpoint_OC}. 

The above derivation explains why the iterations generated by \eqref{AMP_copy} converge to the solution of \eqref{def:xa}, and also provides the motivation for the condition in \eqref{condition0b}, which ensures consistency between the AMP fixed point and the optimality condition of \eqref{def:xa}.}

 \subsection{Proof of Lemma \ref{R1:1.1}}\label{app:convergeA}
Let $F(\tau^2):=1+\frac{1}{\delta}\mathbb{E}\left[(\eta_a(\tau Z;\gamma_a))^2\right], $ then the state evolution can be written as  $\tau_{t+1}^2=F(\tau_t^2)$ and $\tau_a^2$ is the unique solution to $\tau^2=F(\tau^2)$ (see Lemma \ref{R1:1}). 
Similar to the calculation as in Lemma \ref{R1:1},  
we have $F(\tau^2)$ is increasing in $\tau^2$. Hence, the sequence $\{\tau_t^2\}$ determined by the fixed point iterations $\tau_{t+1}^2=F(\tau_t^2)$ is monotonic and bounded with $\tau_0^2=1\leq \tau_t^2\leq \tau_a^2$. It follows that $\tau_t^2$ is convergent. Due to the continuity of $F(\tau^2)$ and the uniqueness of the solution to $\tau^2=F(\tau^2)$, we can conclude that $\tau_t^2\to\tau_a^2$. 
\subsection{Proof of Lemma \ref{difference}}\label{app:convergeB}
We first introduce two auxiliary lemmas  that are important to the proof.
\begin{lemma}[\hspace{-0.001cm}{\cite[Lemma C.1]{LASSO}}]\label{lem:D3}
 Let $Z_1$ and $Z_2$ be jointly Gaussian random variables with $\mathbb{E}[Z_1^2]=\mathbb{E}[Z_2^2]=1$ and
$\mathbb{E}[Z_1Z_2] = c\geq 0$. Let $\mathcal{I}$ be a measurable subset of the real line. Then $\mathbb{P}(Z_1 \in\mathcal{I}, Z_2\in\mathcal{I})$  is an increasing function of $c\in[0,1]$.
\end{lemma}
\begin{lemma}[\hspace{-0.001cm}{\cite[Theorem 4.2]{LASSO}}]\label{lemma:Rst}  Define the following recursion for $\{R_{s,t}\}_{s,t\geq 0}$:
\begin{equation}
\begin{aligned}
R_{0,t}&=1,~~\forall~t\geq 0,\\
R_{s+1,t+1}&=1+\frac{1}{\delta}\mathbb{E}\left[\eta_a(Z_s;\gamma_a)\eta_a(Z_t;\gamma_a)\right],\\
\end{aligned}
\end{equation}
where $(Z_s,Z_t)$ are jointly Gaussian with covariance given by $\mathbb{E}[Z_s^2]=R_{s,s}$, $\mathbb{E}[Z_t^2]=R_{t,t}$, $\mathbb{E}[Z_sZ_t]=R_{s,t}$.
Let $\varphi: \R^2\rightarrow \R$ be a pseudo-Lipschitz function. Then, for all $s,t\geq 0$, 
\begin{equation}
\begin{aligned}
&\lim_{N\to\infty}\frac{1}{N}\sum_{i=1}^N\varphi(x_{s+1,i},x_{t+1,i})\overset{a.s.}{=}\,\mathbb{E}\left[\varphi\left(\eta_a(Z_s;\gamma_a),\eta_a(Z_t;\gamma_a)\right)\right].
\end{aligned}
\end{equation}
\end{lemma}
Now we are ready to prove Lemma \ref{difference}.
By \eqref{AMP_relation} and Proposition \ref{AMP_property} (ii),
\begin{equation}
\begin{aligned}
\lim_{t\to\infty}\lim_{K\to\infty}\frac{1}{K}\|\bz_{t}-\bz_{t-1}\|^2&=\lim_{t\to\infty}\lim_{K\to\infty}\frac{1}{K}\|\bb_{t}-\bb_{t-1}\|^2\\
&\hspace{-0.06cm}\overset{a.s.}{=}\frac{1}{\delta}\lim_{t\to\infty}\lim_{N\to\infty}\frac{1}{N}\|\bq_{t}-\bq_{t-1}\|^2\\
&=\frac{1}{\delta}\lim_{t\to\infty}\lim_{N\to\infty}\frac{1}{N}\|\x_{t+1}-\x_{t}\|^2.
\end{aligned}
\end{equation}
Hence it suffices to prove  Lemma \ref{difference} for variable $\{\x_t\}_{t\geq 0}$.
Applying  Lemma \ref{lemma:Rst} to $\varphi(x,y)=(x-y)^2$, we get  
\begin{equation}
\begin{aligned}
\lim_{N\to\infty}\frac{1}{N}\|\x_{t+1}-\x_t\|^2
&\overset{a.s.}{=}\mathbb{E}\left[(\eta_a(Z_t;\gamma_a)-\eta_a(Z_{t-1};\gamma_a))^2\right]\\
&< \mathbb{E}\left[(Z_t-Z_{t-1})^2\right]\\
&=R_{t,t}+R_{t-1,t-1}-2R_{t,t-1},
\end{aligned}
\end{equation} 
where the inequality follows from the fact that $\eta_a(\cdot;\gamma_a)$ is Lipschitz continuous with Lipschitz constant $(\gamma_a+1)^{-1}$ and $\gamma_a>0$.   
According to the definitions of $R_{s,t}$ in Lemma \ref{lemma:Rst} and $\tau_t^2$ in \eqref{state:AMP}, we have $R_{t,t}=\tau_t^2$ and $R_{t-1,t-1}=\tau_{t-1}^2$, and thus  
$R_{t-1,t-1}\to\tau_a^2$ and $R_{t,t}\to\tau_a^2$ due to Lemma \ref{R1:1.1}. The remaining task is to prove $R_{t,t-1}\to\tau_a^2$.

Note that $\{R_{t+1,t}\}_{t\geq 0}$ is defined through the following recursion:
\begin{equation}\label{recursion:R}
R_{t+1,t}=1+\frac{1}{\delta}\mathbb{E}\left[\eta_a(Z_t;\gamma_a)\eta_a(Z_{t-1};\gamma_a)\right],
\end{equation}
where $(Z_{t-1},Z_t)$ are Gaussian distributed with $\mathbb{E}[Z_{t-1}^2]=\tau_{t-1}^2$, $\mathbb{E}[Z_{t}^2]=\tau_t^2$, and $\mathbb{E}[Z_{t-1}Z_t]=R_{t,t-1}$, i.e., $R_{t+1,t}$ is a function of $\tau_t^2, \tau_{t-1}^2$, and $R_{t,t-1}$. To facilitate the analysis, we define another recursion as follows:
\begin{equation}\label{recursion:tildeR}
\bar{R}_{t+1,t}=1+\frac{1}{\delta}\mathbb{E}\left[\eta_a(\bar{Z}_t;\gamma_a)\eta_a(\bar{Z}_{t-1};\gamma_a)\right],
\end{equation}
with $\bar{R}_{1,0}=1$,
where 
 $(\bar{Z}_{t-1},\bar{Z}_t)$ are jointly Gaussian distributed with $\mathbb{E}[\bar{Z}_{t-1}^2]=\tau_a^2$, $\mathbb{E}[\bar{Z}_{t}^2]=\tau_a^2$, and $\mathbb{E}[\bar{Z}_{t-1}\bar{Z}_t]=\bar{R}_{t,t-1}$. 
 To prove that $R_{t,t-1}\to\tau_a^2,$ we proceed in two steps:
 first we prove that $\bar{R}_{t,t-1}\to \tau_a^2$ and then show $\bar{R}_{t,t-1}-R_{t,t-1}\to 0$.\vspace{0.2cm}
  
 \textbf{Step 1: $\bar{R}_{t,t-1}\to \tau_a^2$}. We claim that $\bar{R}_{t,t-1}\in[1,\tau_a^2]$. The upper bound uses  the Cauchy-Schwarz inequality:
 \begin{equation}
 \bar{R}_{t,t-1}=\mathbb{E}[\bar{Z}_{t}\bar{Z}_{t-1}]\leq \left(\mathbb{E}[\bar{Z}_{t}^2]\,\mathbb{E}[\bar{Z}_{t-1}^2]\right)^{1/2}= \tau_a^2.
 \end{equation} We next prove $\bar{R}_{t,t-1}\geq 1$ by induction. First, $\bar{R}_{1,0}=1$. Assume that $\bar{R}_{t,t-1}\geq 1$, i.e., $\bar{Z}_{t-1}$ and $\bar{Z}_t$ are positively correlated. Then since $\eta_a(\cdot;\gamma_a)$ is increasing, $\mathbb{E}\left[\eta_a(\bar{Z}_{t-1};\gamma_a)\eta_a(\bar{Z}_t;\gamma_a)\right] \geq 1$ \cite{associate}, which further implies that $\bar{R}_{t+1,t}\geq 1$. 
 
 For notation simplicity, let $\bar{y}_t:=\bar{R}_{t,t-1}$. Then $(\bar{Z}_{t-1},\bar{Z}_{t})$ can be represented as 
\begin{equation}
\begin{aligned}
\bar{Z}_{t-1}&=\sqrt{\frac{\tau_a^2}{2}+\frac{\bar{y}_t}{2}}Z-\sqrt{\frac{\tau_a^2}{2}-\frac{\bar{y}_t}{2}}W,~\bar{Z}_t&=\sqrt{\frac{\tau_a^2}{2}+\frac{\bar{y}_t}{2}}Z+\sqrt{\frac{\tau_a^2}{2}-\frac{\bar{y}_t}{2}}W,
\end{aligned}
\end{equation}
where $Z\sim\mathcal{N}(0,1)$ and $W\sim\mathcal{N}(0,1)$ are independent.  Given $y\in[1,\tau_a^2]$, define
\begin{equation}
\begin{aligned}
Z_1(y):&=\sqrt{\frac{\tau_a^2}{2}+\frac{{y}}{2}}Z-\sqrt{\frac{\tau_a^2}{2}-\frac{{y}}{2}}W,~Z_2(y):&=\sqrt{\frac{\tau_a^2}{2}+\frac{{y}}{2}}Z+\sqrt{\frac{\tau_a^2}{2}-\frac{{y}}{2}}W,
\end{aligned}
\end{equation}
 i.e., $Z_1(y)$ and $Z_2(y)$ are Gaussian distributed with 
$\mathbb{E}[Z_1(y)^2]=\mathbb{E}[Z_2(y)^2]=\tau_a^2$ and $\mathbb{E}[Z_1(y)Z_2(y)]=y$.
Let 
\begin{equation}\label{def:ly}
l(y):=1+\frac{1}{\delta}\mathbb{E}\left[\eta_a(Z_1(y);\gamma_a)\eta_a(Z_2(y);\gamma_a)\right],~~y\in[1,\tau_a^2].
\end{equation} Then \eqref{recursion:tildeR} can be expressed as  $\bar{y}_{t+1}=l(\bar{y}_t)$. Clearly, $\tau_a^2$ is a solution to $\tau_a^2=l(\tau_a^2)$ 
 from \eqref{conditiona}. We claim that 
 \begin{equation}\label{deriv:l}
\hspace{-0.1cm}\begin{aligned}
l'(y)&=\frac{1}{\delta}\mathbb{E}\left[\eta_a'(Z_1(y);\gamma_a)\eta_a'(Z_2(y);\gamma_a)\right]=\hspace{-0.1cm}\frac{1}{\delta(\gamma_a+1)^2}\mathbb{P}\bigg(\frac{{Z}_1(y)}{\gamma_a+1}\hspace{-0.05cm}\in\hspace{-0.05cm}[-a, a], \frac{Z_2(y)}{\gamma_a+1}\hspace{-0.05cm}\in\hspace{-0.05cm}[-a,a]\bigg),
\end{aligned}
\end{equation}
which we will show later. 
Hence, $l'(y)\geq 0$ and is increasing in $y\in[1,\tau_a^2]$ due to Lemma \ref{lem:D3}.  It follows that  
\begin{equation}\label{eq:convergey}
\tau_a^2-\bar{y}_{t+1}=l(\tau_a^2)-l(\bar{y}_t)\leq l'(\tau_a^2)(\tau_a^2-\bar{y}_t),
\end{equation}
where the first equation follows from $l(\tau_a^2)=\tau_a^2.$
 In addition, 
\begin{equation}
l'(\tau_a^2)=\frac{1}{\delta(\gamma_a+1)^2}\mathbb{P}\left(\frac{\tau_aZ}{\gamma_a+1} \in [-a,a]\right)<1,
\end{equation}
 where the inequality is  due to \eqref{condition0b}, $\rho>0$, and $\gamma_a>0$. Combining this with \eqref{eq:convergey} and noting that $\bar{y}_t\leq \tau_a^2$ for all $t$, we get $\bar{y}_t\to\tau_a^2$, i.e., $\bar{R}_{t,t-1}\to\tau_a^2$. 

Now we prove our claim in \eqref{deriv:l}.  By simple calculation,
\begin{equation}
\begin{aligned}
l'(y) & =\frac{1}{\delta}\mathbb{E}[\eta_a^{\prime}\left(Z_1(y);\gamma_a\right) \eta_a\left(Z_2(y);\gamma_a\right)Z_3(y)] + \frac{1}{\delta}\mathbb{E}\left[\eta_a\left(Z_1(y);\gamma_a\right) \eta_a'\left(Z_2(y);\gamma_a\right)Z_3(y)\right]\\
&:= T_1+T_2,
\end{aligned}
\end{equation}
where 
\begin{equation}
Z_3(y)=\frac{Z}{4\sqrt{\frac{\tau_a^2}{2}+\frac{y}{2}}}+\frac{W}{4\sqrt{\frac{\tau_a^2}{2}-\frac{y}{2}}}.
\end{equation}
We next calculate $T_1$. For notational simplicity, let $Z_i:=Z_i(y)$,~$i\in\{1,2,3\}$. 
{Clearly, $Z_1, Z_2, Z_3$ are Gaussian distributed with zero mean, and $Z_1$ and $Z_3$ are independent since  $\mathbb{E}[Z_1Z_3]=0$.  For further calculation, it is convenient to rewrite $Z_2$ as a scaling of $Z_1$ plus an independent Gaussian term: 
\begin{equation}\label{Z2}
Z_2=\xi Z_1+\tilde{Z}_2,
\end{equation}
where $\xi=\mathbb{E}[Z_1Z_2]/\mathbb{E}[Z_1^2]=y/\tau_a^2$ and $\tilde{Z}_2$ is independent of $Z_1$ with $\mathbb{E}[\tilde{Z}_2^2]=\tau_a^2-{y^2}/{\tau_a^2}$. 
With the above notations,  $T_1$ can be expressed as 
\begin{equation}
\hspace{-0.2cm}\begin{aligned}
T_1  =\,\frac{1}{\delta}\mathbb{E}\left[\eta_a^{\prime}(Z_1;\gamma_a)  \eta_a(Z_2;\gamma_a)  Z_3\right]
\overset{(a)}{=}&\frac{1}{\delta}\mathbb{E}\left[\mathbb{E}\left[\eta_a^{\prime}(Z_1;\gamma_a)  \eta_a(\xi Z_1+\tilde{Z}_2;\gamma_a)  Z_3|Z_1\right]\right] \\
=\,&\frac{1}{\delta}\mathbb{E}\left[\eta_a^{\prime}(Z_1;\gamma_a) \mathbb{E}\left[ \eta_a(\xi Z_1+\tilde{Z}_2;\gamma_a)  Z_3|Z_1\right]\right]\\
  \overset{(b)}{=}&\frac{1}{\delta}\mathbb{E}\left[\eta_a^{\prime}(Z_1;\gamma_a) \mathbb{E}[\tilde{Z}_2Z_3] \mathbb{E}[\eta_a'(\xi Z_1+\tilde{Z}_2;\gamma_a)|Z_1]\right]\\
\overset{(c)}{=}&\frac{1}{2\delta}  \mathbb{E}\left[\eta_a^{\prime}(Z_1;\gamma_a)  \eta_a^{\prime}(Z_2;\gamma_a)\right],
\end{aligned}
\end{equation}
where (a) uses the law of total probability and \eqref{Z2}, (b) follows from the Stein's lemma and the fact that $\tilde{Z}_2$ and ${Z}_3$ are independent of $Z_1$, and (c) holds since $\mathbb{E}[\tilde{Z}_2Z_3]=\mathbb{E}[(Z_2-\xi Z_1)Z_3]=\mathbb{E}[Z_2Z_3]=\frac{1}{2}$. }
Similarly, we can show that 
\begin{equation}
\begin{aligned}
T_2=\frac{1}{2\delta}  \mathbb{E}\left[\eta_a^{\prime}(Z_1;\gamma_a)  \eta_a^{\prime}(Z_2;\gamma_a)\right].
\end{aligned}
\end{equation}
This gives \eqref{deriv:l} and completes the proof of Step 1.
\vspace{0.1cm}

\textbf{Step 2: $R_{t,t-1}-\bar{R}_{t,t-1}\to0$}.  Let $y_t=R_{t,t-1}$ and 
\begin{equation}
h(\tau_{t-1}^2,\tau_t^2,y_t)=1+\frac{1}{\delta}\mathbb{E}\left[\eta_a(Z_t;\gamma_a)\eta_a(Z_{t-1};\gamma_a)\right].
\end{equation}
 Then the recursion in \eqref{recursion:R} can be written as 
\begin{equation}
y_{t+1}=h(\tau_{t-1}^2,\tau_{t}^2,y_t),\end{equation} and  $l(y)$ defined in \eqref{def:ly} can be expressed as 
 $l(y_t)=h(\tau_a^2,\tau_a^2,y_t)$. Clearly, $y_t\in[1,\tau_a^2]$, where the lower bound can be proved by similar arguments  to that of $\bar{y}_t$, and the upper bound is due to $\tau_t^2\leq \tau_a^2$ (see Appendix \ref{app:convergeA}) and the Cauchy-Schwarz inequality. 
Therefore,  the sequence $\{(\tau_{t-1}^2,\tau_t^2,y_t)\}$ is bounded.  It is easy to check that $h(\tau_{t-1}^2,\tau_t^2,y_t)$ has bounded derivation and is thus Lipschitz continuous on the considered bounded set. We denote the Lipschitz constant by $L_h$.  
In addition, since $l'(\tau_a^2):=c<1$ and $l'(y)$ is increasing in $[1,\tau_a^2]$, we have 
\begin{equation}\label{boundderivative}
l'(y_t)\leq c<1,~l'(\bar{y}_t)\leq c<1,~~\forall ~t.
\end{equation} 
Therefore,
\begin{equation}
\begin{aligned}
\left|y_{t+1}-\bar{y}_{t+1}\right|&=|h(\tau_{t-1}^2,\tau_{t}^2,y_t)-h(\tau_a^2,\tau_a^2,\bar{y}_t)|\\
&\leq |h(\tau_{t-1}^2,\tau_{t}^2,y_t)-h(\tau_a^2,\tau_a^2,{y}_t)|+|h(\tau_{a}^2,\tau_{a}^2,y_t)-h(\tau_a^2,\tau_a^2,\bar{y}_t)|\\
&\overset{(a)}{\leq} L_h\left(|\tau_{t-1}^2-\tau_{a}^2|+|\tau_{t}^2-\tau_{a}^2|\right)+|l(y_t)-l(\bar{y}_t)|\\
&\overset{(b)}{\leq} L_h\left(|\tau_{t-1}^2-\tau_{a}^2|+|\tau_{t}^2-\tau_{a}^2|\right)+c|y_t-\bar{y}_t|\\
&\overset{(c)}{\leq}  L_h\sum_{i=1}^{t}c^{i-1}\left(|\tau_{t-i}^2-\tau_{a}^2|+|\tau_{t-i+1}^2-\tau_{a}^2|\right)+c^{t}|y_{1}-\bar{y}_{1}|\\
&=L_h\sum_{i=1}^{\lfloor {t}/{2}\rfloor}c^{i-1}(|\tau_{t-i}^2-\tau_{a}^2|+|\tau_{t-i+1}^2-\tau_{a}^2|)\\
&~~\,~+L_h\hspace{-0.25cm}\sum_{i=\lfloor {t}/{2}\rfloor+1}^{t}\hspace{-0.25cm}c^{i-1}(|\tau_{t-i}^2\hspace{-0.05cm}-\hspace{-0.05cm}\tau_{a}^2|+\hspace{-0.05cm}|\tau_{t-i+1}^2-\tau_{a}^2|)\hspace{-0.05cm}+\hspace{-0.05cm}c^{t}|y_{1}-\bar{y}_{1}|\\
&\leq \frac{2L_h}{1-c}\sup_{\lceil {t}/{2}\rceil\leq i\leq t}|\tau_i^2-\tau_a^2|+\frac{2L_hc^{\lfloor {t}/{2}\rfloor}}{1-c}\sup_{0\leq i\leq \lceil {t}/{2}\rceil-1}|\tau_i^2-\tau_a^2|+c^{t}|y_{1}-\bar{y}_{1}|
\end{aligned}
\end{equation}
where (a) uses the Lipschitz continuity of $h$ and the definition of $l$, (b) follows from \eqref{boundderivative}  and the increasing of $l'$, and (c) is obtained by unfolding the inequality in (b) recursively. 
Let $t\to\infty$ in the above inequality and noting $c<1$ and $\lim_{t\to\infty}\tau_t^2=\tau_a^2$, we have $y_{t}-\bar{y}_t\to0$, i.e., ${R}_{t,t-1}-\bar{R}_{t,t-1}\to 0,$ which completes the proof.

\subsection{Proof of Proposition \ref{prob:converge}}\label{app:probconverge}
We now proceed to the proof of Proposition \ref{prob:converge}. {In contrast to \cite{LASSO, elasticnet}, the problem in \eqref{def:xa} involves constraints. To quantify the distance between a given point and the optimal solution of a constrained optimization problem like \eqref{def:xa}, we adopt the following error bound condition.
\begin{lemma}\label{Errorbound}
Consider the following problem:
\begin{equation}
\x^*=\arg\min_{\x\in\mathcal{C}}~f(\x),
\end{equation}
where $f(\x)$ is $\mu$-strongly convex with  an $L$-Lipschitz continuous gradient, and $\mathcal{C}$ is  a nonempty, compact, and  convex set. The following error bound condition holds for all $\omega>0$:
\begin{equation}
\|\x-\x^*\|\leq C_\omega\|\omega^{-1}(\x-\mathcal{P}_\mathcal{C}(\x-\omega\nabla f(\x))\|,~~\forall~\x\in\mathcal{C},
\end{equation}
where $C_\omega=(2\mu^{-1}+\omega)(1+\omega L)$.
\end{lemma}
\begin{proof}
This lemma follows directly from \cite[Corollary 3.6]{errorbound}, and we omit the details.
\end{proof}
Applying Lemma \ref{Errorbound} with $f(\x)=\frac{1}{2}\|\s-\bH\x\|^2+\frac{\rho}{2}\|\x\|^2$ and $\mathcal{C}=[-a,a]^N$, we have 
\begin{equation}\label{errorbound}
\begin{aligned}
&\frac{\|\x_t-\x_a\|}{\sqrt{N}}\leq \frac{\omega^{-1}C_\omega}{\sqrt{N}} \left\|\x_t-\mathcal{P}_{[-a,a]^N}((1-\omega\rho)\x_t+\omega\H^T(\s-\H\x_t))\right\|,
\end{aligned}
\end{equation}
where $C_{\omega}=(\rho^{-1}+\omega)\left(1+\omega(\|\bH\|^2+\rho)\right)$ by noting that $\|\nabla ^2f(\x)\|=\|\bH^T\bH+\rho\mathbf{I}\|\leq \|\bH\|^2+\rho$, i.e., $\|\bH\|^2+\rho$ is a Lipschitz constant of $\nabla f(\x)$.  With  \eqref{errorbound}, we can further bound $\|\x_t-\x_a\|/\sqrt{N}$ as 
\begin{equation}\label{errorbound2}
\begin{aligned}
&\frac{\|\x_t-\x_a\|}{\sqrt{N}}\leq \frac{\omega^{-1}C_\omega}{\sqrt{N}} \|\x_t-\x_{t+1}\|+ \frac{\omega^{-1}C_\omega}{\sqrt{N}}\|\x_{t+1}\hspace{-0.05cm}-\hspace{-0.05cm}\mathcal{P}_{[-a,a]^N}((1-\omega\rho)\x_t+\omega\H^T(\s-\H\x_t))\|.
\end{aligned}
\end{equation}
For any $\omega>0$, \begin{equation}\label{Comega}\lim_{N\to\infty}C_{\omega}=(\rho^{-1}+\omega)\bigg(1+\omega(1+1/\sqrt{\delta})^2+\omega\rho\bigg)<\infty\end{equation} due to Lemma \ref{lem:eigenvalue} (see further ahead).  Hence, according to Lemma \ref{difference},  the first term on the right-hand side of \eqref{errorbound2} satisfies 
\begin{equation}\label{firstterm}\lim_{t\to\infty}\lim_{N\to\infty} \frac{\omega^{-1}C_\omega}{\sqrt{N}} \|\x_t-\x_{t+1}\|\overset{a.s.}{=}0,~~\forall~ \omega>0.
\end{equation}
For the second term, we note that $\x_{t+1}$ can be rewritten as follows using the formula of the AMP algorithm in \eqref{AMP_copy}: 
\begin{equation}\label{xt+1}
\begin{aligned}
\x_{t+1}\hspace{-0.05cm}=\hspace{-0.05cm}\mathcal{P}_{[-a,a]^N}\bigg(&\frac{\x_t}{1+\gamma_a}\hspace{-0.05cm}+\hspace{-0.05cm}\frac{\H^T(\s-\H\x_t)}{(1\hspace{-0.05cm}+\hspace{-0.05cm}\gamma_a)(1\hspace{-0.05cm}-\hspace{-0.05cm}\alpha_{t})}+\frac{\alpha_{t}\H^T(\z_{t-1}-\z_t)}{(1+\gamma_a)(1-\alpha_{t})}\bigg),
\end{aligned}
\end{equation}
where \begin{equation}\label{def:alphat}
\alpha_{t}:=\frac{1}{\delta}\left<\eta_a'(\x_{t-1}+\H^T\bz_{t-1};\gamma_a)\right>.
\end{equation} 
By Proposition \ref{pro:4} and Lemma \ref{R1:1.1},  
\begin{equation}\label{alpha_to_alpha_a}
\hspace{-0.15cm}\begin{aligned}
\lim_{t\to\infty}\lim_{N\to\infty}\alpha_{t}&=\lim_{t\to\infty}\frac{1}{\delta(\gamma_a+1)}\mathbb{P}\bigg(\frac{\tau_t Z}{\gamma_a+1}\in[-a,a]\bigg)\\
&=\frac{1}{\delta(\gamma_a+1)}\mathbb{P}\bigg(\frac{\tau_a Z}{\gamma_a+1}\in[-a,a]\bigg)\\
&\hspace{-0.1cm}:=\alpha_a.
\end{aligned}
\end{equation}
To bound the second term on the right-hand side of \eqref{errorbound2}, we set $\omega=\frac{1}{(1+\gamma_a)(1-\alpha_a)}$, which yields} 
\begin{equation}\label{secondterm}
\begin{aligned} 
&\frac{1}{\sqrt{N}}\left\|\x_{t+1}-\mathcal{P}_{[-a,a]^N}\left((1-\omega\rho)\x_t+\omega\H^T(\s-\H\x_t)\right)\right\|\\
&\overset{(a)}{=}\frac{1}{\sqrt{N}} \left\|\x_{t+1}-\mathcal{P}_{[-a,a]^N}\bigg(\frac{\x_t}{1+\gamma_a}+\frac{\H^T(\s-\H\x_t)}{(1+\gamma_a)(1-\alpha_a)}\bigg)\right\|\\
&\overset{(b)}{\leq}\frac{1}{1+\gamma_a}\frac{\|\H^T(\s-\H\x_t)\|}{\sqrt{N}}\left(\frac{1}{1-\alpha_{t}}-\frac{1}{1-\alpha_a}\right)+\frac{\alpha_{t}}{(1+\gamma_a)(1-\alpha_{t})}\frac{\|\H^T(\bz_{t-1}-\bz_t)\|}{\sqrt{N}},
\end{aligned}
\end{equation}
where (a) is obtained by substituting $\omega=\frac{1}{(1+\gamma_a)(1-\alpha_a)}$ and applying \eqref{condition0b}, and (b) uses \eqref{xt+1} and the non-expansiveness of  the projection operator. Note that 
\begin{equation}
\begin{aligned}
\lim\sup_{N\to\infty}\frac{\|\H^T(\s-\H\x_t)\|}{\sqrt{N}}&\leq \lim\sup_{N\to\infty}\|\H\|\left(\frac{\|\s\|}{\sqrt{N}}+\frac{\|\H\|\|\x_t\|}{\sqrt{N}}\right)\\
&\leq \left(1+\frac{1}{\sqrt{\delta}}\right)\left(\sqrt{\delta}{B_s}+\left(1+\frac{1}{\sqrt{\delta}}\right) a\right)<\infty,
\end{aligned}
\end{equation}
where 
\begin{equation}\label{def:Bs}B_s:=\max_{s\in\mathcal{S}}|s| \end{equation}
is finite since $\mathcal{S}$ is a finite set,  the second inequality uses Lemma \ref{lem:eigenvalue}  and the facts that {$\s\in\mathcal{S}^K$} and $\x\in[-a,a]^N$.  Combining this with \eqref{Comega}, \eqref{alpha_to_alpha_a}, \eqref{secondterm}, and  Lemma \ref{difference}, we get 
\begin{equation}\label{xttoxa2}
\hspace{-0.1cm}\begin{aligned}
&\lim_{t\to\infty}\hspace{-0.05cm}\lim_{N\to\infty}\hspace{-0.05cm}\frac{1}{\sqrt{N}}\hspace{-0.05cm}\|\x_{t+1}\hspace{-0.1cm}-\hspace{-0.1cm}\mathcal{P}_{[-a,a]^{\hspace{-0.03cm}N}}\hspace{-0.03cm}((1\hspace{-0.05cm}-\hspace{-0.05cm}\omega\rho)\x_t\hspace{-0.05cm}+\hspace{-0.05cm}\omega\H^T\hspace{-0.05cm}(\s\hspace{-0.05cm}-\hspace{-0.05cm}\H\x_t))\|\overset{a.s.}=0
\end{aligned}
\end{equation}
for  $\omega=\frac{1}{(1+\gamma_a)(1-\alpha_a)}$, which, together with \eqref{errorbound2} and \eqref{firstterm}, 
 gives Proposition \ref{prob:converge} (i).

 Let $\varphi(\cdot)$ be pseudo-Lipschitz  with constant $L_{\varphi}$, then
\begin{equation}
\begin{aligned}
\left|\frac{1}{N}\sum_{i=1}^N\varphi(x_{a,i})-\frac{1}{N}\sum_{i=1}^N\varphi(x_{t+1,i})\right|&\leq \frac{L_{\varphi}}{N}\sum_{i=1}^N(1+|x_{t+1,i}|+|x_{a,i}|)|x_{a,i}-x_{t+1,i}|\\
&\leq L_{\varphi}(1+2a)\frac{\|\x_a-\x_{t+1}\|}{\sqrt{N}},
\end{aligned}
\end{equation}
where the last inequality follows from $|x_{t+1,i}|\leq a,$ $|x_{a,i}|\leq a$, and  $\frac{\|\x\|_1}{\sqrt{N}}\leq \|\x\|$.
Combining the above inequality with \eqref{xttoxa2}, we get
\begin{equation}
\begin{aligned}
\lim_{N\to\infty}\frac{1}{N}\sum_{i=1}^N \varphi(x_{a,i})&\overset{a.s.}{=}\lim_{t\to\infty}\lim_{N\to\infty}\frac{1}{N}\sum_{i=1}^N\varphi(x_{t+1,i})\\
&\underset{a.s.}{\overset{(a)}{=}}\lim_{t\to\infty}\mathbb{E}\left[\varphi\circ\eta_a(\tau_t Z;\gamma_a)\right]\\
&\overset{(b)}{=}\mathbb{E}\left[\varphi\circ\eta_a(\tau_aZ;\gamma_a)\right]=\mathbb{E}\left[\varphi(X_a)\right],
\end{aligned}
\end{equation}
where $(a)$ is due to Proposition \ref{pro:4} (ii), $(b)$ uses Lemma \ref{R1:1.1} and  the fact that $\varphi\circ\eta_a(\cdot;\gamma_a)$ is bounded and continuous such that the dominated  convergence theorem can be applied. This completes the proof of Proposition \ref{prob:converge} (ii).

\section{Corollaries of Proposition \ref{prob:converge}}
This appendix collects two  important corollaries of Proposition \ref{prob:converge}. In Appendix \ref{app:unique}, we prove  the uniqueness of $(\tau_a^2,\gamma_a)$ defined in Lemma \ref{lemma:unique}. In Appendix \ref{app:convergefN}, we give the proof of    Lemma \ref{convergefN}.

\subsection{Proof of Lemma \ref{lemma:unique}: Uniqueness of $(\tau_a^2,\gamma_a)$ }\label{app:unique}
We note that Proposition \ref{prob:converge} holds for any AMP algorithm defined by \eqref{AMP_copy} with $\gamma_a$ being a solution to \eqref{condition0}. 
Applying Proposition \ref{prob:converge} (ii) with $\varphi(x)=x^2$, we get 
\begin{equation}
\lim_{N\to\infty}\frac{\|\x_a\|^2}{N}\overset{a.s.}{=}\mathbb{E}\left[\eta_a^2(\tau_a Z;\gamma_a)\right]=(\tau_a^2-1)\delta,
\end{equation}
where the last equality follows from \eqref{conditiona}. Due to the uniqueness of $\x_a$ (as problem \eqref{def:xa} is strongly convex), $\tau_a^2$ is unique. 

Next, we prove the uniqueness of $\gamma_a$. Assume for contradiction that there exists $\gamma_{a,1}\neq\gamma_{a,2}$ such that $(\tau_a^2,\gamma_{a,1})$ and $(\tau_a^2,\gamma_{a,2})$ both satisfy \eqref{condition0}. Without loss of generality, we assume that $\gamma_{a,1}<\gamma_{a,2}$.  Following the notations in Appendix  \ref{app:lemma1}, \eqref{condition0b} can be written as $F_2(\tau^2,\gamma;a)=0$; see \eqref{condition0:2b}. According to \eqref{deriv},  
\begin{equation}
\begin{aligned}
\frac{\partial F_2}{\partial \gamma}(\tau^2,\gamma;a)&=-1+\frac{1}{\delta(\gamma+1)^2}u+\frac{2a\gamma}{\delta\sqrt{\tau^2}(\gamma+1)}v\\
&=-1+\frac{1}{\delta(\gamma+1)}u-\frac{\gamma}{\delta(\gamma+1)^2}\left(u-\frac{2a(\gamma+1)}{\sqrt{\tau^2}}v\right)\\
&\leq -1+\frac{1}{\delta(\gamma+1)}u,
\end{aligned}
\end{equation}
where $u$ and $v$ are defined in \eqref{def:uv} and the last inequality holds since $2\Phi(z)-2z\phi(z)-1\geq 0$ for all $z\geq 0$. Furthermore, it is easy to check that the function $\frac{u}{\delta(\gamma+1)}$ is  decreasing in $\gamma$. Hence, for $\gamma>\gamma_{a,1}$, we have 
\begin{equation}
\begin{aligned}
\frac{\partial F_2}{\partial \gamma}(\tau_a^2,\gamma;a)\leq -1+\frac{1}{\delta(\gamma_{a,1}+1)}u=-\frac{\rho}{\gamma_{a,1}}<0,
\end{aligned}
\end{equation}
where the equality uses the fact that $F_2(\tau_a^2,\gamma_{a,1};a)=0$. It follows that $F_2(\tau_a^2,\gamma;a)$ is strictly decreasing in $\gamma>\gamma_{a,1}$, and thus $F_2(\tau_a^2,\gamma_{a,2};a)<F_2(\tau_a^2,\gamma_{a,1};a)=0$, which gives a contradiction. 
\subsection{Proof of Lemma \ref{convergefN}}\label{app:convergefN}
The function $f_N(a;\s,\H)$ can be expressed as 
\begin{equation}
f_N(a;\s,\H)=\frac{1}{N}\|\s-\H\x_a\|^2+\frac{\rho}{N}\|\x_a\|^2.
\end{equation}
Applying Proposition \ref{prob:converge} (ii),  we have
\begin{equation}
\lim_{N\to\infty}\frac{\|\x_a\|^2}{N}\overset{a.s.}{=}\mathbb{E}\left[\eta_a^2(\tau_a Z;\gamma_a) \right]=\delta(\tau_a^2-1),
\end{equation}
where the last equality is due to \eqref{conditiona}. Note that 
\begin{equation}
\begin{aligned}
\left|\frac{1}{N}\|\s-\H\x_a\|^2-\frac{1}{N}\|\s-\bH\x_{t+1}\|^2\right|&{=}\left|\frac{2}{N}(\bH\x_\mu-\s)^T\bH(\x_a-\x_{t+1})\right| \\
&{\leq}\,\frac{2\|\bH\|(\|\bH\|\|\x_{\mu}\|+\|\s\|)}{\sqrt{N}}\frac{\|\x_a-\x_{t+1}\|}{\sqrt{N}},
\end{aligned}
\end{equation}
where $\x_\mu=\mu(1-\x_a)+(1-\mu)\x_{t+1}$ for some $\mu\in[0,1]$. Hence, according to Lemma \ref{difference} and Lemma \ref{lem:eigenvalue}, and by noting that    ${\|\s\|\leq\sqrt{K}B_s}$ and  $\|\x_{\mu}\|\leq \sqrt{N}a$, {where $B_s$ is defined in \eqref{def:Bs}}, we have
\begin{equation}\lim_{t\to\infty}\lim_{N\to\infty}\left|\frac{1}{N}\|\s-\H\x_a\|^2-\frac{1}{N}\|\s-\bH\x_{t+1}\|^2\right|=0.\end{equation}
It follows that 
\begin{equation}
\hspace{-0.2cm}\begin{aligned}
\lim_{N\to\infty}\frac{1}{N}\|\s-\H\x_a\|^2&\overset{a.s.}{=}\delta\lim_{t\to\infty}\lim_{K\to\infty}\frac{1}{K}\|\s-\H\x_{t+1}\|^2\\
&\overset{(a)}{=}\delta\lim_{t\to\infty}\lim_{K\to\infty}\frac{1}{K}\|\bz_{t+1}-\alpha_t\bz_t\|^2\\
&=\delta\lim_{t\to\infty}\lim_{K\to\infty}\bigg(\frac{\|\bz_{t+1}\|^2}{K}+\frac{\alpha_t^2\|\bz_t\|^2}{K}-2\alpha_t\left<\bz_{t+1},\bz_t\right>\bigg)\\
&\underset{a.s.}{\overset{(b)}{=}}\delta(\alpha_a-1)^2\tau_a^2\\
&\overset{(c)}{=}\delta \rho^2\tau_a^2/\gamma_a^2,
\end{aligned}
\end{equation}
where (a) uses the update rule of the AMP algorithm in \eqref{AMP_copy} and the definition of $\alpha_t$ in \eqref{def:alphat}, 
 (b) holds since  $\lim_{t\to\infty}\lim_{K\to\infty} \alpha_t\overset{a.s.}{=}\alpha_a$ (see \eqref{alpha_to_alpha_a}), 
\begin{equation}
\lim_{t\to\infty}\lim_{K\to\infty}\hspace{-0.05cm}\left<\bz_{t+1},\bz_t\right>\hspace{-0.05cm}=\hspace{-0.05cm}\lim_{t\to\infty}\lim_{K\to\infty}\hspace{-0.1cm}\frac{\|\bz_{t}\|^2}{K},\end{equation}
  which follows from Lemma \ref{difference}, and
\begin{equation}
\begin{aligned}
&\lim_{t\to\infty}\lim_{K\to\infty}\frac{\|\bz_{t}\|^2}{K}=\lim_{t\to\infty}1+\sigma_t^2=\lim_{t\to\infty}\tau_t^2=\tau_a^2,
\end{aligned}
\end{equation}
which follows from  Proposition \ref{pro:4} (ii) and Lemma \ref{R1:1.1}, and  (c) applies   \eqref{condition0b}.
Combining the above discussions gives the desired result.

\section{Proof of Lemma \ref{converge:a}, Theorem \ref{mainresult:ed}, and Proposition \ref{the1}}\label{app:the1}
\subsection{Proof of Lemma \ref{converge:a}}\label{proof:Lemma4}
\subsubsection{Useful Results}Before proving Lemma \ref{converge:a}, we first list a few results that will be used in the proof. 
\begin{lemma}[
{\hspace{-0.001cm}\cite[Section 3.2.5]{convexopt}}]\label{lem:convex}
 Let $f(\x,\y)$ be a  jointly convex function in $(\x, \y)$ and $\mathcal{C}$ be a convex nonempty set. The function 
 \begin{equation}
 g(\x) := \inf_{\y\in\mathcal{C}} f(\x,\y)
 \end{equation}
  is convex in $\x$.
\end{lemma}
\begin{lemma}[\hspace{-0.001cm}{\cite[Theorem 2.7]{convex_lemma}}]\label{lem:convergeastar}
Let$\{f_N(a)\}_{N\geq0}$ be a sequence of convex functions that converges pointwise almost everywhere (a.e.) to a  function $f(a)$ on a convex set $\mathcal{C}$, i.e., 
\begin{equation}
\lim_{N\to\infty}f_N(a)\overset{a.s.}{=}f(a),~~\forall~a\in\mathcal{C}.
\end{equation} 
Let $\hat{a}_N\in\arg\min_{a\in\mathcal{C}}f_N(a).$ 
If  (i) $f(a)$ is uniquely minimized at $a^*$ on $\mathcal{C}$; (ii) $a^*$ is an element of the interior of $\mathcal{C}$, 
 then \begin{equation}
 \lim_{N\to\infty}\hat{a}_N=a^*.
 \end{equation}
\end{lemma}
\begin{lemma}[\hspace{-0.001cm}{\cite[Theorem 9.1]{RMTbook}}]\label{lem:eigenvalue}
 Under Assumption \ref{ass} (i) -- (ii), the following results hold for the extreme eigenvalues of $\bH^T\bH$: \\
 (i) The largest eigenvalue satisfies 
 \begin{equation}
 \lim_{N\to\infty}\lambda_{\max}(\bH^T\bH)\overset{a.s.}{=}\bar{\lambda}_{\max}:=\left(1+\frac{1}{\sqrt{\delta}}\right)^2.
 \end{equation}
 (ii) The smallest eigenvalue satisfies 
 \begin{equation}
 \lim_{N\to\infty}\lambda_{\min}(\H^T\H)\overset{a.s.}{=}\bar{\lambda}_{\min}:=\left(\max\left\{0,1-\frac{1}{\sqrt{\delta}}\right\}\right)^2.
 \end{equation}
\end{lemma}
\subsubsection{Proof of Lemma \ref{converge:a}}
Let $\mathcal{F}:=\{(\x,a)\mid -a\leq x_i\leq a,~a\geq 0,~\forall \,i\}$. 
Then $f_N(a;\s,\H)$ can be equivalently expressed as 
\begin{equation}\label{def:fN2}f_N(a;\s,\H)=\min_{\x}~\frac{1}{N}\|\s-\H\x\|_2^2+\frac{\rho}{N}\|\x\|_2^2+\mathbb{I}_\mathcal{F}(\x,a),
\end{equation}where $\mathbb{I}_{\mathcal{F}}$ denotes the indicator function of set $\mathcal{F}$. Since $\mathcal{F}$ is a convex set,  $\mathbb{I}_{\mathcal{F}}$ is jointly convex in $(\x,a)$. Therefore,   $f_N(a;\s,\H)$ is convex in $a$ according to Lemma \ref{lem:convex}. 
According to Lemma \ref{lemma:astar}, $f(a)$ is strongly convex on $[0,\infty)$, whose unique minimizer $a^*$ satisfies $a^*>0$.   By Lemma \ref{lem:convergeastar}, we get the first result of Lemma \ref{converge:a}.

  We next prove the second result of Lemma  \ref{converge:a}.  In particular, we prove that  for any $a_1$ and $a_2$, the following inequality holds for sufficiently large $N$: 
\begin{equation}\label{eq:diffxa1xa2}
\frac{1}{\sqrt{N}}\|\x_{a_1}-\x_{a_2}\|\leq C|a_1-a_2|,~a.s.,
\end{equation}
 where $\x_{a_1}$ and $\x_{a_2}$ are the solutions to \eqref{def:xa} with $a=a_1$ and $a=a_2$, and $C$ is a constant independent of $N$. The desired result then follows immediately by replacing $a_1$ and $a_2$ in \eqref{eq:diffxa1xa2} with $\hat{a}_N$ and $a^*$ and using $\lim_{N\to\infty}\hat{a}_N\overset{a.s.}{=}a^*.$  Now we prove \eqref{eq:diffxa1xa2}. According to the first-order optimality condition of problem \eqref{def:xa},  the following holds for any $\omega>0$:
\begin{equation}
\begin{aligned}
{\x}_{a_1}&\hspace{-0.05cm}=\hspace{-0.05cm}\mathcal{P}_{[-{a}_1,{a}_1]^N}((1\hspace{-0.05cm}-\hspace{-0.05cm}\omega\rho)\x_{a_1}\hspace{-0.05cm}-\hspace{-0.05cm}\omega\H^T\H{\x}_{a_1}\hspace{-0.05cm}+\hspace{-0.05cm}\omega\H^T\s),\\
\x_{a_2}&\hspace{-0.05cm}=\hspace{-0.05cm}\mathcal{P}_{[-{a}_2,{a}_2]^N}((1\hspace{-0.05cm}-\hspace{-0.05cm}\omega\rho)\x_{a_2}\hspace{-0.05cm}-\hspace{-0.05cm}\omega\H^T\H{\x}_{a_2}\hspace{-0.05cm}+\hspace{-0.05cm}\omega\H^T\s).
\end{aligned}
\end{equation}
Hence, \begin{equation}
\begin{aligned}
\|{\x}_{a_1}-\x_{a_2}\|
&\leq \big\|\mathcal{P}_{[-{a}_1,{a}_1]^N}((1-\omega\rho)\x_{a_1}-\omega\H^T\H{\x}_{a_1}+\omega\H^T\s)\\
&\qquad-\mathcal{P}_{[-{a}_2,{a}_2]^N}((1-\omega\rho)\x_{a_1}-\omega\H^T\H{\x}_{a_1}+\omega\H^T\s)\big\|\\
&\quad+\big\|\mathcal{P}_{[-{a}_2,{a}_2]^N}((1-\omega\rho)\x_{a_1}-\omega\H^T\H{\x}_{a_1}+\omega\H^T\s)\\
&\qquad~\,-\mathcal{P}_{[-{a}_2,{a}_2]^N}((1-\omega\rho)\x_{a_2}-\omega\H^T\H{\x}_{a_2}+\omega\H^T\s)\big\|\\
&\leq\hspace{-0.1cm} \sqrt{N({a}_1-a_2)^2}+\hspace{-0.05cm}\|(1-\omega\rho)\mathbf{I}-\omega\H^T\H\|\|{\x}_{a_1}\hspace{-0.05cm}-\hspace{-0.05cm}\x_{a_2}\|,
\end{aligned}
\end{equation}
where the last inequality uses the fact that for all $\x\in\R^{N}$, 
\begin{equation}
\|\mathcal{P}_{[-a_1,{a}_1]^N}(\x)-\mathcal{P}_{[-a_2,a_2]^N}(\x)\|\leq \sqrt{N({a}_1-a_2)^2},
\end{equation} and the non-expansiveness  of the projection operator. 
  It follows that  
\begin{equation}\label{operator}
(1-\|(1-\omega\rho)\mathbf{I}-\omega\H^T\H\|)\frac{\|{\x}_{a_1}-\x_{a_2}\|}{\sqrt{N}}\leq  \left|{a}_1-a_2\right|,~\forall~\omega>0.
\end{equation}
According to Lemma \ref{lem:eigenvalue}, 
\begin{equation}\label{eq:coefficient}
\begin{aligned}
&\lim_{N\to\infty}\|(1-\omega\rho)\mathbf{I}-\omega\H^T\H\|\overset{a.s.}{=}\max\left\{\left|1-\omega\rho-\omega\bar{\lambda}_{\min}\right|,\left|\omega\bar{\lambda}_{\max}-(1-\omega\rho)\right|\right\}.
\end{aligned}
\end{equation}
Let 
$\omega=\left(\bar{\lambda}_{\max}+\rho\right)^{-1}.$
Then 
\begin{equation}\lim_{N\to\infty}\|(1-\omega\rho)\mathbf{I}-\omega\H^T\H\|\overset{a.s.}{=}1-\frac{\rho+\bar{\lambda}_{\min}}{\rho+\bar{\lambda}_{\max}}<1,
\end{equation}  which, together with \eqref{operator}, implies that for sufficiently large $N$, \eqref{eq:diffxa1xa2} holds with 
$C=\frac{2(\rho+\bar{\lambda}_{\max})}{\rho+\bar{\lambda}_{\min}}.$ This completes the proof of the second result.

\subsection{Proof and Extension of Theorem \ref{mainresult:ed}}\label{proof:theorem1}
With all the results obtained so far, we are now ready to prove Theorem \ref{mainresult:ed}.
\subsubsection{Proof of Theorem \ref{mainresult:ed}}
Since $\varphi(x)$ is pseudo-Lipschitz, we have
\begin{equation}\label{phixhat}
\begin{aligned}
\left|\frac{1}{N}\sum_{i=1}^N \varphi(\hat{x}_i)-\frac{1}{N}\sum_{i=1}^N\varphi(x^*_i)\right|&\leq\frac{L_{\varphi}}{N}\sum_{i=1}^N\left(1+|\hat{x}_i|+|x_i^*|\right)|\hat{x}_i-x_i^*| \\
&\leq L_\varphi\left(1+\hat{a}_N+{a^*}\right) \frac{1}{N}\sum_{i=1}^N|\hat{x}_i-x_i^*| \\
&\leq L_\varphi\left(1+\hat{a}_N+{a^*}\right)\frac{\|\hat{\x}-\x^*\|}{\sqrt{N}}.
\end{aligned}\end{equation}
By letting $N\to\infty$ and applying Lemma \ref{converge:a} to \eqref{phixhat} and noting Proposition \ref{prob:converge} (ii), we obtain Theorem \ref{mainresult:ed}. 
\subsubsection{From pseudo-Lipschitz to discontinuous function} \label{app:nonsmooth}
In this part, we show that Theorem \ref{mainresult:ed} can encompass a broader class of test function $\varphi(\cdot)$. Specifically, 
assume that $\varphi:\R\rightarrow \R$ is continuous except at a finite of points $\{x_1,x_2,\dots, x_M\}$, where $x_1,x_2,\dots, x_M\notin\{-a^*,a^*\}$.
  We next prove that 
\begin{equation}\label{converge:phi0}
\begin{aligned}
\lim_{N\to\infty}\frac{1}{N}\sum_{i=1}^N\varphi(\hat{x}_i)= \mathbb{E}\left[\varphi(\hat{X})\right],
\end{aligned}
\end{equation}
where $\hat{X}$ is defined in \eqref{def:hatX}.
 Note that for a function with finite discontinuous points, those points are isolated. Hence, there exists $\zeta>0$ such that $x_1,x_2,\dots,x_M\notin [-a^*-\zeta,-a^*+\zeta]\cup[a^*-\zeta,a^*+\zeta]$. Since $\hat{a}_N\xrightarrow{a.s.} a^*$, there exists $N_0>0$ such that 
\begin{equation}
|\hat{x}_i|\leq \hat{a}_N\leq a^*+\zeta/2,~i=1,2,\dots, N,~\forall\, N\geq N_0
\end{equation}
 with probability one.
Let $\varphi_0=\varphi|_{[-a^*-\zeta/2,a^*+\zeta/2]}$. It suffices to prove that 
\begin{equation}\begin{aligned}
\lim_{N\to\infty}\frac{1}{N}\sum_{i=1}^N\varphi_0(\hat{x}_i)= \mathbb{E}\left[\varphi_0(\hat{X})\right].
\end{aligned}
\end{equation}

In the following, we will prove \eqref{converge:phi0}. First, if $\varphi_0(x)$ is continuously differentiable, i.e., $x_1,x_2,\dots, x_M\notin\text{dom}\,{\varphi_0}:= [-a^*-\zeta/2,a^*+\zeta/2]$, then $\varphi_0(x)$ is Lipschitz continuous as $\text{dom}\,\varphi_0$ is compact. In this case, \eqref{converge:phi0} follows immediately from Theorem \ref{mainresult:ed}.

Now assume that  $\varphi_0$ is discontinuous.  Without loss of generality, we assume that there is a single discontinuous point of $\varphi_0$, denoted as $\omega$, and that  $\lim_{x\to\omega^-}\varphi_0(x)<\lim_{x\to\omega^+}\varphi_0(x)$. By assumption, $|\omega|\leq a^*-\zeta$. Consider the following lower and upperbound functions of $\varphi_0$ parameterized by $0<\epsilon<\zeta$, where $\varphi_\omega^{-}:=\lim_{x\to\omega^-}\varphi_0(x)$ and $\varphi_\omega^{+}:=\lim_{x\to\omega^+}\varphi_0(x)$:
\begin{equation}
\underline{\varphi}_{\epsilon}(x)=\left\{
\begin{aligned}
\varphi_{\omega}^{-}+\frac{\varphi_{\omega}^+-\varphi_{\omega}^-}{\epsilon}(x-\omega),~&\omega<x<\omega+\epsilon;\\
\varphi_0(x),~~~~~~~~~~&\text{otherwise},\end{aligned}
\right.
\end{equation}
\begin{equation}
\hspace{-0.05cm}\bar{\varphi}_{\epsilon}(x)=\left\{
\begin{aligned}
\varphi_{\omega}^{-}+\frac{\varphi_{\omega}^+-\varphi_{\omega}^-}{\epsilon}(x-\omega+\epsilon),~&\omega-\epsilon<x<\omega;\\
\varphi_0(x),~~~~~~~~~~~~~&\text{otherwise}.\end{aligned}
\right.
\end{equation}
Clearly,   both $\underline{\varphi}_{\epsilon}$ and $\bar{\varphi}_{\epsilon}$ are continuous. According to our previous discussions, 
\begin{equation}
\lim_{N\to\infty}\frac{1}{N}\sum_{i=1}^N\underline{\varphi}_{\epsilon}(\hat{x}_i)\overset{a.s.}{=}\mathbb{E}\left[\underline{\varphi}_{\epsilon}(\hat{X})\right],
\end{equation}
\begin{equation}
\lim_{N\to\infty}\frac{1}{N}\sum_{i=1}^N\bar{\varphi}_{\epsilon}(\hat{x}_i)\overset{a.s.}{=}\mathbb{E}\left[\bar{\varphi}_{\epsilon}(\hat{X})\right].
\end{equation}
In addition, since  $(\omega,\omega+\epsilon)\subseteq(-a^*,a^*)$ and noting the definition $\hat{X}=\eta_{a^*}(\frac{\tau_*Z}{\gamma_*+1})$, we have
\begin{equation}
\begin{aligned}
0&<\mathbb{E}\left[{\varphi}_{0}(\hat{X})\right]-\mathbb{E}\left[\underline{\varphi}_{\epsilon}(\hat{X})\right]\\
&<\mathbb{E}\left[\sup_{x\in\text{dom}\,\varphi_0}\left|{\varphi}_{0}(x)-\underline{\varphi}_{\epsilon}(x)\right|\mathbf{1}_{\{\frac{\tau_*Z}{\gamma_*+1}\in(\omega,\omega+\epsilon)\}}\right]\\
&<\frac{(\varphi_{\omega}^+-\varphi_{\omega}^-)(\gamma_*+1)}{\tau_*}\epsilon,
\end{aligned}
\end{equation}
and hence 
\begin{equation}
\lim_{\epsilon\to 0}\mathbb{E}\left[\underline{\varphi}_{\epsilon}(\hat{X})\right]=\mathbb{E}\left[{\varphi}_{0}(\hat{X})\right].\end{equation}
Similarly, we can show that 
\begin{equation}
\lim_{\epsilon\to 0}\mathbb{E}\left[\bar{\varphi}_{\epsilon}(\hat{X})\right]=\mathbb{E}\left[{\varphi}_{0}(\hat{X})\right].
\end{equation}
Using 
$\underline{\varphi}_{\epsilon}(x)\leq \varphi_0(x)\leq \bar{\varphi}_{\epsilon}(x)$,
 we obtain 
 \begin{equation}
 \begin{aligned}
\liminf_{N\to\infty}\frac{1}{N}\sum_{i=1}^N\varphi_0(\hat{x}_i)&\geq \lim_{\epsilon\to 0}\lim_{N\to\infty}\frac{1}{N}\sum_{i=1}^N\underline{\varphi}_{\epsilon}(\hat{x}_i)=\mathbb{E}\left[{\varphi}_{0}(\hat{X})\right]
\end{aligned}
 \end{equation}
\begin{equation}
 \begin{aligned}
\limsup_{N\to\infty}\frac{1}{N}\sum_{i=1}^N\varphi_0(\hat{x}_i)&\leq \lim_{\epsilon\to 0}\lim_{N\to\infty}\frac{1}{N}\sum_{i=1}^N\bar{\varphi}_{\epsilon}(\hat{x}_i)=\mathbb{E}\left[{\varphi}_{0}(\hat{X})\right].
\end{aligned}
\end{equation}
Combining the above gives \eqref{converge:phi0} and completes the proof.

\subsection{Proof of Proposition \ref{the1}}\label{proof:1B}
Now we are ready to prove Proposition \ref{the1}.
For any $\epsilon>0$, let $q_{\epsilon}(x)$ be a smoothing of $q(x)$, whose rigorous definition is given in Definition \ref{def:qeps} in Appendix \ref{app:smoothing}.  The following inequality holds:
\begin{equation}
\begin{aligned}
&\left|\frac{1}{K}\sum_{k=1}^K\psi\left( \h_k^Tq(\hat{\x}),s_k\right)-\frac{1}{K}\sum_{k=1}^K\psi\left(\h_k^Tq({\x}^*),s_k\right)\right|\\
\leq\,&\left|\frac{1}{K}\sum_{k=1}^K\psi\left(\h_k^Tq(\hat{\x}),s_k\right)-\frac{1}{K}\sum_{k=1}^K\psi\left(\h_k^Tq_{\epsilon}(\hat{\x}),s_k\right)\right|\\
&+\left|\frac{1}{K}\sum_{k=1}^K\psi\left( \h_k^Tq_{\epsilon}(\hat{\x}),s_k\right)-\frac{1}{K}\sum_{k=1}^K\psi\left(\h_k^Tq_{\epsilon}({\x}^*),s_k\right)\right|\\
&+\left|\frac{1}{K}\sum_{k=1}^K\psi\left( \h_k^Tq({\x}^*),s_k\right)-\frac{1}{K}\sum_{k=1}^K\psi\left(\h_k^Tq_{\epsilon}({\x}^*),s_k\right)\right|\\
:=&T_1+T_2+T_3.
\end{aligned}
\end{equation}
Using the pseudo-Lipschitz continuity of $\psi$, $T_1$  can be bounded as
\begin{equation}\label{boundT1}
\begin{aligned}
T_1\leq \,&\frac{L_\psi}{K}\sum_{k=1}^K\bigg(1+\left\|[\h_k^Tq(\hat{\x}),s_k]\right\|+\big\|[\h_k^T q_\epsilon(\hat{\x}),s_k]\big\|\bigg)|\h_k^T(q(\hat{\x})-q_{\epsilon}(\hat{\x}))|\\
\overset{(a)}{\leq}\,&\frac{L_\psi}{K}\sum_{k=1}^K(C_0+|\h_k^Tq(\hat{\x})|+\hspace{-0.05cm}|\h_k^T q_\epsilon(\hat{\x})|)|\h_k^T(q(\hat{\x})-q_{\epsilon}(\hat{\x}))|\\
\overset{(b)}{\leq}\,&{L_{\psi}}\left(C_0\hspace{-0.05cm}+\hspace{-0.05cm}{\frac{\|\H q(\hat{\x})\|}{\sqrt{K}}}\hspace{-0.05cm}+\hspace{-0.05cm}{\frac{\|\H q_{\epsilon}(\hat{\x})\|}{\sqrt{K}}}\right)\frac{\|\bH\left(q(\hat{\x})-q_{\epsilon}(\hat{\x})\right)\|}{{\sqrt{K}}}\\\
{\leq}\,&{L_{\psi}\|\H\|}\bigg(C_0+\hspace{-0.05cm}\frac{\|\bH\|(\|q(\hat{\x})\|\hspace{-0.05cm}+\hspace{-0.05cm}\|q_{\epsilon}(\hat{\x})\|)}{\sqrt{K}}\bigg)\frac{\left\|q(\hat{\x})-q_{\epsilon}(\hat{\x})\right\|}{\sqrt{K}}\\
\xrightarrow[a.s.]{(c)} &\frac{L_{\psi}\bar{\lambda}_{\max}^{\frac{1}{2}}}{\sqrt{\delta}}\bigg(C_0+\frac{\bar{\lambda}_{\max}^{\frac{1}{2}}}{\sqrt{\delta}}\big(\mathbb{E}[q(\hat{X})^2]^{\frac{1}{2}}+\mathbb{E}[q_\epsilon(\hat{X})^2]^{\frac{1}{2}}\big) \bigg){\mathbb{E}\left[\left(q(\hat{X})-q_\epsilon(\hat{X})\right)^2\right]^{\frac{1}{2}}},
\end{aligned}
\end{equation}
where $C_0=1+2B_s$,  (a) uses $\|\x\|\leq \|\x\|_1$ and $|s_k|\leq B_s$, (b) applies the Cauchy-Schwarz inequality,  and (c) applies Lemma \ref{lem:eigenvalue} and Theorem \ref{mainresult:ed} with $\varphi(x)=q(x)$, $\varphi(x)=q_\epsilon(x)$, and $\varphi(x)=(q(x)-q_\epsilon(x))^2$ (note that Theorem \ref{mainresult:ed} holds  for the discontinuous function $q(x)$ according to Appendix \ref{app:nonsmooth}). 
By further letting $\epsilon\to0$ in \eqref{boundT1} and applying the property of the smoothing function $q_\epsilon$ in Lemma \ref{lem:difference}  (see further ahead), we get 
\begin{equation}
\lim_{\epsilon\to 0}\lim_{N\to 0}T_1\overset{a.s.}{=}0.
\end{equation}
Using a similar argument as \eqref{boundT1}, we can show that 
\begin{equation}
\lim_{\epsilon\to 0}\lim_{N\to 0}T_3\overset{a.s.}{=}0.
\end{equation}
In addition, $T_2$ can be bounded as 
\begin{equation}\label{boundT2}
\begin{aligned}
&T_2\hspace{-0.05cm}\leq\hspace{-0.05cm} {L_{\psi}\|\H\|}\hspace{-0.05cm}\bigg(\hspace{-0.05cm}C_0\hspace{-0.05cm}+\hspace{-0.1cm}\frac{\|\bH\|(\|q_{\epsilon}(\hat{\x})\|\hspace{-0.05cm}+\hspace{-0.05cm}\|q_{\epsilon}({\x}^*)\|)}{\sqrt{K}}\hspace{-0.05cm}\bigg)\hspace{-0.05cm}\frac{\left\|q_\epsilon\hspace{-0.02cm}(\hat{\x})\hspace{-0.05cm}-\hspace{-0.05cm}q_{\epsilon}\hspace{-0.02cm}({\x}^*)\right\|}{\sqrt{K}}.
\end{aligned}
\end{equation}
Note that  $\{\hat{x}_i\}_{1\leq i\leq N}$ and $\{x_i^*\}_{1\leq i\leq N}$ are bounded. Hence, according to Lemma  \ref{lemma:boundderi} (see further ahead), there exists a constant $C$ such that 
\begin{equation}
\begin{aligned}
T_2\leq &\frac{{L_{\psi}C\|\H\|}}{\epsilon}\hspace{-0.05cm}\bigg(\hspace{-0.05cm}C_0+\hspace{-0.1cm}\frac{\|\bH\|(\|q_{\epsilon}(\hat{\x})\|\hspace{-0.05cm}+\hspace{-0.05cm}q_{\epsilon}({\x}^*)\|)}{\sqrt{K}}\hspace{-0.05cm}\bigg) 
\frac{\|\hat{\x}-\x^*\|}{\sqrt{K}}
\xrightarrow{a.s.}\,0,  ~~\text{ as } K\to\infty,
\end{aligned}
\end{equation}
where the convergence is due to Lemma \ref{converge:a} and  Lemma \ref{lem:eigenvalue}. Combining the above, we get  the desired result.
\section{Smoothing of $q(x)$}\label{app:smoothing}
  In this appendix, we introduce the smoothing technique applied to $q(\cdot)$ and present auxiliary results that characterize the properties of the smoothing function. We assume throughout this appendix that $q(\cdot)$ is continuously differentiable except at a finite number of points $x_1,x_2,\dots,x_M\notin\{-a,a\}$.
  
   The construction of the smoothing function is classical, which relies on the standard mollifier given as follows
\begin{equation}
 \zeta(x)=\left\{\begin{aligned}
 ce^{-1/(1-x^2)},~~&\text{if }|x|<1;\\
 0,~~~~~~~~&\text{otherwise},\end{aligned}
 \right.
\end{equation}
 where $c$ is a constant chosen such that $\int_{\R}\zeta(x)dx=1$. Given $\epsilon>0$, we define
\begin{equation}
\zeta_{\epsilon}(x)=\frac{1}{\epsilon}\zeta\left(\frac{x}{\epsilon}\right)
\end{equation}
 as  a scaling of the mollifier. Clearly, $\zeta_\epsilon(x)$ is  supported on $(-\epsilon,\epsilon)$ and is infinitely differentiable. Now we are ready to define the smoothing function $q_{\epsilon}(x)$. 
 \appendixdef
  \begin{definition}\label{def:qeps}
 For $\epsilon>0$, the $\epsilon$-mollification of  $q(x)$ is defined as the convolution $q_{\epsilon}=\zeta_{\epsilon}*q$ on $\R$, i.e.,
  \begin{equation}
  q_{\epsilon}(x)=\int_{\R}\zeta_{\epsilon}(x-y)q(y)dy.
  \end{equation}
  \end{definition}  
  The $\epsilon$-mollification $q_{\epsilon}(x)$ is infinitely differentiable since $\zeta_{\epsilon}(x)$ is infinitely differentiable. In addition, $q_{\epsilon}$ satisfies the following properties.
  \begin{lemma}\label{lemma:boundderi}
 Let $\mathcal{C}\subset\R$ be a compact set.  There exists a constant $C$ (independent of $\epsilon$) such the derivative of $q_\epsilon(x)$ satisfies 
                 \begin{equation}
 |q_\epsilon'(x)|\leq \frac{C}{\epsilon},~~\forall\, x\in\mathcal{C}.
 \end{equation}

  \end{lemma}
     \begin{proof}
     First, it is easy to check that there exists a constant $C_{\zeta'}>0$ such that 
     $|\zeta'(x)|\leq C_{\zeta'}$ for all $x\in\R.$
     It follows that 
                \begin{equation}
|\zeta_\epsilon'(x)|\leq \frac{C_{\zeta'}}{\epsilon^2},~\forall\,x\in\R.
\end{equation}
          In addition, by the  assumption on $q(x)$, it is bounded on $\mathcal{C}$.  Let $C_q$ denote the bound of $|q(x)|$ on $\mathcal{C}$,  we have the following inequality:
                 \begin{equation}
           \begin{aligned}
           |q_{\epsilon}'(x)|&\leq \int_{|y-x|\leq \epsilon}|\zeta_{\epsilon}'(x-y)q(y)|dy\leq \frac{C_{\zeta'}C_{q}}{\epsilon},~~\forall~x\in\mathcal{C}.
           \end{aligned}
           \end{equation}
           Therefore, Lemma \ref{lemma:boundderi} holds with $C:=C_{\zeta'}C_q$.

     \end{proof}
   \begin{lemma}\label{lemma:lip}
   The function   $q_{\epsilon}\circ\eta_a(x;\gamma_a)$ is Lipschitz continuous, where $q_{\epsilon}$ is  given in Definition \ref{def:qeps} and $\eta_a(x;\gamma_a)$ is defined in \eqref{def:etaa}.
   \end{lemma}
   \begin{proof}
   Given $x_1, x_2$, we have 
        \begin{equation}
\begin{aligned}
   \left|q_{\epsilon}\circ\eta_a(x_1;\gamma_a)-q_{\epsilon}\circ\eta_a(x_2;\gamma_a)\right|
  &\leq \sup_{t\in[-a,a]}|q_{\epsilon}'(t)||\eta_a(x_1;\gamma_a)-\eta_a(x_2;\gamma_a)|\\
   &\leq \frac{\sup_{t\in[-a,a]}|q_{\epsilon}'(t)|}{\gamma_a+1}|x_1-x_2|.
   \end{aligned}
         \end{equation}
By Lemma \ref{lemma:boundderi}, $q_{\epsilon}\circ\eta_a(x;\gamma_a)$ is Lipschitz continuous.
   \end{proof}
    \begin{lemma}\label{lem:difference}
   Given $a>0$, let $X_a$ be a random variable defined in \eqref{def:Xa}. Then we have 
         \begin{equation}
\lim_{\epsilon\to 0}\mathbb{E}\left[(q_{\epsilon}(X_a)-q(X_a))^2\right]=0.
         \end{equation}
        \end{lemma}
        \begin{proof}
         The proof follows a similar idea as \cite[Lemma 4]{Universality}. Denote the discontinuous points of $q$ by  $\{x_1,x_2,\dots, x_M\}$, and let $\mathcal{D}_{\epsilon}:=\cup_{i=1}^M[x_i-\epsilon,x_i+\epsilon]$. Using the definition of $q_\epsilon$, we have 
      \begin{equation}
      \begin{aligned}
      &\mathbb{E}\left[(q_{\epsilon}(X_a)-q(X_a))^2\right]
      =\int_{x\in[-a,a]}\hspace{-0.1cm}\bigg(\int_{|y-x|\leq \epsilon}(q(x)-q(y))\zeta_{\epsilon}(x-y)dy\bigg)^2p_{X_a}(x)dx,
      \end{aligned}
       \end{equation}
      where we have used the fact that $\int_{|y-x|\leq \epsilon}\zeta_{\epsilon}(x-y)dy=1$ to express $q(x)$ as $q(x)=\int_{|y-x|\leq \epsilon}q(x)\zeta_\epsilon(x-y)dy$, and $p_{X_a}(x)$ denotes the PDF of the random variable $X_a$. 
        We next analyze $\mathbb{E}\left[(q_{\epsilon}(X_a)-q(X_a))^2\right]$ by splitting the above integral into two parts as follows:
      \begin{equation}
 \begin{aligned}
 \mathbb{E}\left[(q_{\epsilon}(X_a)-q(X_a))^2\right]
 =&\hspace{-0.15cm}\int_{x\in[-a,a]\backslash \mathcal{D}_{2\epsilon}}\hspace{-0.15cm}\bigg(\int_{|y-x|\leq \epsilon}\hspace{-0.1cm}(q(x)\hspace{-0.05cm}-\hspace{-0.05cm}q(y))\zeta_{\epsilon}(x\hspace{-0.05cm}-\hspace{-0.05cm}y)dy\bigg)^{\hspace{-0.1cm}2}p_{X_a}(x)dx\\
 &\hspace{-0.25cm}+\hspace{-0.15cm}\int_{x\in[-a,a]\cap \mathcal{D}_{2\epsilon}}\hspace{-0.2cm}\bigg(\int_{|y-x|\leq \epsilon}\hspace{-0.1cm}(q(x)\hspace{-0.05cm}-\hspace{-0.05cm}q(y))\zeta_{\epsilon}(x\hspace{-0.05cm}-\hspace{-0.05cm}y)dy\bigg)^{\hspace{-0.1cm}2}\hspace{-0.03cm}p_{X_a}(x)dx\\
:=&\,T_1+T_2.
 \end{aligned}
        \end{equation}
 Since the function $q(\cdot)$ is continuously differentiable on $[-a,a]\backslash\{x_1,x_2,\dots, x_M\}$, the first term can be bounded as
\begin{equation}
 \begin{aligned}
 T_1\hspace{-0.1cm}&\overset{(a)}{=} \hspace{-0.2cm}\int_{\hspace{-0.02cm}x\in[-a,a]\backslash \mathcal{D}_{2\epsilon}}\hspace{-0.2cm}\bigg(\hspace{-0.1cm}\int_{|y-x|\leq \epsilon}\hspace{-0.14cm}q'\hspace{-0.05cm}(t_{xy})(x\hspace{-0.05cm}-\hspace{-0.05cm}y)\zeta_{\epsilon}(x\hspace{-0.05cm}-\hspace{-0.05cm}y)dy\hspace{-0.05cm}\bigg)^{\hspace{-0.1cm}2}\hspace{-0.05cm}p_{X_a}\hspace{-0.05cm}(x)dx\\
 &\overset{(b)}{\leq} \left(\sup_{t\in[-a,a]\backslash\mathcal{D}_{\epsilon}} |q'(t)|\right)^2\epsilon^2 \int_{x\in[-a,a]\backslash \mathcal{D}_{2\epsilon}}p_{X_a}(x)dx\\
 &\leq  \left(\sup_{t\in[-a,a]\backslash\mathcal{D}_{\epsilon}} |q'(t)|\right)^2\epsilon^2, \end{aligned}
\end{equation}
 where in (a), $t_{xy}=(1-\mu)x+\mu y$ for some $\mu\in[0,1]$, and (b) follows from the fact that $t_{xy}\notin \mathcal{D}_{\epsilon}$ as $x\notin \mathcal{D}_{2\epsilon}$ and $|y-x|\leq \epsilon$. By 
 Assumption \ref{ass} (iv), we obtain  $\lim_{\epsilon\to0}T_1= 0$. 
  
 Next we analyze $T_2$. Since $x_1,x_2,\dots,x_M\notin\{-a,a\}$, we have $\pm a\notin \mathcal{D}_{2\epsilon}$ for sufficiently small $\epsilon$.   Hence, the second term $T_2$ can be bounded as follows
\begin{equation}
 \begin{aligned}
 T_2&\leq 4\left(\sup_{t\in[-a-\epsilon,a+\epsilon]}|q(t)|\right)^2\int_{x\in[-a,a]\cap \mathcal{D}_{2\epsilon}}p_{X_a}(x)dx\\
 &\leq4\left(\sup_{t\in[-a-\epsilon,a+\epsilon]}|q(t)|\right)^2\int_{\frac{\tau_a Z}{\gamma_a+1}\in(-a,a)\cap\mathcal{D}_{2\epsilon}}p_Z(z)dz\\
 &\leq 4\left(\sup_{t\in[-a-\epsilon,a+\epsilon]}|q(t)|\right)^2\frac{4L(\gamma_a+1)}{\tau_a}\epsilon,
 \end{aligned}
 \end{equation}
 where the second inequality holds  since $X_a\in[-a,a]\cap\mathcal{D}_{2\epsilon}$ implies  that $X_a=\frac{\tau_aZ}{\gamma_a+1}\in(-a,a)\cap\mathcal{D}_{2\epsilon}$ as  $a\not\in\mathcal{D}_{2\epsilon}$. By the assumption on $q(\cdot)$, it is bounded on any compact set, and thus $\lim_{\epsilon\to 0}T_2=0$.
    \end{proof}

\section{Proof of Proposition \ref{pro:3} and Proposition \ref{pro:convergexa}}\label{app:AMP2}
\subsection{Proof of Proposition \ref{pro:3}}\label{app:pro3}
To prove Proposition~\ref{pro:3}, we begin with an auxiliary lemma that will be used in the proof.
{
\begin{lemma}\label{tilde_property}
Let $(\tilde{\x}_{t+1},\tilde{\z}_{t+1})$ be the outcome of the post-processing steps applied to $(\x_t,\z_t)$ as defined in \eqref{AMP2}, and let  $\tilde{\bb}_{t+1}=\s-\tilde{\bz}_{t+1}$ and $\bb_t=\s-\bz_t$. The following results hold under Assumption \ref{ass}:
\begin{itemize}
\item[(i)]
$\displaystyle\lim_{K\to\infty}\langle \bb_t,\tilde{\bb}_{t+1}\rangle=\frac{1}{\delta}\lim_{N\to\infty}\langle \x_t,\tilde{\x}_{t+1}\rangle$.\item[(ii)] For any pseudo-Lipschitz function $\varphi:\R^{3}\rightarrow\R$,
\begin{equation}
\hspace{-0.05cm}\lim_{K\to\infty}\hspace{-0.05cm}\frac{1}{K}\hspace{-0.1cm}\sum_{i=1}^K\varphi(b_{t,k},\tilde{b}_{t+1,k},s_k)\hspace{-0.08cm}=\hspace{-0.05cm}\mathbb{E}[\varphi(\sigma_tZ_t,\tilde{\sigma}_{t+1}\tilde{Z}_{t+1},S)],
\end{equation}
where $(Z_{t},\tilde{Z}_{t+1})$ are jointly Gaussian distributed and are independent of $S\sim\text{\normalfont Unif}(\mathcal{S})$ with $Z_{t}\sim\mathcal{N}(0,1)$ and $\tilde{Z}_{t+1}\sim\mathcal{N}(0,1)$, $\sigma_t$ is defined in Proposition \ref{pro:4}, and $\tilde{\sigma}_{t+1}^2=\frac{1}{\delta}\mathbb{E}\left[(q_\epsilon\circ\eta_a)^2(\tau_tZ;\gamma_a)\right]$. \end{itemize}
\end{lemma}
\begin{proof}
Note that $(\tilde{\x}_{t+1},\tilde{\bz}_{t+1})$  are obtained by performing another AMP iteration after obtaining $(\x_t,\bz_t)$ by \eqref{AMP_copy}, with $\eta_t(x)=q_\epsilon\circ\eta_a(x,\gamma_a)$. Hence, the desired results in Lemma \ref{tilde_property} follow immediately from Proposition \ref{AMP_property}.
\end{proof}}

Now we are ready to prove Proposition \ref{pro:3}. First, 
 \begin{equation}
 \begin{aligned}
 \lim_{N\to\infty}\tilde{\alpha}_t&=\frac{1}{\delta}\lim_{N\to\infty}\left<(q_\epsilon\circ\eta_a)'({\x}_{t}+\H^T{\z}_{t};\gamma_a)\right>\\
 &\overset{(a)}{=}\frac{1}{\delta}\lim_{N\to\infty}\left<(q_\epsilon\circ\eta_a)'(-{\br}_{t+1};\gamma_a)\right>\\
 &\overset{(b)}{\underset{a.s.}{=}}\frac{1}{\delta}\mathbb{E}\left[(q_\epsilon\circ\eta_a)'(\tau_tZ;\gamma_a)\right],
 \end{aligned}
 \end{equation}
 where (a) uses the definition in \eqref{AMP_relation}, and (b) is obtained by  applying Proposition \ref{pro:4} (i). Since $(q_\epsilon\circ\eta_{a})'(x;\gamma_a)=q_{\epsilon}'(\eta_a(x;\gamma_a))\eta_a'(x;\gamma_a)$ is bounded for given $\epsilon>0$ and $\lim_{t\to\infty}\tau_{t}^2\overset{a.s.}=\tau_a^2$,  we can apply the dominated  convergence theorem to obtain 
 \begin{equation}\label{tildealpha_converge}
 \lim_{t\to\infty}\lim_{N\to\infty}\tilde{\alpha}_t\overset{a.s.}{=}\frac{1}{\delta}\lim_{t\to\infty}\mathbb{E}\left[(q_\epsilon\circ\eta_a)'(\tau_tZ;\gamma_a)\right]=\tilde{\alpha}_*.
 \end{equation}
Note that 
\begin{equation}\label{pro3:1}
\begin{aligned}
&\bigg|\frac{1}{K}\sum_{k=1}^K\psi\left(\tilde{\alpha}_{t}s_k+\tilde{\alpha}_t{b}_{t,k}-\tilde{b}_{t+1,k},s_k\right)-\frac{1}{K}\sum_{k=1}^K\psi\left(\tilde{\alpha}_*s_k+\tilde{\alpha}_* {b}_{t,k}-\tilde{b}_{t+1,k},s_k\right)\bigg|\\
\overset{(a)}{\leq}\,&\frac{L_\psi}{K}\sum_{k=1}^K\bigg(1+\|[\tilde{\alpha}_{t}s_k+\tilde{\alpha}_t{b}_{t,k}-\tilde{b}_{t+1,k},s_k]\|+\|[\tilde{\alpha}_*s_k+\tilde{\alpha}_*{b}_{t,k}-\tilde{b}_{t+1,k},s_k]\|\bigg)|s_k+{b}_{t,k}||\tilde{\alpha}_t-\tilde{\alpha}_*|\\
\overset{(b)}{\leq}\,&\frac{L_\psi}{K}\sum_{k=1}^K\bigg(1+2B_s+(|\tilde{\alpha}_t|+|\tilde{\alpha}_*|)(B_s+|{b}_{t,k}|)+2|\tilde{b}_{t+1,k}|\bigg)(B_s+|{b}_{t,k}|)|\tilde{\alpha}_t-\tilde{\alpha}_*|\\
\overset{(c)}{\leq} \,&{L_\psi}|\tilde{\alpha}_t-\tilde{\alpha}_*|\bigg(1+2B_s+(|\tilde{\alpha}_t|+|\tilde{\alpha}_*|)\big(B_s+\frac{\|{\bb}_t\|}{\sqrt{K}}\big)+\frac{2\|\tilde{\bb}_{t+1}\|}{\sqrt{K}}\bigg)\bigg(B_s+\frac{\|{\bb}_{t}\|}{\sqrt{K}}\bigg)\\
\overset{(d)}{\underset{a.s.}{\longrightarrow}}\,&0~~(\text{as }t\to\infty, K\to\infty),
\end{aligned}
\end{equation}
where (a) follows from the pseudo-Lipschitz continuity of $\psi(\cdot,\cdot)$,  (b) uses the fact that $|s_k|\leq B_s$, and $\|\x\|\leq \|\x\|_1$, (c) applies the Cauchy-Schwarz inequality, and (d) is due to \eqref{tildealpha_converge} and the boundedness of $\|\bb_t\|/\sqrt{K}$ and $\|\tilde{\bb}_{t+1}\|/\sqrt{K}$, which we show next. Specifically, by  Lemma \ref{tilde_property}, we have
\begin{equation}\label{bt}
\begin{aligned}
\lim_{K\to\infty}\frac{\|{\bb}_{t}\|^2}{K}&\overset{a.s.}{=}{\sigma}_{t}^2=\frac{1}{\delta}\mathbb{E}\left[\left(\eta_a(\tau_{t-1}Z;\gamma_a)\right)^2\right],\\
\lim_{K\to\infty}\frac{\|\tilde{\bb}_{t+1}\|^2}{K}&\overset{a.s.}{=}\tilde{\sigma}_{t+1}^2=\frac{1}{\delta}\mathbb{E}\left[\left(q_\epsilon\circ\eta_a({\tau}_{t}Z;\gamma_a)\right)^2\right].
\end{aligned}
\end{equation}
Since $\eta_a(\cdot;\gamma_a)$ and $q_{\epsilon}\circ\eta_a(\cdot;\gamma_a)$ are both bounded, we can apply the dominated convergence theorem to obtain
   \begin{equation}\label{converge:sigma}
   \begin{aligned}
   &\lim_{t\to\infty}{\sigma}_{t}^2=\frac{1}{\delta}\mathbb{E}\left[(\eta_a(\tau_a Z;\gamma_a))^2\right]=\frac{1}{\delta}\mathbb{E}\left[X_a^2\right],\\
   &\lim_{t\to\infty}\tilde{\sigma}_{t+1}^2=\frac{1}{\delta}\mathbb{E}\left[(q_{\epsilon}\circ\eta_a(\tau_a Z;\gamma_a))^2\right]=\frac{1}{\delta}\mathbb{E}\left[(q_{\epsilon}(X_a))^2\right],
   \end{aligned}
\end{equation}
where  we have used the fact that $\lim_{t\to\infty}\tau_t^2=\tau_a^2$ and the definition of $X_a$ in \eqref{def:Xa}. The boundedness of ${\|{\bb}_{t}\|}/{\sqrt{K}}$ and $\|\tilde{\bb}_{t+1}\|/\sqrt{K}$ then follows from \eqref{bt} and \eqref{converge:sigma}.

In the following,  we will  focus on  
\begin{equation}\label{def:phi}
\begin{aligned}
\lim_{t\to\infty}\lim_{K\to\infty}\frac{1}{K}\sum_{k=1}^K\psi\left(\tilde{\alpha}_* s_k+\tilde{\alpha}_*\,{b}_{t,k}-\tilde{b}_{t+1,k},s_k\right).
\end{aligned}\end{equation}
Applying Lemma  \ref{tilde_property} (ii) with $\varphi(b,\tilde{b},s)=\psi(\tilde{\alpha}_*s+\tilde{\alpha}_*b-\tilde{b},s)$, we have 
\begin{equation}\label{pro3:2}
\begin{aligned}
\lim_{K\to\infty}\frac{1}{K}\sum_{k=1}^K\psi\left(\tilde{\alpha}_* s_k+\tilde{\alpha}_*\,{b}_{t,k}-\tilde{b}_{t+1,k},s_k\right)
&\overset{a.s.}{=}\, \mathbb{E}\left[\psi(\tilde{\alpha}_*S+\tilde{\alpha}_*{\sigma}_{t}Z_{t}-\tilde{\sigma}_{t+1}\tilde{Z}_{t+1},S)\right].
\end{aligned}
\end{equation}
   Next, we characterize the covariance between $Z_{t}$ and $\tilde{Z}_{t+1}$. On the one hand, 
 \begin{equation}\label{bbT1}
 \begin{aligned}
\lim_{t\to\infty}\lim_{K\to\infty}\left<\tilde{\bb}_{t+1},{\bb}_{t}\right>
&\overset{(a)}{\underset{a.s.}{=}}\lim_{t\to\infty}\tilde{\sigma}_{t+1}{\sigma}_t\mathbb{E}\left[\tilde{Z}_{t+1}Z_t\right]\\
 &\overset{(b)}{=}\frac{1}{\delta}\mathbb{E}\left[X_a^2\right]^{\frac{1}{2}}\mathbb{E}\left[(q_{\epsilon}(X_a))^2\right]^{\frac{1}{2}}\lim_{t\to\infty}\mathbb{E}[\tilde{Z}_{t+1}Z_t],
 \end{aligned}
 \end{equation}
 where (a) applies Lemma  \ref{tilde_property} (ii), and (b) is due to \eqref{converge:sigma}.  On the other hand, 
 \begin{equation}\label{bbT}
 \begin{aligned}
\lim_{t\to\infty}\lim_{K\to\infty}\left<\tilde{\bb}_{t+1},{\bb}_{t}\right>
&\overset{(a)}{\underset{a.s.}{=}}\frac{1}{\delta}\lim_{t\to\infty}\lim_{N\to\infty}\left<\tilde{\x}_{t+1},{\x}_{t}\right>\\
 &\overset{(b)}{=}\frac{1}{\delta}\lim_{t\to\infty}\lim_{N\to\infty}\left<q_\epsilon\circ\eta_a(-{\br}_{t+1};\gamma_a),\eta_{a}(-{\br}_{t};\gamma_a)\right>\\
 &   =\frac{1}{\delta}\lim_{t\to\infty}\lim_{N\to\infty}\left<q_\epsilon\circ\eta_a(-{\br}_{t+1};\gamma_a),\eta_{a}(-{\br}_{t+1};\gamma_a)\right>\\
    &~~~+\frac{1}{\delta}\lim_{t\to\infty}\lim_{N\to\infty}\left<q_\epsilon\circ\eta_a(-{\br}_{t+1};\gamma_a),
\eta_{a}(-{\br}_{t};\gamma_a)-\eta_{a}(-{\br}_{t+1};\gamma_a)\right>,
    \end{aligned}
    \end{equation}
     where (a) applies Lemma \ref{tilde_property} (i),  and (b) uses the formula of $\tilde{\x}_{t+1}$ in \eqref{AMP2} and the definition of $\br_t$ in \eqref{AMP_relation}. For the first term in \eqref{bbT}, we have
\begin{equation}
    \begin{aligned}
    &\frac{1}{\delta}\lim_{t\to\infty}\lim_{N\to\infty}\left<q_\epsilon\circ\eta_a(-{\br}_{t+1};\gamma_a),\eta_{a}(-{\br}_{t+1};\gamma_a)\right>\\
    &\overset{(a)}{=}\frac{1}{\delta}\lim_{t\to\infty}\mathbb{E}\left[\eta_{a}(\tau_tZ;\gamma_a)\,q_\epsilon\circ\eta_a(\tau_tZ;\gamma_a)\right]\\
   & \overset{(b)}{=}\frac{1}{\delta}\mathbb{E}\left[X_a\,q_{\epsilon}(X_a)\right],
    \end{aligned}
    \end{equation}
    where (a) applies Proposition \ref{pro:4} (i), and (b) follows from  the dominated convergence theorem. 
Regarding the second term in \eqref{bbT}, the following result holds:
\begin{equation}
\begin{aligned}
\left|\left<q_\epsilon\circ\eta_a(-{\br}_{t+1};\gamma_a),\eta_{a}(-{\br}_{t};\gamma_a)-\eta_{a}(-{\br}_{t+1};\gamma_a)\right>\right|
&\leq \,\frac{1}{N}\|q_\epsilon\circ\eta_a(-{\br}_{t+1};\gamma_a)\|\|\eta_{a}(-{\br}_{t};\gamma_a)-\eta_{a}(-{\br}_{t+1};\gamma_a)\|\\
&\overset{(a)}{\leq}\, \frac{\sup_{x\in[-a,a]}|q_\epsilon(x)|}{\gamma_a+1}\frac{\|{\br}_t-{\br}_{t+1}\|}{\sqrt{N}}\\
&\overset{(b)}{\leq}\, \frac{\sup_{x\in[-a,a]}|q_\epsilon(x)|}{\gamma_a+1}\frac{\|\bH\|\|\bz_{t}-\bz_{t-1}\|+\|\x_{t}-\x_{t-1}\|}{\sqrt{N}}\\
&\overset{(c)}{\underset{a.s.}{\longrightarrow}}0,~~~\text{as }t, N\to\infty,
\end{aligned}
\end{equation}
where (a) holds since $\eta_a(\cdot;\gamma_a)$ is Lipschitz continuous with constant $(\gamma_a+1)^{-1}$, (b) follows from \eqref{AMP_relation}, and (c) applies  Lemma \ref{difference}.
 Combining the above, we get
  \begin{equation}
\lim_{t\to\infty}\lim_{K\to\infty}\left<\tilde{\bb}_{t+1},{\bb}_{t}\right>=\frac{1}{\delta}\mathbb{E}\left[X_a\,q_{\epsilon}(X_a)\right].
  \end{equation}
 This, together with \eqref{bbT1}, gives
 \begin{equation}
 \frac{1}{\delta}\mathbb{E}[X_a^2]^{\frac{1}{2}}\mathbb{E}[(q_{\epsilon}(X_a))^2]^{\frac{1}{2}}\lim_{t\to\infty}\mathbb{E}[\tilde{Z}_{t+1}Z_t]=\frac{1}{\delta}\mathbb{E}[X_aq_{\epsilon}(X_a)].
 \end{equation}
It follows that 
 \begin{equation}
({\sigma}_{t}Z_{t}, \tilde{\sigma}_{t+1}\tilde{Z}_{t+1},S)\xrightarrow[t\to\infty]{d}({Z}_{1}, {Z}_2,S),
\end{equation}  where  $({Z}_1,{Z}_2)$ are jointly Gaussian distributed with zero-mean and covariance matrix
\begin{equation}
\mathbf{C}_{{Z}_1{Z}_2}=\left(\begin{matrix}\frac{1}{\delta}\mathbb{E}\left[X_a^2\right]&\frac{1}{\delta}\mathbb{E}\left[X_aq_{\epsilon}(X_a)\right]\\\frac{1}{\delta}\mathbb{E}\left[X_aq_{\epsilon}(X_a)\right]&\frac{1}{\delta}\mathbb{E}\left[(q_{\epsilon}(X_a))^2\right]\end{matrix}\right).
\end{equation}
Define 
 \begin{equation}
X_t:=\tilde{\alpha}_*\tilde{\sigma}_{t+1}Z_{t+1}-{\sigma}_{t}Z_{t}\sim\mathcal{N}(0,\tilde{\beta}_t)
\end{equation}
 and 
\begin{equation}
X:=\tilde{\alpha}_*{Z}_1-{Z}_2\sim\mathcal{N}(0,\tilde{\beta}_*),
\end{equation}
 where $\tilde{\beta}_t=\mathbb{E}[(\tilde{\alpha}_*\tilde{\sigma}_{t+1}\tilde{Z}_{t+1}-{\sigma}_{t}Z_{t})^2]$ and  $\tilde{\beta}_*$ is defined in Proposition \ref{pro:3}. 
 It follows from the above discussions that $\lim_{t\to\infty}\tilde{\beta}_t=\tilde{\beta}_*$, and thus $X_t\xrightarrow{d} X$. The remaining task is to prove that 
 \begin{equation}\label{pro3:3}
 \lim_{t\to\infty}\mathbb{E}[\psi(\tilde{\alpha}_*S+X_t,S)]=\mathbb{E}[\psi(\tilde{\alpha}_*S+X,S)],
 \end{equation} which, together with \eqref{pro3:1} and \eqref{pro3:2}, gives the desired result. 
 
 We next prove \eqref{pro3:3}.  Since $X_t$ and $X$ are independent of $S$, 
 it suffices to prove $\lim_{t\to\infty}\mathbb{E}[\psi(\tilde{\alpha}_*s+X_t,s)]=\mathbb{E}[\psi(\tilde{\alpha}_*s+X,s)]$ for $s=1$ and $s=-1$. Without loss of generality, we prove for the case $s=1$. 
Let $f(x):=\psi(\tilde{\alpha}_*+x;1)$. 
Our goal is to show that 
\begin{equation}\label{XTtoX}
\lim_{t\to\infty}\mathbb{E}[f(X_t)]=\mathbb{E}[f(X)].
\end{equation} 
Given $M>0$, define $f_M(x)=f(x)\mathbf{1}_{[-M,M]}(x)$. Since $\psi(x,s)$ is pseudo-Lipschitz, $f(x)$ is pseudo-Lipschitz, and hence $f_M(x)$ is continuous and bounded. Therefore, it follows from $X_t\xrightarrow{d}X$ that 
\begin{equation}
\lim_{t\to\infty}\mathbb{E}[f_M(X_t)]=\mathbb{E}[f_M(X)].
\end{equation}We next investigate $\mathbb{E}[|f(X_t)-f_M(X_t)|]$ and $\mathbb{E}[|f(X)-f_M(X)|]$.  Since $f(x)$ is pseudo-Lipschitz, there exists a constant $C$ such that 
\begin{equation}
f(x)\leq C(1+x^2).
\end{equation}
Hence, 
\begin{equation}
\begin{aligned}
\mathbb{E}[|f(X_t)-f_M(X_t)|]&=\mathbb{E}[|f(X_t)|\mathbf{1}_{|X_t|\geq M}]\\
&\leq C\,\mathbb{P}(|X_t|\geq M)+C\,\mathbb{E}[X_t^2\mathbf{1}_{|X_t|\geq M}]\\
&\leq C\frac{\mathbb{E}[|X_t|]}{M}+C\frac{\mathbb{E}[|X_t|^3]}{M},
\end{aligned}
\end{equation}
where the last inequality uses the Markov inequality and the fact that $X_t^2\mathbf{1}_{|X_t|\geq M}\leq\frac{|X_t|^3}{M}$. A similar result can be established for $\mathbb{E}[|f(X)-f_M(X)|]$. 
Clearly, $\mathbb{E}[|X_t|]\rightarrow\mathbb{E}[|X|]$ and $\mathbb{E}[|X_t|^3]\rightarrow\mathbb{E}[|X|^3]$ as $\tilde{\beta}_t\to\tilde{\beta}_*$.
By further noting the following inequality
\begin{equation}
\begin{aligned}
|\mathbb{E}[f(X_t)]-\mathbb{E}[f(X)]|&\leq \mathbb{E}[|f(X)-f_M(X)|] + \left|\mathbb{E}[f_M(X_t)]-\mathbb{E}[f_M(X)]\right| \\
&~~~+ \mathbb{E}[|f(X_t)-f_M(X_t)|],
\end{aligned}
\end{equation}
we get 
\begin{equation}
\lim_{t\to\infty}|\mathbb{E}[f(X_t)]-\mathbb{E}[f(X)]|\leq2C\frac{\mathbb{E}[|X|]+\mathbb{E}[|X|^3]}{M}.
\end{equation}
Since the above result holds for any $M>0$, our desired result in \eqref{XTtoX} can be obtained by letting $M\to\infty$, which completes the proof.  
  
  \subsection{Proof of Proposition \ref{pro:convergexa}}\label{app:convergexa}
Following a similar procedure to the proof of Proposition \ref{the1} in Appendix \ref{proof:1B}, we have
\begin{equation}\label{hx-hx_T_2}
\begin{aligned}
&\left|\frac{1}{K}\sum_{k=1}^K\psi({\h_k^Tq(\x_a)},s_k)-\frac{1}{K}\sum_{k=1}^K\psi(\h_k^Tq_{\epsilon}({\x}_{t+1}),s_k)\right|\\
&\leq \left|\frac{1}{K}\sum_{k=1}^K\psi({\h_k^Tq(\x_a)},s_k)-\frac{1}{K}\sum_{k=1}^K\psi(\h_k^Tq_{\epsilon}({\x}_a),s_k)\right|\\
&~~~+\left|\frac{1}{K}\sum_{k=1}^K\psi({\h_k^Tq_\epsilon(\x_a)},s_k)-\frac{1}{K}\sum_{k=1}^K\psi(\h_k^Tq_{\epsilon}({\x}_{t+1}),s_k)\right|\\
&\leq C\left(\frac{1}{\sqrt{K}}\left\|q(\x_a)-q_{\epsilon}(\x_a)\right\|+\frac{1}{\sqrt{K}\epsilon}\|\x_a-{\x}_{t+1}\|\right),~~a.s.
\end{aligned}
\end{equation}
where $C$ is a constant independent of $K$. 
Applying Proposition \ref{prob:converge} and Lemma \ref{lem:difference}, we get  
\begin{equation}
\begin{aligned}
\lim_{\epsilon\to 0}\lim_{t\to\infty}\lim_{K\to\infty}\bigg|&\frac{1}{K}\sum_{k=1}^K\psi({\h_k^Tq(\x_a)},s_k)
-\frac{1}{K}\sum_{k=1}^K\psi(\h_k^Tq_{\epsilon}({\x}_{t+1}),s_k)\bigg|\overset{a.s.}{=}0,\end{aligned}
\end{equation}
which implies 
\begin{equation}
\begin{aligned}
&\lim_{K\to\infty}\frac{1}{K}\sum_{k=1}^K\psi\left({\h_k^Tq(\x_a)},s_k\right)=\lim_{\epsilon\to 0}\lim_{t\to\infty}\lim_{K\to\infty}\frac{1}{K}\sum_{k=1}^K\psi\left(\h_k^Tq_{\epsilon}({\x}_{t+1}),s_k\right).
\end{aligned}
\end{equation}
Based on  Proposition \ref{pro:3}, the remaining task is to prove
\begin{equation}\label{converge:epsilonto0}
\lim_{\epsilon\to 0}\mathbb{E}[\psi(\tilde{\alpha}_* S+\tilde{\beta}_* Z, S)]=\mathbb{E}[\psi(\bar{\alpha}_* S+\bar{\beta}_* Z, S)].
\end{equation}  For this purpose, it suffices to prove 
\begin{equation}\label{converge:alphabeta}
\lim_{\epsilon\to0}\tilde{\alpha}_*=\bar{\alpha}_* ,~~\lim_{\epsilon\to0}\tilde{\beta}_*=\bar{\beta}_*.
\end{equation}
Then following the same proof as that of  \eqref{XTtoX}, we get \eqref{converge:epsilonto0}. 

Next we prove \eqref{converge:alphabeta}. By Stein's Lemma, $\tilde{\alpha}_*$ can be expressed as 
\begin{equation}
\tilde{\alpha}_*=\frac{1}{\delta\tau_a}\mathbb{E}[Z q_\epsilon(\eta_a(\tau_aZ;\gamma_a))].
\end{equation}
It follows that 
\begin{equation}
\begin{aligned}
|\tilde{\alpha}_*-\bar{\alpha}_*|&\leq \frac{1}{\delta\tau_a}\mathbb{E}\left[|Z||q_{\epsilon}(\eta_a(\tau_aZ;\gamma_a))-q(\eta_a(\tau_aZ;\gamma_a))|\right]\\
&\overset{(a)}{\leq}\frac{1}{\delta\tau_a}\mathbb{E}[Z^2]^{\frac{1}{2}}\mathbb{E}[(q_\epsilon(X_a)-q(X_a))^2]^{\frac{1}{2}}\\
&\overset{(b)}{\rightarrow}0~~(\text{as }\epsilon\to 0),
\end{aligned}
\end{equation}
where (a) uses the Cauchy-Schwarz inequality and the definition of $X_a$ in \eqref{def:Xa}, and (b) applies Lemma \ref{lem:difference}. This proves \eqref{converge:alphabeta} for $\tilde{\alpha}_*$. Similarly, we have 
 \begin{equation}
 \begin{aligned}
|\mathbb{E}[X_a q_\epsilon(X_a)]-\mathbb{E}[X_aq(X_a)]|
&\leq \mathbb{E}[|X_a| |q_\epsilon(X_a)-q(X_a)|]\\
&\leq\mathbb{E}[X_a^2]^{\frac{1}{2}}\mathbb{E}[(q_\epsilon(X_a)-q(X_a))^2]^{\frac{1}{2}}\rightarrow 0
\end{aligned}
\end{equation}
 and 
\begin{equation}
 \begin{aligned}
 |\mathbb{E}[(q_\epsilon(X_a))^2]-\mathbb{E}[(q(X_a))^2]|
&\leq \mathbb{E}\left[|q_\epsilon(X_a)+q(X_a)||q_\epsilon(X_a)-q(X_a)|\right]\\
&\leq \mathbb{E}[(q_\epsilon(X_a)+q(X_a))^2]^{\frac{1}{2}}\mathbb{E}[(q_\epsilon(X_a)-q(X_a))^2]^{\frac{1}{2}}\\
 &\rightarrow 0.
 \end{aligned}
 \end{equation}
 Hence, $\lim_{\epsilon\to 0}\mathbb{E}[X_a q_\epsilon(X_a)]=\mathbb{E}[X_aq(X_a)]$ and $\lim_{\epsilon\to 0}\mathbb{E}[(q_\epsilon(X_a))^2]=\mathbb{E}[(q(X_a))^2]$, which further gives $\lim_{\epsilon\to 0}\tilde{\beta}_*=\bar{\beta}_*$ and completes the proof.
\bibliographystyle{IEEEtran}
\bibliography{IEEEabrv,reference_dce}

\end{document}